\documentclass[a4paper,11pt,reqno]{amsart}
\usepackage[utf8]{inputenc}
\usepackage{mathrsfs}
\usepackage{dsfont}
\usepackage{hyperref}
\usepackage{amsmath}
\usepackage{amssymb}
\usepackage{amsthm}
\usepackage{amsfonts}
\usepackage{amstext}
\usepackage{amsopn}
\usepackage{amsxtra}
\usepackage{mathrsfs}
\usepackage{dsfont}
\usepackage{color}
\usepackage{tikz}
\usepackage{mathtools}
\usepackage{xparse}
\usepackage[capitalize]{cleveref}

\newtheorem{theorem}{Theorem}
\newtheorem{lemma}[theorem]{Lemma}
\newtheorem{proposition}[theorem]{Proposition}
\newtheorem{corollary}[theorem]{Corollary}
\newtheorem{remark}[theorem]{Remark}

\newtheorem{assumption}[theorem]{Assumption}

\newcommand{\R}{\mathbb{R}}
\newcommand{\Z}{\mathbb{Z}}

\newcommand{\N}{\mathbb{N}}
\newcommand{\bP}{\mathbb{P}}

\newcommand{\dps}{\displaystyle}
\newcommand{\ii}{\infty}
\newcommand\1{{\ensuremath {\mathds 1} }}
\renewcommand\phi{\varphi}

\newcommand{\cS}{\mathcal{S}}
\newcommand{\cP}{\mathcal{P}}

\newcommand{\cV}{\mathcal{V}}

\newcommand{\cU}{\mathcal{U}}
\newcommand{\cR}{\mathcal{R}}

\newcommand{\cB}{\mathcal{B}}
\newcommand{\cI}{\mathcal{I}}

\newcommand{\cF}{\mathcal{F}}
\newcommand{\cG}{\mathcal{G}}
\newcommand{\cN}{\mathcal{N}}

\newcommand\pscal[1]{{\ensuremath{\left\langle #1 \right\rangle}}}
\newcommand{\norm}[1]{ \left\| #1 \right\|}

\renewcommand{\geq}{\geqslant}
\renewcommand{\leq}{\leqslant}

\renewcommand{\tilde}{\widetilde}
\newcommand{\eps}{\varepsilon}

\newcommand{\nn}{\nonumber}
\newcommand{\rd}{\mathrm{d}}
\newcommand{\dx}{\rd x}
\newcommand{\dy}{\rd y}

\crefname{figure}{Figure}{Figures}

\DeclareMathOperator{\supp}{supp}

\DeclarePairedDelimiter\myp()
\DeclarePairedDelimiter\myt\{\}
\DeclarePairedDelimiter\myb[]
\DeclarePairedDelimiter\abs\lvert\rvert
\DeclarePairedDelimiter\expec\langle\rangle

\newcommand{\bQ}{\mathbb{Q}}
\newcommand{\id}{\, \mathrm{d}}

\newcommand\nonarg{{}\cdot{}}

\NewDocumentCommand\For{ s m }{
	\IfBooleanF{#1}{\quad}
	\quad
	\text{#2}
}

\providecommand\given{}
\newcommand\SetSymbol[1][]{%
	\nonscript\:#1\vert
	\allowbreak
	\nonscript\:
	\mathopen{}}
\DeclarePairedDelimiterX\Set[1]\{\}{%
	\renewcommand\given{\SetSymbol[\delimsize]}
	#1
}

\NewDocumentCommand\normt{ s o m m }{
	\IfBooleanTF{#1}{
		\norm*{#3}
	}{
		\IfNoValueTF{#2}{
			\norm{#3}
		}{
			\norm[#2]{#3}
		}
	}
	_{#4}
}

\title[LDA in Classical DFT]{Classical Density Functional Theory: The Local Density Approximation}

\author[M. Jex]{Michal Jex}
\address{Department of Physics, Faculty of Nuclear Sciences and Physical Engineering, Czech Technical University in Prague, B\v rehov\'a 7, 11519 Prague, Czech Republic }
\email{michal.jex@fjfi.cvut.cz}

\author[M. Lewin]{Mathieu Lewin}
\address{CNRS \& CEREMADE, Universit\'e Paris-Dauphine, PSL University, 75016 Paris, France}
\email{mathieu.lewin@math.cnrs.fr}

\author[P. S. Madsen]{Peter S. Madsen}
\address{Department of Mathematics, LMU Munich, Theresienstrasse 39, 80333 Munich, Germany}
\email{madsen@math.lmu.de}

\date{\today.}

\begin{document}

\begin{abstract}
We prove that the lowest free energy of a classical interacting system at temperature $T$ with a prescribed density profile $\rho(x)$ can be approximated by the local free energy $\int f_T(\rho(x))\,\rd x$, provided that $\rho$ varies slowly over sufficiently large length scales. A quantitative error on the difference is provided in terms of the gradient of the density. Here $f_T$ is the free energy per unit volume of an infinite homogeneous gas of the corresponding uniform density. The proof uses quantitative Ruelle bounds (estimates on the local number of particles in a large system), which are derived in an appendix.

\medskip

\tiny \noindent \copyright\;2023, the authors.
\end{abstract}

\maketitle

\tableofcontents

Classical Density Functional Theory is used to describe finite and infinite classical particle systems, which have a non-uniform density profile. A typical practical situation is the interface between liquid-gas, liquid-liquid (in fluid mixtures), crystal-liquid and crystal-gas phases at bulk coexistence~\cite{EbnSaaStr-76,EbnPun-76,SaaEbn-77,YanFleGib-76,Evans-79,Evans-92}. The density is then essentially constant within each phase and varies continuously in a neighborhood of the interface.

The \emph{quantum} theory of DFT was developed since the 1960s with the famous works by Hohenberg-Kohn-Sham~\cite{HohKoh-64,KohSha-65,LewLieSei-23_DFT}. In this theory one has a ladder of approximate models of increasing precision~\cite{PerSch-01}. The situation is somewhat similar in classical DFT, with the additional difficulty that the interaction between the particles is only known empirically.

The simplest nonlinear functional of the density is obtained by assuming that the gas is locally infinite and uniform, which is called the \emph{Local Density Approximation} (LDA). In this theory, the lowest free energy $G_T[\rho]$ at density $\rho(x)$ and temperature $T$ is approximated by the local functional
\begin{equation}
 G_T[\rho]\approx\int_{\R^d}f_T\big(\rho(x)\big)\,\dx,
 \label{eq:LDA_intro}
\end{equation}
where $f_T(\rho_0)$ is by definition the free energy per unit volume of an infinite homogeneous gas of density $\rho_0$.
A common strategy to gain in precision is to add gradient corrections to the LDA~\eqref{eq:LDA_intro}, in order to account for the non-uniformity of $\rho(x)$. Constructing efficient functionals of this type turned out to be easier for the liquid-gas transition~\cite{EbnSaaStr-76,EbnPun-76,SaaEbn-77,YanFleGib-76} than for the solid-liquid transition~\cite{RamYus-77,YamYus-79,HayOxt-81}.

In a previous paper~\cite{JexLewMad-23}, we have investigated the representability of any density $\rho(x)$ for a given interaction and proved local upper bounds on $G_T[\rho]$. The most difficult situation is when the interaction potential is very repulsive (non integrable) at the origin, in which case some correlations have to be inserted in the trial state to ensure that the particles never get too close to each other. Using tools from functional analysis (in particular Besicovitch covering techniques), we could construct such states and evaluate their free energy, providing thereby upper bounds on the free energy $G_T[\rho]$ at fixed density.

In this paper we continue our investigation of the fixed density problem and justify the Local Density Approximation~\eqref{eq:LDA_intro}, in a regime where the density varies very slowly over long length scales. To be more precise, our goal is to prove a quantitative estimate in the form
\begin{equation}
\left| G_T[\rho]-\int_{\R^d}f_T\big(\rho(x)\big)\,\dx\right|\leq \text{Err}[\rho]
 \label{eq:LDA_err_intro}
\end{equation}
valid for all densities, with the error $\text{Err}[\rho]$ involving gradients of the density so that it becomes negligible compared to the LDA whenever $\rho$ varies slowly. Although we expect the exact same result for the canonical free energy (denoted by $F_T[\rho]$ in this article), our results will mainly concern the easier grand-canonical free energy $G_T[\rho]$.

An estimate in the form~\eqref{eq:LDA_err_intro} has recently been proved in the classical and quantum Coulomb cases in~\cite{LewLieSei-18,LewLieSei-20}. The Coulomb potential is long range and the interaction between subsystems in a large assembly of particles is very hard to control, unless good screening properties have been shown. The works~\cite{LewLieSei-18,LewLieSei-20} relied on the Graf-Schenker inequality~\cite{GraSch-95}, following ideas from~\cite{HaiLewSol_1-09,HaiLewSol_2-09}. This inequality permits to withdraw such interactions in a lower bound, when the system is split using a tiling of space made of simplices. This is rather particular to the Coulomb case. Similar tools exist for general Riesz potentials $|x|^{-s}$~\cite{Fefferman-85,Gregg-89,Hugues-85,CotPet-19,CotPet-19b} as well as other positive-definite interactions~\cite{Mietzsch-20}, but not all interactions of physical interest can be handled by such a method.

In this work we only consider short range interaction potentials $w$ and do not have to face the problem of screening. Our goal is to however deal with arbitrary (stable) potentials decaying fast enough at infinity, without any assumption on the sign of the Fourier transform $\widehat{w}$ nor on that of $w$. The interaction can be partially attractive at medium or long distances. To control the interaction between subsystems we rely on \emph{Ruelle bounds}~\cite{Ruelle-70}, which provide estimates on the local moments of the number of particles in a large system. In a slowly-varying external potential, such a strategy has already been used in~\cite{Millard-72,GarSim-72,SimGar-73,MarPre-72}. The main novelties of the present work is that we work at fixed density and provide quantitative bounds, that is, an explicit error functional $\text{Err}[\rho]$ in~\eqref{eq:LDA_err_intro}. This requires to make Ruelle bounds quantitative (as is discussed in Appendix~\ref{app:Ruelle}) and to prove convergence rates for the thermodynamic limit.

The paper is organized as follows. In the next section we properly define the canonical and grand-canonical free energies $F_T[\rho]$ and $G_T[\rho]$ for a system at a given density profile $\rho(x)$. We also recall the definition of the free energy $f_T(\rho_0)$ of an infinite gas of uniform density $\rho_0$ and state our main result about the LDA, Theorem~\ref{thm:lda}. The rest of the paper is devoted to the proof of this theorem. In Section~\ref{sec:estimates} we gather useful \emph{a priori} estimates from~\cite{JexLewMad-23}, and others whose proofs are given much later in Section~\ref{sec:tdlim_proofs} for convenience. In Section~\ref{sec:thermolim} we prove that the thermodynamic limit at fixed uniform density is the same as the usual thermodynamic limit without any pointwise constraint on $\rho$. In Section~\ref{sec:gcconvrate} we study the speed of convergence for the grand-canonical problem, which depends on the decay of the interaction potential at infinity. In Section~\ref{sec:lda} we are finally able to provide the full proof of Theorem~\ref{thm:lda}. Appendix~\ref{app:Ruelle} contains the proof of quantitative Ruelle bounds in the spirit of~\cite{Ruelle-70}. We mainly follow the arguments of~\cite{Ruelle-70} but provide all the details for the convenience of the reader and give more explicit estimates than in~\cite{Ruelle-70}.

\bigskip

\noindent{\textbf{Acknowledgement.}} This project has received funding from the European Research Council (ERC) under the European Union's Horizon 2020 research and innovation programme (grant agreement MDFT No 725528 of ML). MJ also received financial support from the Ministry of Education, Youth and Sport of the Czech Republic under the Grant No.\ RVO 14000.

\section{Main result}

\subsection{Free energies at given density}
In order to state our main result, we need to first give the definition of the free energy at fixed density. Although most of our results will concern the grand-canonical case, we also partially look at the canonical case.

\subsubsection{Interaction potential $w$}
For convenience, we work in $\R^d$ with a general dimension $d\geq1$. The physical cases of interest are of course $d\in\{1,2,3\}$ but the proofs are the same for all $d$.
We consider systems of indistinguishable classical particles  interacting through a short-range pair potential $ w $.
Throughout the paper, we work with an interaction satisfying the following properties.

\begin{assumption}[on the short-range potential $w$]
	\label{de:shortrangenew}
	Let $ w :\R^d\to\R\cup\{+\ii\} $ be an even lower semi-continuous function satisfying the following properties, for some constant $\kappa>0$:

\smallskip

\noindent \textnormal{(1)} $ w $ is \emph{superstable}, that is, can be written in the form
$$w=\kappa^{-1}\1(|x|\leq\kappa^{-1})+w_2$$
with $w_2$ stable:
		\begin{equation}
		\label{eq:stabilitynew}
			\sum\limits_{1 \leq j < k \leq N} w_2 \myp{x_j - x_k} \geq - \kappa N
		\end{equation}
for all $ N \in \N $ and $ x_1, \dotsc, x_N \in \R^d $;

\medskip

\noindent \textnormal{(2)} $w$ is \emph{upper and lower regular}, that is,
		\begin{equation}
			\frac{\1(|x|<r_0)}{\kappa} \left(\frac{r_0}{|x|}\right)^\alpha - \frac{\kappa}{1 + \abs{x}^s} \leq w(x)
			\leq \kappa\1(|x|<r_0)\left(\frac{r_0}{|x|}\right)^\alpha+ \frac{\kappa}{1 + \abs{x}^s}
		\label{eq:assumption_w}
		\end{equation}
		for some $ r_0 \geq 0 $, $ 0 \leq \alpha < \infty $ and $s>d$.
\end{assumption}

In words, we assume that $w$ is repulsive close to the origin and behaves like $|x|^{-\alpha}$, for some $0\leq\alpha<\ii$. Even if some of our results are also valid when $\alpha=+\ii$ (our convention is then that $(r_0/|x|)^\alpha=(+\ii)\1(|x|<r_0)$), we assume $\alpha<\ii$ for simplicity throughout the paper and only make some remarks about the hard core case $\alpha=+\ii$. At infinity, we assume that $w$ decays at least polynomially and is integrable. The definition differs sligthly from the one we used in \cite{JexLewMad-23}, where no lower bound was required since we only studied upper bounds on the free energy. Assumption~\ref{de:shortrangenew} covers most cases of physical interest such as the Lennard-Jones potential $w(x)=a|x|^{-12}-b|x|^{-6}$ for $a>0$ or the Yukawa potential $w(x)=e^{-m|x|}/|x|$, for instance.

\subsubsection{Canonical free energy}

In this subsection we define the canonical free energy $F_T[\rho]$ at given density $\rho$. Suppose that we have $ N $ particles in $ \R^d $, distributed according to some Borel probability measure $ \bP $ on $ \R^{dN} $. Since the particles are indistinguishable, we demand that the measure $ \bP $ is symmetric, that is,
\begin{equation*}
	\bP(A_{\sigma \myp{1}}\times\cdots\times A_{\sigma \myp{N}})
	= \bP(A_1\times\cdots\times A_N)
\end{equation*}
for any permutation $ \sigma $ of $ \Set{1, \dotsc, N} $, and any Borel sets $A_1,...,A_N\subset\R^d$. The one-body density of such a symmetric probability $ \bP $ equals $N$ times the first marginal of $\bP$, that is,
\begin{equation*}
	\rho_{\bP} = N\int_{\R^{d \myp{N-1}}} \id \bP \myp{\nonarg,x_2, \dotsc, x_N},
\end{equation*}
where the integration is over $ x_2,\dotsc,x_N $. Equivalently, $\rho_\bP(A)=N\bP(A\times(\R^d)^{N-1})$ for every Borel set $A$. Note the normalization convention $ \rho_{\bP}(\R^d) = N $. For a non-symmetric probability $\bP$ we define $\rho_\bP$ as the sum of the $N$ marginals.

The pairwise average interaction energy of the particles is given by
\begin{equation*}
	\cU_{N} \myp{\bP}
	= \int_{\R^{dN}} \sum\limits_{1 \leq j < k \leq N} w \myp{x_j - x_k} \id \bP \myp{x_1, \dotsc, x_N}.
\end{equation*}
It could in principle be equal to $+\ii$, but it always satisfies $\cU_{N} \myp{\bP}\geq -\kappa N$ due to the stability condition on $w$ in Assumption~\ref{de:shortrangenew}.
When considering systems at positive temperature $T>0$, it is necessary to also include the entropy of the system,
\begin{equation}
\label{eq:canentropy}
	\cS_N \myp{\bP} := - \int_{\R^{dN}} \bP \myp{x} \log \big(N! \, \bP \myp{x}\big) \id x.
\end{equation}
If $ \bP $ is not absolutely continuous with respect to the Lebesgue measure on $ \R^{dN} $, we use the convention that $\cS_N \myp{\bP}=-\ii$.
The total free energy of the system in the state $ \bP $ at temperature $ T \geq 0 $ equals
\begin{equation}
\label{eq:FFcan}
	\mathcal{F}_T \myp{\bP}
	:={} \cU_{N} \myp{\bP} - T \cS_N \myp{\bP}
	= \int_{\R^{dN}} \sum\limits_{j < k} w \myp{x_j - x_k} \id \bP \myp{x} + T \int_{\R^{dN}} \bP \log \myp{N! \, \bP }.
\end{equation}

Throughout the paper, we will mostly consider systems with a given one-body density $\rho$, which is absolutely continuous with respect to the Lebesgue measure. At $T>0$ we also assume that $\int_{\R^d}\rho|\log\rho|<\ii$. This allows us to consider the minimal energy of $ N $-particle classical systems with density $ \rho $, given by
\begin{equation}
\label{eq:Fcan}
	\boxed{
		F_T [\rho]
		:= \inf_{\rho_{\bP} = \rho} \mathcal{F}_T \myp{\bP}
	}
\end{equation}
where the infimum is taken over $ N $-particle states $ \bP $ on $ \R^{dN} $ with one-particle density $ \rho_{\bP} $ equal to $ \rho $. The number $F_T[\rho]$ turns out to be finite for all $\rho\in L^1(\R^d,\R_+)$ of integer mass such that $T\int_{\R^d}\rho|\log\rho|<\ii$, as we have proved in~\cite{JexLewMad-23} using results from optimal transport theory in~\cite{ColMarStr-19}.

\subsubsection{Grand-canonical free energy}
In the grand-canonical ensemble, the exact particle number of the system is not fixed. A state $ \bP $ is a family of symmetric $ n $-particle positive measures $ \bP_n $ on $(\R^d)^n$, so that
\begin{equation*}
	\sum_{n \geq 0} \bP_n\big((\R^d)^n\big)=1.
\end{equation*}
Here $\bP_0$ is just a number, interpreted as the probability that there is no particle at all in the system. After replacing $\bP_n$ by $\bP_n/\bP_n(\R^{dn})$, we can equivalently think that $\bP$ is a convex combination of canonical states.
The entropy of $ \bP $ is defined by
\begin{equation}
\label{eq:gcentropy}
	\cS \myp{\bP}
	:= \sum\limits_{n \geq 0} \cS_n \myp{\bP_n}
	= -\bP_0\log(\bP_0)- \sum\limits_{n \geq 1} \int_{\R^{dn}} \bP_n \log \myp{n! \, \bP_n },
\end{equation}
and the single particle density of the state $ \bP $ is
\begin{equation*}
	\rho_{\bP} = \sum\limits_{n \geq 1} \rho_{\bP_n}=\sum_{n\geq1}n\int_{(\R^d)^n}\rd\bP_n(\nonarg,x_2,\dotsc,x_n).
\end{equation*}
Its integral gives the average number of particles in the system:
\begin{equation}
 \int_{\R^d}\rho_\bP(x)\,\rd x=\sum\limits_{n \geq 1} n\, \bP_n\big(\R^{dn}\big)=:\cN(\bP).
 \label{eq:average_N}
\end{equation}
The grand-canonical free energy of the state $ \bP $ at temperature $ T \geq 0 $  is
\begin{equation}
\label{eq:FFgc}
	\mathcal{G}_T \myp{\bP}
	:= \cU \myp{\bP} - T \cS \myp{\bP}, 
\end{equation}
where $ \cU \myp{\bP} $ denotes the interaction energy in the state $ \bP $,
\begin{equation}
\label{eq:gcinteraction}
	\cU \myp{\bP}
	:= \sum\limits_{n \geq 2} \cU_{n} \myp{\bP_n}
	= \sum\limits_{n \geq 2} \int_{\R^{dn}} \sum\limits_{j<k}^n w \myp{x_j - x_k} \id \bP_n \myp{x_1,...,x_N}.
\end{equation}
From the stability of $w$ we have
$$\cU_{n} \myp{\bP_n}\geq -\kappa n\,\bP_n(\R^{dn})$$
so that, after summing over $n$,
$$	\cU \myp{\bP}\geq -\kappa \int_{\R^d}\rho_\bP(x)\,\dx.$$
When keeping the one-particle density $ \rho = \rho_{\bP} \in L^1 \myp{\R^d} $ fixed, we denote the minimal grand-canonical free energy by
\begin{equation}
\label{eq:gcenergydensity}
\boxed{
	G_T [\rho]
	:= \inf_{\rho_{\bP} = \rho} \mathcal{G}_T \myp{\bP}.
}
\end{equation}
When $\int_{\R^d}\rho$ is an integer, it is clear that $F_T[\rho]\geq G_T[\rho]$. The functional $\rho\mapsto G_T[\rho]$ was studied in~\cite{MarLewNen-22_ppt}. It is a kind of weak--$\ast$ convex envelope of $\rho\mapsto F_T[\rho]$.

\subsection{The thermodynamic limit}
Here we introduce the free energy per unit volume $f_T$, which is obtained when the system is placed at equilibrium in a large domain $\Omega_N\subset\R^d$ which then grows so as to cover the whole space, without any pointwise constraint on the density $\rho(x)$.

First we discuss how the domain is allowed to grow. Let $ \Omega_N \subseteq \R^d $ be a sequence of bounded, connected domains with $ \abs{\Omega_N} \to \infty $. The sequence $ \Omega_N $ is said to have a uniformly regular boundary~\cite{Fisher-64,HaiLewSol_1-09} if there exists a $ t_0 > 0 $ such that
	\begin{equation}
	\label{eq:regularity}
		\abs[\big]{\Set[\big]{x \in \R^d \mid \rd \myp{x,\partial \Omega_N} \leq \abs{\Omega_N}^{\frac{1}{d}} t}}
		\leq \abs{\Omega_N} \frac{t}{t_0}, \For{for all $ t \in \myb{0,t_0} $}.
	\end{equation}
This assumption can be weakened in many ways but it already covers any rescaled sequence $\Omega_N=N^{1/d}\Omega$ with $\partial\Omega$ sufficiently smooth (by parts), such as cubes or balls, or in fact any convex set~\cite{HaiLewSol_1-09}. 	For any sequence $(\Omega_N)$ with a uniformly regular boundary and satisfying $ N / \abs{\Omega_N} \to \rho_0 \geq0$, it is well known~\cite{Ruelle} that the thermodynamic limit exists,
	\begin{equation}
	\label{eq:fcan}
		f_T \myp{\rho_0} := \lim_{\substack{N \to \infty \\ \frac{N}{|\Omega_N|}\to \rho_0}} |\Omega_N|^{-1}\min_{\bP\in\cP_s(\Omega_N^N)}\cF_T(\bP)=\lim_{\substack{N \to \infty \\ \frac{N}{|\Omega_N|} \to \rho_0}} |\Omega_N|^{-1}\min_{\cN(\bP)=N}\cG_T(\bP),
	\end{equation}
	and is independent on the sequence of domains $ \myp{\Omega_N} $. The number $f_T(\rho_0)$ is interpreted as the free energy per unit volume at temperature $T$ and density $\rho_0$ of an infinite gas in equilibrium. The first minimum in~\eqref{eq:fcan} is over all the symmetric probability measures $\bP$ on $(\Omega_N)^N$. Recall that the free energy $\cF_T(\bP)$ was defined in~\eqref{eq:FFcan}. The second minimum is over all grand-canonical probabilities $\bP=(\bP_n)_{n\geq0}$ supported on $\Omega_N$, which have the average number of particles $\cN(\bP)$ equal to the given $N$. In other words, we get the same limit by fixing the number of particles exactly, or by only fixing its average value.
	Let us insist on the fact that the limit~\eqref{eq:fcan} holds without any pointwise constraint on the density function $\rho(x)$ of the system. One can in fact rewrite it as
	\begin{equation}
	\label{eq:fcan_rho}
		f_T \myp{\rho_0} = \lim_{\substack{N \to \infty \\ \frac{N}{|\Omega_N|} \to \rho_0}} |\Omega_N|^{-1}\min_{\substack{\rho\in L^1(\Omega_N)\\ \int_{\Omega_N}\rho=N}}F_T[\rho]		=\lim_{\substack{N \to \infty \\ \frac{N}{|\Omega_N|} \to \rho_0}} |\Omega_N|^{-1}\min_{\substack{\rho\in L^1(\Omega_N)\\ \int_{\Omega_N}\rho=N}}G_T[\rho]
	\end{equation}
since minimizing $F_T[\rho]$ (resp. $G_T[\rho]$) over $\rho$ is the same as minimizing $\cF_T[\bP]$ (resp. $\cG_T[\bP]$) over all possible $\bP$.

The thermodynamic limit function $\rho_0\mapsto  f_T(\rho_0) $ is known to be convex (and $C^1$ if $T>0$)~\cite{Ruelle-70}. Its derivative is the chemical potential $\mu=f'_T(\rho_0)$ and its Legendre transform is the grand-canonical free energy in this chemical potential, defined through the following limit
	\begin{align}
	g_T \myp{\mu} &:= \lim_{n \to \infty}  |\Omega_n|^{-1}\min_{\bP}\big\{\cG_T(\bP)-\mu \,\cN(\bP)\big\}\nn\\
		&=\lim_{n \to \infty}  |\Omega_n|^{-1}\min_{\rho\in L^1(\Omega_n)}\left\{G_T[\rho]-\mu\int_{\Omega_n}\rho\right\},\label{eq:fgc}
	\end{align}
	again independent on the sequence of domains $ \myp{\Omega_n} $. The thermodynamic limit function $ g_T $ is concave (strictly concave when $T > 0$), non-positive, and continuous in the chemical potential $ \mu $. It is the Legendre transform of $f_T$, that is,
	\begin{equation}
		f_T \myp{\rho_0} = \sup_{\mu \in \R} \Set{\mu \rho_0 + g_T \myp{\mu}}, \qquad
		g_T \myp{\mu} = \inf_{\rho_0 \geq 0} \Set{ f_T \myp{\rho_0}-\mu \rho_0}.
	\label{eq:equiv}
		\end{equation}

\subsection{The local density approximation}

The free energies $F_T[\rho]$ and $G_T[\rho]$ are highly nonlinear and nonlocal functionals of the density function $\rho$. The Local Density Approximation (LDA) consists in replacing them by the simpler local functional
$$\rho\mapsto \int_{\R^d}f_T\big(\rho(x)\big)\,\rd x$$
where $f_T$ is the free energy per unit volume of an infinite gas defined previously in~\eqref{eq:fcan} and~\eqref{eq:fgc}.
Our goal is to justify this approximation when $\rho$ varies sufficiently slowly over some large length scale. To measure the amplitude of the variations of $\rho$, we introduce the function
\begin{equation}
\boxed{\delta\rho_\ell(z):=\sup_{x,y\in z+C_\ell}\frac{|\rho(x)-\rho(y)|}{\ell}}
\label{eq:def_delta_rho}
\end{equation}
where $C_\ell:=[-\ell/2,\ell/2]^d$ is the cube of side length $\ell$ centered at the origin and `sup' means the essential supremum (that is, up to sets of zero Lebesgue measure). The function~\eqref{eq:def_delta_rho} measures the variations of $\rho$ at distance of order~$\ell$ of a point $z\in\R^d$ and it is a kind of smeared derivative. When $\rho$ is constant over the whole space we just get $\delta\rho_\ell\equiv0$. When $\rho$ is constant over a large domain $\Omega$ and vanishes outside, we get $\delta\rho_\ell\equiv0$ except at distance $\ell$ from the boundary $\partial\Omega$. The following is the main result of this paper.

\begin{theorem}[Local Density Approximation]
\label{thm:lda}
	Let $ M > 0 $, $ p \geq 1 $, $T\geq0$ and
	\begin{equation}
	\label{eq:b_cond0}
		b > \begin{cases}
			2 - \frac{1}{2p} & \text{if } p \geq 2, \\
			\frac{3}{2}+\frac{1}{2p} & \text{if } 1 \leq p < 2.
		\end{cases}
	\end{equation}
	Let $ w $ be a short-range interaction satisfying \cref{de:shortrangenew} with $s>d+1$. There exists a constant $ C > 0 $ depending on $ M,T,w,d,p,b$, such that
	\begin{equation}
		\abs[\bigg]{G_T \myb{\rho} - \int_{\R^d} f_T \myp{\rho \myp{x}} \id x}
		\leq  \frac{C}{\sqrt{\ell}} \myp[\bigg]{\int_{\R^d} \sqrt{\rho} + \ell^{bp} \int_{\R^d} \delta\rho_\ell(z)^p\,\rd z},
	\label{eq:lda}
	\end{equation}
	for any $ \ell > 0 $, and any density $ \rho\geq 0$ such that $\sqrt\rho\in (L^1\cap L^\ii)(\R^d)$ with $\normt{\rho}{\infty} \leq M $.
\end{theorem}

The proof of Theorem~\ref{thm:lda} is provided in Section~\ref{sec:lda}. Under the stated assumptions on $ \rho $, we also have $ \rho \log \rho \in L^1 \myp{\R^d} $, since
\begin{equation*}
	\int_{\R^d} \rho \abs{\log \rho}
	\leq \normt{\sqrt{\rho} \log \rho}{\infty} \int_{\R^d} \sqrt{\rho}.
\end{equation*}
In particular, by the \emph{a priori} bounds on $ G_T $ and $ f_T $ proved in~\cite{JexLewMad-23} and recalled in \cref{sec:estimates} below, the quantities appearing on the left hand side of \eqref{eq:lda} are all finite.

Theorem~\ref{thm:lda} states that the grand-canonical free energy functional can be approximated by the LDA functional $\rho\mapsto \int f_T(\rho(x))\,\rd x$, whenever the variations of $\rho$ are much smaller than the values of $\rho$, in the average sense
$$\int_{\R^d} \delta\rho_\ell(z)^p\,\rd z\ll \int_{\R^d}\sqrt\rho,$$
for some global large length $\ell\gg1$. We emphasize that our estimate~\eqref{eq:lda} depends on the $L^\ii$ norm of $\rho$ through the parameter $M$. From the universal bounds proved in~\cite{JexLewMad-23} and recalled in Section~\ref{sec:estimates}, it seems natural to expect a similar estimate without any constraint on $\|\rho\|_{L^\ii}$ and involving
$$\int_{\R^d}\rho+\rho^{\max(2,1+\alpha/d)}+T\rho(\log\rho)_-$$
in place of $\int_{\R^d}\sqrt\rho$ on the right side. Here we used the much stronger $L^1$ and $L^\ii$ norms of $\sqrt\rho$ to control errors.
We crucially use the boundedness of $\rho$ and, unfortunately, our proof provides a constant $C$ which diverges exponentially fast in the parameter $M$. Nevertheless, we believe that our assumption $\sqrt{\rho}\in (L^1\cap L^\ii)(\R^d)$ is reasonable in most situations of physical interest.

Let us emphasize the condition $s>d+1$ on the decay at infinity of $w$, which we have only used to simplify the statement. The powers of $\ell$ and the condition on $b$ are slightly different for $d<s\leq d+1$. This is all explained later in Propositions~\ref{thm:ldalow} and~\ref{thm:ldaup}.

If we rescale a function $\rho$ in the manner $\rho'(x)=\rho(hx)$ then we have the scaling relation $\delta\rho'_\ell(z)=h\delta\rho_{h\ell}(hz)$. Hence, applying~\eqref{eq:lda} to $\rho_N(x):=\rho(N^{-1/d}x)$ with a smooth~$\rho$, we obtain after taking $\ell=N^{\frac{1}{bd}}$ and $p=2$
$$G_T \Big[\rho(N^{-\frac1d}\cdot)\Big]=N\int_{\R^d} f_T \myp{\rho \myp{x}} \id x+O\!\left(N^{1-\frac{1}{2bd}}\right),\qquad\forall b>\frac{7}{4}.$$
To be more precise, we need here that $\limsup_{h\to0}\int_{\R^d}\delta\rho_h(z)^2\,\rd z<\ii$, which is the case if for instance $\rho$ is $C^1$ with compact support.

Next we mention two easy corollaries of Theorem~\ref{thm:lda}. The first is when we measure the variations of $\rho$ using derivatives instead of the function $\delta\rho_\ell$. In the Coulomb case, this point of view goes back to~\cite{LewLieSei-20}.

\begin{corollary}
\label{cor:lda}
	Suppose, in addition to the assumptions in \cref{thm:lda}, that $ p > d $ and $ \nabla \rho \in L^p \myp{\R^d} $.
	Then we have
	\begin{equation}
		\abs[\bigg]{G_T \myb{\rho} - \int_{\R^d} f_T \myp{\rho \myp{x}} \id x}
		\leq C \eps \left(\int_{\R^d} \sqrt{\rho} + \frac{1}{\eps^{2bp}} \int_{\R^d} \abs{\nabla \rho}^p \right)
		\label{eq:ldalow_cor}
	\end{equation}
for any $ \eps > 0 $.
	\end{corollary}

\begin{proof}
Morrey's inequality~\cite[Proof of Lem.~4.28]{AdamsFournier} in a cube $Q$ states that
	\begin{equation*}
		\abs{\rho(x)-\rho(y)}
		\leq{} K^{\frac1p}\abs{x-y}^{1-\frac{d}{p}} \myp[\Big]{\int_{Q} \abs{\nabla \rho}^p}^{\frac{1}{p}},
		\qquad x,y \in Q,
	\end{equation*}
for $p>d$ with $K$ independent of the location of the cubes. After scaling, this implies the pointwise bound
	\begin{equation*}
\delta\rho_\ell(z)^p\leq \frac{K}{\ell^d}\int_{z+C_\ell} \abs{\nabla \rho}^p,\qquad\forall z\in\R^d.
	\end{equation*}
Hence we obtain after integration
	\begin{equation*}
\int_{\R^d}\delta\rho_\ell(z)^p\,\rd z\leq \frac{K}{\ell^d}\int_{\R^d}\left(\int_{z+C_\ell} \abs{\nabla \rho}^p\right)\rd z=K\int_{\R^d} \abs{\nabla \rho}^p.
	\end{equation*}
The bound \eqref{eq:ldalow_cor} thus follows from \eqref{eq:lda} with $ \eps=\ell^{-1/2}$.
\end{proof}

In some practical situations, it is useful to have the variations of $\rho$ expressed using $\delta\rho_\ell$ instead of derivatives as in Corollary~\ref{cor:lda}.
An example is when $\rho$ is constant over a large domain.

\begin{corollary}[Thermodynamic limit at constant density]
\label{cor:thermo-limit}
	Let $ w $ be a short-range interaction satisfying \cref{de:shortrangenew}.
	For any $\rho_0,>0$, $T\geq0$, and $r>1$ we have
	\begin{equation}
		\Big|G_T \big[\rho_0\1_\Omega\big] - |\Omega|\,f_T \myp{\rho_0}\Big|\leq C|\Omega|^{1-\frac{1}{4dr}}
	\label{eq:lda_cnst}
	\end{equation}
	for any large-enough domain $\Omega$ with a regular boundary as in~\eqref{eq:regularity}.
	The constant $ C > 0 $ depends on $ \rho_0,T, r, w $ as well as on the parameter $t_0$ in~\eqref{eq:regularity}.
\end{corollary}

\begin{proof}
The function $\delta\rho_\ell$ vanishes everywhere except at a distance $\sqrt{d}\ell/2$ of $\partial\Omega$, where it is bounded above by $\rho_0/\ell$. Thus, we have from the regularity~\eqref{eq:regularity} of the boundary
$$\int_{\R^d} \delta\rho_\ell(z)^p\,\rd z\leq \frac{\rho_0^p}{\ell^p} \Big|\big\{x \in \R^d \mid \rd \myp{x,\partial \Omega}\leq\sqrt{d}\ell/2\big\}\Big|\leq C|\Omega|^{1-\frac1d}\ell^{1-p},$$
provided that $\ell|\Omega|^{-1/d}\leq 2t_0/\sqrt{d}$. Optimizing over $\ell$ leads us to the choice
$$\ell^{bp+1-p}=|\Omega|^{\frac1d}$$
with $ b $ satisfying \eqref{eq:b_cond0}, which is allowed provided that $|\Omega|$ is large enough compared with $t_0/\sqrt d$.
Since we can choose any $p\geq 1$, this can be stated as in~\eqref{eq:lda_cnst} by minimizing the power $bp+1-p$.
\end{proof}

Corollary~\ref{cor:thermo-limit} means that we get the same thermodynamic limit as in~\eqref{eq:fcan} when we enforce the constraint that the density is constant everywhere instead of just fixing the average total number of particles per unit volume. This is of course a consequence of the translation-invariance of the problem. For Coulomb and Riesz gases, the similar property is not at all obvious and was recently proved  in~\cite{CotPet-19b,LewLieSei-19b,Lauritsen-21}.
Although we have stated~\eqref{eq:lda_cnst} as a corollary, our proof of Theorem~\ref{thm:lda} in fact goes by considering the case of constant densities first.
In \cref{prop:gcconvrate_fixed} we give a direct proof in the case of \emph{cubes}, which even provides a better estimate than \eqref{eq:lda_cnst}, with the right hand side instead behaving like $\abs{\Omega}^{1-\frac{1}{2d}}$ for large $\abs{\Omega}$.
A result similar to Corollary~\ref{cor:thermo-limit} holds in the case of several phases of different densities (e.g. over two half spaces).

\begin{remark}[Cluster expansions]
In the recent paper~\cite{JanKunTsa-22}, the dual potential $V$ (whose equilibrium Gibbs state has density $\rho$) is expressed as a convergent series in $\rho$, under the assumption that $\|\rho\|_{L^\ii}$ is small enough. One can then express $G_T[\rho]$ as a convergent series in $\rho$, and thereby probably obtain bounds better than~\eqref{eq:lda} and~\eqref{eq:ldalow_cor} for $M$ small enough. We insist that our result is valid for all possible values of $\|\rho\|_{L^\ii}$. The latter only appears implicitly in the constant $C$.
\end{remark}

We can deduce from~\eqref{eq:lda} a lower bound on the canonical free energy $F_T[\rho]$ using the fact that $F_T[\rho]\geq G_T[\rho]$ whenever $\int_{\R^d}\rho\in\N$. We expect an upper bound similar to~\eqref{eq:lda} but it is always harder to construct good trial states in the canonical case. Our paper will nevertheless contain several intermediate results valid for $F_T[\rho]$. In particular, we are able to prove the existence of the thermodynamic limit at constant density in the canonical case.

\begin{theorem}[Thermodynamic limit at constant density, canonical case]\label{thm:thermolim}
Let $ w $ be a short-range interaction satisfying \cref{de:shortrangenew}.
Let $\rho_0>0$. Suppose that $ \Omega_N \subseteq \R^d $ is a sequence of bounded connected domains with uniformly regular boundaries as in~\eqref{eq:regularity} for some $t_0$, and such that $ \abs{\Omega_N} \to \infty $ and $ \rho_0 \abs{\Omega_N} \in \mathbb{N} $ for all $ N $. Then we have
	\begin{equation}
		\lim_{N\to\ii}\frac{F_T \big[\rho_0\1_{\Omega_N}\big]}{|\Omega_N|}=f_T \myp{\rho_0}
	\label{eq:lda_cnst_can}
	\end{equation}
for any $ T \geq 0 $.
\end{theorem}

The proof of Theorem~\ref{thm:thermolim} is provided in Section~\ref{sec:thermolim}.

\section{A priori estimates}
\label{sec:estimates}
In this short section, we briefly recall and discuss few \emph{a priori} estimates on the free energy functionals $ F_T $ and $ G_T $, as well as the minimal energies per unit volume for cubes.

\subsection{Universal bounds}
By~\cite[Lemma 6.1]{MarLewNen-22_ppt} and \cite{JexLewMad-23}, we have the universal lower energy bound, which holds for any $ 0 \leq \rho \in L^1 \myp{\R^d} $ satisfying $ \int_{\R^d} \rho \abs{\log \rho} < \infty $,
\begin{equation}
	G_T [\rho]
	\geq - \myp[\big]{\kappa  + T} \int_{\R^d} \rho + T \int_{\R^d} \rho \log \rho,
\label{eq:lower_bd_G_T_simple}
\end{equation}
where $ \kappa $ is the stability constant of $ w $ in Assumption~\ref{de:shortrangenew}.
From \cite[Thm.~10 \& 11]{JexLewMad-23}, we also have for $ \alpha \neq d $ the upper bound
\begin{equation}
\label{eq:upper_bd_G_T1}
	G_T \myb{\rho}
	\leq  C \int_{\R^d}\rho^{\gamma} + C\myp{1+T} \int_{\R^d} \rho + T \int_{\R^d} \rho \log \rho,
\end{equation}
where
\begin{equation}
\label{eq:gammadef}
	\boxed{
		\gamma := 1 + \max \myp{1, \alpha/d},
	}
\end{equation}
and the constant $ C $ depends only on the dimension $ d $ and the interaction $ w $.
When $ \alpha = d $, the bound instead takes the form
\begin{equation}
\label{eq:upper_bd_G_T2}
	G_T \myb{\rho}
	\leq  C \int_{\R^d}\rho^2 \myp{\log \rho}_+ + C\myp{1+T} \int_{\R^d} \rho + T \int_{\R^d} \rho \log \rho.
\end{equation}
In particular, when the density $ \rho $ is uniformly bounded, $ \normt{\rho}{\infty} \leq M $, we have the simple bound
\begin{equation}
\label{eq:upper_bd_G_T_simple}
	G_T \myb{\rho}
	\leq  C_M \int_{\R^d} \rho + T \int_{\R^d} \rho \log \rho,
\end{equation}
where the constant $ C_M $ now also depends on $ M $ and the temperature $ T $.
By~\cite[Thm.~14]{JexLewMad-23}, this bound also holds in the hard-core case ($ \alpha = \infty $), provided that $ \normt{\rho}{\infty} \leq \myp{1-\eps}^d r_0^{-d} \rho_c \myp{d} $, where $ \rho_c \myp{d} $ is the sphere packing density in $ d $ dimensions, and $ 0 < \eps < 1 $.
The constant $ C_M $ in this case then also depends on $ \eps $, and behaves like $ \abs{\log \eps} $.

In the canonical case, we have by~\cite[Thm.~12]{JexLewMad-23} a bound depending on the local radius $ R \myp{x} $ of $ \rho $, which is defined to be the largest number satisfying
\begin{equation}
\label{eq:Rdef}
	\int_{B \myp{x, R \myp{x}}} \rho \myp{y} \id y = 1.
\end{equation}
For $ \alpha \neq d $, the bound takes the form
\begin{equation}
\label{eq:upper_bd_F_T}
	F_T \myb{\rho}
	\leq  C \int_{\R^d}\rho^{\gamma}+ C \int_{\R^d} \rho + T \int_{\R^d} \rho \log \rho+T \int_{\R^d} \rho \log R^d,
\end{equation}
and similarly for $ \alpha = d $, where the $ \rho^{\gamma} $ term is replaced by $ \rho^2 \myp{\log \rho}_+ $.

Finally, we also have the following useful (but non-optimal) sub-additivity bound for the grand-canonical energy.

\begin{proposition}[Sub-additivity estimate]
\label{prop:subadditivity}
	Suppose that $ w $ satisfies Assumption~\ref{de:shortrangenew}.
	For all $ 0 < \eps \leq 1/2 $, there is a constant $ C > 0 $ such that for any pair of densities $ 0 \leq \rho_1, \rho_2 \in L^1 \myp{\R^d} \cap L^{\gamma} \myp{\R^d} $, we have for $ \alpha \neq d $,
	\begin{align}
		\label{eq:subadditivity}
		G_T \myb{\rho_1 + \rho_2}
		\leq{}& G_T \myb{\rho_1} + C \eps \int_{\R^d} \rho_1^{\gamma} + \rho_1 \nn \\
		&+ \frac{C}{\eps^{\gamma-1}} \int_{\R^d} \rho_2^{\gamma} + C \myp{\log \eps}_- \int_{\R^d} \rho_2 + T \int_{\R^d} \rho_2 \log \rho_2\,,
	\end{align}
	where $ \gamma = 1 + \max \myp{1, \alpha / d} $.
	When $ \alpha = d $, we get instead
	\begin{align}
	\label{eq:subadditivity2}
		G_T \myb{\rho_1 + \rho_2}
		\leq{}& G_T \myb{\rho_1} + C \eps \int_{\R^d} \rho_1^2 \myp{\log \rho_1}_+ + \rho_1 + \frac{C}{\eps} \int_{\R^d} \rho_2^2 \myp{\log \rho_2}_+ \nn \\
		&+ \frac{C \myp{\log \eps}_-}{\eps} \int_{\R^d} \rho_2^2 + C \myp{\log \eps}_- \int_{\R^d} \rho_2 + T \int_{\R^d} \rho_2 \log \rho_2.
	\end{align}
\end{proposition}
\begin{proof}
	We use the same approach as in \cite[Lemma 19]{LewLieSei-20}.
	Write $ \rho_1 + \rho_2 = \myp{1-\eps} \rho_1 + \eps \myp{\frac{\rho_2}{\eps} + \rho_1} $, and let $ \bP_1 $ and $ \bP_2 $ be any grand-canonical states with densities $ \rho_1 $ and $ \frac{\rho_2}{\eps} + \rho_1 $, respectively.
	Then $ \bP := \myp{1-\eps} \bP_1 + \eps \bP_2 $ is a trial state with density $ \rho_1 + \rho_2 $, so by minimizing over $ \bP_1 $, $ \bP_2 $, we obtain by convexity,
	\begin{equation*}
		G_T \myb{\rho_1+\rho_2}
		\leq{} \myp{1-\eps} G_T \myb{\rho_1} + \eps G_T \myb[\Big]{\frac{\rho_2}{\eps} + \rho_1}.
	\end{equation*}
The universal upper bound \eqref{eq:upper_bd_G_T1} on $ G_T \myb{\rho} $ provides
	\begin{align*}
		\eps G_T \myb[\Big]{\frac{\rho_2}{\eps} + \rho_1}
		\leq{}& C \eps \int_{\R^d} \myp[\Big]{\frac{\rho_2}{\eps} + \rho_1}^{\gamma} + \myp[\Big]{\frac{\rho_2}{\eps} + \rho_1} \\
		&+ \eps T \int_{\R^d} \myp[\Big]{\frac{\rho_2}{\eps} + \rho_1} \log \myp[\Big]{\frac{\rho_2}{\eps} + \rho_1} \\
		\leq{}& C \eps \int_{\R^d} \myp{\rho_1^{\gamma} + \rho_1} + \frac{C}{\eps^{\gamma - 1}} \int_{\R^d} \rho_2^{\gamma} + C \int_{\R^d} \rho_2 \\
		&+ \eps T \int_{\R^d} \myp[\Big]{\frac{\rho_2}{\eps} + \rho_1} \log \myp[\Big]{\frac{\rho_2}{\eps} + \rho_1}.
	\end{align*}
Writing $ \frac{\rho_2}{\eps} + \rho_1 = \eps \frac{\rho_2}{\eps^2} + \myp{1-\eps} \frac{\rho_1}{1-\eps} $ and using the concavity of the entropy, along with $ \eps \leq 1/2 $, we get
	\begin{align*}
		\MoveEqLeft[4] \eps \int_{\R^d} \myp[\Big]{\frac{\rho_2}{\eps} + \rho_1} \log \myp[\Big]{\frac{\rho_2}{\eps} + \rho_1}
		\leq{} \int_{\R^d} \rho_2 \log \myp[\Big]{\frac{\rho_2}{\eps^2}} + \eps \int_{\R^d} \rho_1 \log \myp[\Big]{\frac{\rho_1}{1-\eps}} \\
		\leq{}& 2 \myp{\log \eps}_- \int_{\R^d} \rho_2 + \int_{\R^d} \rho_2 \log \rho_2 + \eps \log 2 \int_{\R^d} \rho_1 + \eps \int_{\R^d} \rho_1 \log \rho_1.
	\end{align*}
	Applying the universal lower bound \eqref{eq:lower_bd_G_T_simple}, we conclude that \eqref{eq:subadditivity} holds.

	In the $ \alpha = d $ case, \eqref{eq:subadditivity2} is obtained in the same way, using instead the upper bound \eqref{eq:upper_bd_G_T2},
	\begin{align*}
		G_T \myb[\Big]{\frac{\rho_2}{\eps} + \rho_1}
		\leq{}& C \int_{\R^d} \myp[\Big]{\frac{\rho_2}{\eps} + \rho_1}^2 \myp[\Big]{\log \myp[\Big]{\frac{\rho_2}{\eps} + \rho_1}}_+ + \myp[\Big]{\frac{\rho_2}{\eps} + \rho_1} \\
		&+ T \int_{\R^d} \myp[\Big]{\frac{\rho_2}{\eps} + \rho_1} \log \myp[\Big]{\frac{\rho_2}{\eps} + \rho_1}.
	\end{align*}
The first term can be estimated by
	\begin{align*}
		\MoveEqLeft[6] \int_{\R^d} \myp[\Big]{\frac{\rho_2}{\eps} + \rho_1}^2 \myp[\Big]{\log \myp[\Big]{\frac{\rho_2}{\eps} + \rho_1}}_+ \\
		\leq{}& \int_{\R^d} 4 \max \myp[\Big]{ \frac{\rho_2}{\eps} , \rho_1}^2 \myp[\Big]{\log \myp[\Big]{2 \max \myp[\Big]{ \frac{\rho_2}{\eps} , \rho_1}}}_+ \\
		\leq{}& \int_{\R^d} 4 \myp[\Big]{\frac{\rho_2}{\eps} }^2 \myp[\Big]{\log \myp[\Big]{2 \frac{\rho_2}{\eps}}}_+ + \int_{\R^d} 4 \rho_1^2 \myp{\log \myp{2 \rho_1}}_+ \\
		\leq{}& \frac{4}{\eps^2} \int_{\R^d} \rho_2^2 \myp[\Big]{ \myp{\log \rho_2}_+ + \log \frac{2}{\eps}} + 4 \int_{\R^d} \rho_1^2 \myp[\big]{ \myp{\log \rho_1}_+ + \log 2}.
	\end{align*}
	The bound \eqref{eq:subadditivity2} now follows in the same way as before.
\end{proof}

\subsection{Energy per unit volume}
We denote by
\begin{equation}
\boxed{ G_T(\mu,\Omega):=\min_{\bP}\big\{\cG_T(\bP)-\mu\cN(\bP)\big\}=\min_{\rho\in L^1(\Omega)}\left\{G_T[\rho]-\mu\int_\Omega\rho\right\}}
\label{eq:def_G_T_mu_Omega}
\end{equation}
the minimum free energy with a chemical potential $\mu$ in an arbitrary bounded domain $\Omega$. The first minimum is above all possible grand-canonical states $\bP$ in the domain $\Omega$. It is well known~\cite{Ruelle} that the first minimum is attained at a unique $\bP$, called the Gibbs state and given by
	\begin{equation}
	\label{eq:gibbs}
		\bP_{T, \mu, \Omega} := \frac{1}{Z_{T,\mu,\Omega}} \sum\limits_{n \geq 0} \frac{e^{-\frac{1}{T} \myp{H_n - \mu n}}}{n!} ,\quad 		Z_{T, \mu, \Omega}
		:= \sum\limits_{n \geq 0 } \frac{1}{n!} \int_{\Omega^n} e^{-\frac{1}{T} \myp{H_n \myp{x} - \mu n}} \id x,
	\end{equation}
	where $ H_n \myp{x} := \sum_{i<j}^n w \myp{x_i-x_j} $. The minimal free energy is then given by
\begin{equation*}
G_T(\mu,\Omega)=- T \log Z_{T, \mu, \Omega}.
\end{equation*}

Our goal here is to provide bounds in the case of a large cube $\Omega=C_L=(-L/2,L/2)^d$, which are uniform with respect to the side length $L$. For later purposes it is convenient to divide by the volume and introduce
\begin{equation*}
	\boxed{g_T \myp{\mu,L} := \frac{G_T \myp{\mu, C_L}}{L^d}.}
\end{equation*}
We will also study its Legendre transform which is defined as
\begin{equation}
\label{eq:fl_def}
	\boxed{f_T \myp{\rho,L} := \sup_{\mu \in \R} \Set{\mu \rho + g_T \myp{\mu,L}}.}
\end{equation}
Note that $ f_T \myp{\rho,L} $ is different from the minimal canonical energy per unit volume in $ C_L $. The equivalence of ensembles as in~\eqref{eq:equiv} is only true after taking the limit $L\to\ii$. More precisely, one can see that $f_T(\rho,L)$ is rather the free energy per unit volume of the grand-canonical Gibbs state in $ C_L $ which has average density $ \rho $. Indeed, at positive temperature $T > 0$, the supremum \eqref{eq:fl_def} is attained at the unique chemical potential $ \mu_L = \mu_L \myp{\rho} $ satisfying
\begin{align}
	\rho ={}& - \frac{\partial}{\partial \mu} g_T \myp{\mu_L,L}
	={} \frac{T}{L^d} \frac{\partial}{\partial \mu} \log Z_{T,\mu_L,C_L} \nn \\
	={}& \frac{T}{L^d} \frac{1}{Z_{T,\mu_L,C_L}} \frac{\partial}{\partial \mu} Z_{T,\mu_L,C_L}
	={} \frac{\mathcal{N}(\bP_{T,\mu_L,C_L})}{L^d},
\label{eq:avg_particles}
\end{align}
so the Gibbs state $ \bP_{T,\mu_L,C_L} $ corresponding to this $ \mu_L $ has average density $ \rho $ in $ C_L $.
It now immediately follows that
\begin{equation*}
	f_T \myp{\rho,L}
	={} \mu_L \rho + \frac{G_T \myp{\mu_L, L}}{L^d}
	={} \frac{ \mathcal{G}_T \myp{\bP_{T,\mu_L,C_L}} }{L^d},
\end{equation*}
as claimed.
In fact, $ f_T \myp{\rho,L} $ can also be obtained by minimizing over all states in $ C_L $ with average particle number $ \rho L^d $,
\begin{equation}
\label{eq:fT_variational}
	f_T \myp{\rho,L}
	= \inf_{ \substack{ \supp \bP \subseteq C_L \\ \mathcal{N}(\bP) = \rho L^d }} \frac{\mathcal{G}_T \myp{\bP}}{L^d}.
\end{equation}
Indeed, if $ \bP $ is a minimizer for the right hand side, then
\begin{equation*}
	\mathcal{G}_T \myp{\bP}-\mu_L\cN(\bP)
	={} \mathcal{G}_T \myp{\bP} - \mu_L \rho L^d=L^d\big(f_T(\rho,L)-\mu_L\rho\big)=G_T(\mu_L,L)
\end{equation*}
so by uniqueness of minimizers of $ G_T \myp{\mu_L, L} $ we must have $ \bP = \bP_{T, \mu_L, C_L} $.

Using the equivalence of ensembles \eqref{eq:equiv}, one can see that $ f_T \myp{\rho,L} $ has the infinite volume limit
\begin{equation*}
	\lim_{L \to \infty} f_T \myp{\rho,L}
	= \sup_{\mu \in \R} \myt[\Big]{\mu \rho + \lim_{L \to \infty} g_T \myp{\mu,L} }
	= f_T \myp{\rho}
\end{equation*}
as we have already stated in~\eqref{eq:fcan}. This is well known~\cite{Ruelle} and follows for instance from Corollary~\ref{cor:thermolim} below.

Next we state some bounds which are uniform in $L$ and thus carry over to the infinite volume limit. Loosely speaking, those estimates state that the free energy $f_T(\rho)$ contains a positive term $\rho^\gamma$ which dominates at large densities, due to the repulsive nature of the potential $w$ close to the origin (recall that $\gamma=\max(2,1+\alpha/d)$). At low density, the energy rather behaves like $T\rho\log\rho+C\rho$. Although we think that these bounds must be well known, we have not found them stated explicitly in the literature. The estimates are enclosed in the following three statements.

\begin{proposition}[Fixed average density bounds]
\label{thm:tdlimbounds}
	Let $w$ satisfy Assumption~\ref{de:shortrangenew}.
	There are constants $ C,c > 0 $ depending only on $ s $ and the dimension $ d $, such that for any $ L > 0 $ and $ \rho \geq 0 $, we have
	\begin{align}
	\label{eq:fcanupbound}
		f_T \myp{\rho,L}
		\leq{}&	\begin{cases}
					C \kappa r_0^{d \myp{\gamma-1}} \rho^{\gamma} + C \kappa \rho^2 + C T \rho + T \rho \log \rho,					& \alpha \neq d, \\
					C \kappa r_0^d \rho^2 \myp{\log r_0^d \rho}_+ + C \kappa \myp{1+r_0^d} \rho^2 + C T \rho + T \rho \log \rho,	& \alpha = d, \\
		\end{cases}
	\end{align}
	and
	\begin{equation}
	\label{eq:fcanlowbound}
		f_T \myp{\rho,L}
		\geq	\begin{cases}
					\frac{c}{\kappa} r_0^{d \myp{\gamma-1}} \rho^{\gamma} - \myp[\big]{\kappa + \frac{c}{\kappa} +T}\rho + T \rho \log \rho,	& \alpha \neq d, \\
					\frac{c}{\kappa} r_0^d \rho^2 \myp[\big]{\log \frac{r_0^d \rho}{2 \sqrt{d}}}_+ - \myp{\kappa+T} \rho + T \rho \log \rho,	& \alpha = d. \\
				\end{cases}
	\end{equation}
\end{proposition}

\begin{corollary}[Bounds on the dual variable]
	\label{thm:mubounds}
	Let $w$ satisfy Assumption~\ref{de:shortrangenew}.
	Denote by $ \mu_L \myp{\rho} $ any chemical potential maximizing \eqref{eq:fl_def} (which is unique if $T>0$).
	There are constants $ C,c > 0 $ depending only on $ s $ and the dimension $ d $, such that for all $ L > 0 $ and $ \rho \geq 0 $,
	\begin{equation}
	\label{eq:muupbounds}
		\mu_L \myp{\rho}
		\leq	\begin{cases}
					C \kappa r_0^{d \myp{\gamma-1}} \rho^{\gamma-1} + C \kappa \rho + C \myp{\kappa + \frac{1}{\kappa} + T} + T \log \rho,	& \alpha \neq d, \\
					C \kappa r_0^d \rho \myp{\log r_0^d \rho}_+ + C \kappa \myp{1+ r_0^d} \rho +\kappa + C T + T \log \rho,					& \alpha = d, \\
				\end{cases}
	\end{equation}
	and
	\begin{align}
	\label{eq:mulowbound}
		\mu_L \myp{\rho}
		\geq{}&	\begin{cases}
					\frac{c}{\kappa} r_0^{d \myp{\gamma-1}} \rho^{\gamma-1} - \kappa - \frac{c}{\kappa} -T + T \log \rho,	& \alpha \neq d, \\
					\frac{c}{\kappa} r_0^d \rho \myp[\big]{\log \frac{r_0^d \rho}{2 \sqrt{d}}}_+ - \kappa-T + T \log \rho,	& \alpha = d. \\
				\end{cases}
	\end{align}
\end{corollary}
\begin{corollary}[Grand-canonical bounds]
\label{thm:gcbounds}
	Let $w$ satisfy Assumption~\ref{de:shortrangenew}.
	There are constants $ C $, $ K > 0 $, depending only on $d$ and $w$, such that for any $ \mu \in \R $ and $ L > 0 $,
	\begin{equation}
	\label{eq:fgclowbound}
		g_T \myp{\mu,L}
		\geq{}	\begin{cases}
					- K \myp{\mu+C}_+^{\frac{\gamma}{\gamma-1}} - T e^{- \frac{1}{T} \myp{\mu+C}_-},	& \alpha \neq d, \\[0.2cm]
					- K \frac{\myp{\mu+C}_+^2}{\log(2+\myp{\mu+C}_+)}- T e^{- \frac{1}{T} \myp{\mu+C}_-},	& \alpha = d,
				\end{cases}
	\end{equation}
	and
	\begin{equation}
	\label{eq:fgcupbound}
		g_T \myp{\mu,L}
		\leq{}	\begin{cases}
					- \frac{1}{K(1+T)} \myp{\mu-C}_+^{\frac{\gamma}{\gamma-1}} - T e^{- \frac{1}{T} \myp{\mu-C}_-},	& \alpha \neq d, \\
					- \frac1{K(1+T)} \frac{\myp{\mu-C}_+^2}{\log(2+\myp{\mu-C}_+)}- T e^{- \frac{1}{T} \myp{\mu-C}_-}	& \alpha = d,
				\end{cases}
	\end{equation}
	where $ \gamma = 1+\max \myp{1, \alpha / d} $.
\end{corollary}

We defer the proofs to Section~\ref{sec:tdlim_proofs}.
We will prove Proposition~\ref{thm:tdlimbounds} by first noting that the upper bounds on $ f_T $ follow from \eqref{eq:fT_variational} and the universal bounds \eqref{eq:upper_bd_G_T1} and \eqref{eq:upper_bd_G_T2} (as they are stated in \cite{JexLewMad-23}).
The lower bounds on $ f_T $ we provide by hand.
The bounds on the dual variable $ \mu_L $ follow as an easy corollary, using that $ f_T $ is convex with $ f_T \myp{0,L}=0 $.
We include the bounds on $ g_T $ here for completeness, though we do not need them for our purposes.
They can be obtained using the bounds on $ f_T $ and the fact that $ g_T $ are $ f_T $ are connected through the Legendre transform.
We have gathered all the details in \cref{sec:tdlim_proofs} below.

\section{Proof of Theorem~\ref{thm:thermolim} on the thermodynamic limit at fixed density}
\label{sec:thermolim}

In this section we provide the proof of Theorem~\ref{thm:thermolim} on the thermodynamic limit at fixed constant density (uniform gas), in the canonical case. We will also quickly explain how the grand-canonical case follows from this result, even if we prove more on this case later. To simplify the notation we denote by $\rho:=\rho_0$ the constant density from the statement.

Note first that it suffices to provide an upper bound, since the energy at fixed density $ \rho \1_{\Omega_N} $ is bounded from below by the minimal energy obtained when not putting restrictions on the density, that is,
	\begin{equation*}
		\liminf_{N \to \infty} \frac{F_T \myb{\rho \1_{\Omega_N}} }{\abs{\Omega_N}}
		\geq f_T \myp{\rho}.
	\end{equation*}

	\medskip

\noindent\textbf{A partition of unity.}
	In order to build a suitable trial state with density $ \rho \1_{\Omega_N} $ and compare its energy to $ f_T \myp{\rho} $, we will need to construct a specific partition of unity on $ \R^d $.
	We split $ \R^d $ into a lattice of cubes $ C_L + kL $, $ k \in \Z^d $, where $ C_L = [-L/2, L/2)^d $, and $ L $ is chosen such that $ n:= \rho L^d \in \N $.
	To handle the core of the interaction, we place corridors between the cubes (by shrinking the cubes a little) to keep particles from getting too close to each other.
	That is, we fix a tiny $ \eps > 0 $, and define $ \ell, \lambda \in \myp{0,L} $ by
	\begin{equation*}
		L := \myp[\Big]{\frac{\rho + \eps}{\rho}}^{\frac{1}{d}} \ell,
	\end{equation*}
	and
	\begin{equation*}
		\lambda := L - \ell
		= \myp[\Big]{\myp[\Big]{\frac{\rho + \eps}{\rho}}^{\frac{1}{d}} - 1} \ell.
	\end{equation*}
	Then $ \lambda $ and $ \ell $ satisfy $ L = \ell + \lambda $ and $ n = \rho L^d = \myp{\rho + \eps} \ell^d $.
	Now, let $ \bP_n $ be the minimizer of the unrestricted $ n $-body problem in the slightly smaller cube $ C_{\ell} $, and let $ \rho_n := \rho_{\bP_n} $ be its density. From the lower bound on $w$ in Assumption~\ref{de:shortrangenew}, we know that $\bP_n$ and $\rho_n$ are bounded functions, although $\|\bP_n\|_{L^\ii}$ and $\|\rho_n\|_{L^\ii}$ could in principle depend on $n$. By considering $ \rho_n $ as a function in the larger cube $ C_L $, the family of translated densities $ \rho_n \myp{\nonarg - kL} $, $ k \in \Z^d $, yields a partition of unity at density $ \rho $ when averaged over translations in $ C_L $, i.e.,
	\begin{equation}
	\label{eq:rhores}
		\frac{1}{\abs{C_L}} \int_{C_L} \sum_{k \in \Z^d} \rho_n \myp{\nonarg - kL - \tau} \id \tau
		= \frac{\myp{\rho + \eps} \ell^d}{L^d}
		= \rho,
		\For{a.e. in $ \R^d $}.
	\end{equation}
	Note that the distance between any pair of the smaller cubes $ C_{\ell} + kL $, $ k \in \Z^d $, is at least $ \lambda $, and that $ \lambda $ is of order $ L $.

	\medskip

\noindent\textbf{Constructing the trial state.}
A simple way of constructing a constant density in $\Omega_N$ would be to multiply~\eqref{eq:rhores} by the characteristic function $\1_{\Omega_N}$ and pull it inside the integral. The difficulty here is that $\1_{\Omega_N}\sum_{k \in \Z^d} \rho_n \myp{\nonarg - kL - \tau}$ might not have a constant mass for all $\tau$, which is needed to build a canonical trial state. Our main observation is that the mass is in fact independent of $\tau$ in the particular situation that the domain is a union of cubes of side length $L$ (the mass is in fact constant in each of the initial small cubes of side length $L$). The main idea is thus to first approximate $\Omega_N$ from inside by a union of cubes $\Lambda_N$ as in Figure~\ref{fig:existence}. In the complementary region close to the boundary we just take any reasonable trial state with constant density and use our universal bound~\eqref{eq:upper_bd_F_T}. Inside we take independent copies of a given trial state in each cube. When we translate the tiling, some cubes will leave $\Lambda_N$ but the same amount will penetrate on the other side by periodicity. We simply merge the incomplete cubes with the boundary part.

As announced, we decompose the density $ \rho \1_{\Omega_N} $ into a bulk term consisting of copies of $ \rho_n $ placed on a cubic lattice, and a boundary term giving rise to lower order energy.
We denote by
\begin{equation}
	\Lambda_N := \bigcup_{\substack{k \in \Z^d \\ C_L + k_L \subseteq \Omega_N}} \myp{C_L + kL}
\end{equation}
our interior approximation of $\Omega_N$ by cubes, and for $\tau \in C_L$,
\begin{equation}
	\Lambda_{N,\tau}
	:= \bigcup_{k \in \mathcal{A}_{\tau}} \myp{C_L + kL + \tau},
	\qquad
	\mathcal{A}_{\tau}
	:= \Set{k \in \Z^d \given C_L + kL +\tau \subseteq \Lambda_N},
\end{equation}
the union of cubes contained in $\Lambda_N$ after shifting the lattice of cubes by $\tau$.
Multiplying \eqref{eq:rhores} by $\1_{\Lambda_N}$ allows us to write
	\begin{align*}
		\rho \1_{\Omega_N}
		={}& \frac{1}{\abs{C_L}} \int_{C_L} \sum\limits_{k \in \mathcal{A}_{\tau}} \rho_n \myp{\nonarg-kL-\tau} \id \tau \\
		&+ \frac{1}{\abs{C_L}} \int_{C_L} \underbrace{ \rho \1_{\Omega_N \setminus \Lambda_N} + \1_{\Lambda_N} \sum\limits_{k \in \Z^d \setminus \mathcal{A}_{\tau} } \rho_n \myp{\nonarg-kL-\tau} }_{\dps \qquad =: \beta_{\tau}} \id \tau,
	\end{align*}
see Figure~\ref{fig:existence}.

  \begin{figure}
\begin{tikzpicture}
    \node[anchor=south west,inner sep=0] (image) at (0,0) {\includegraphics[width=.9\textwidth]{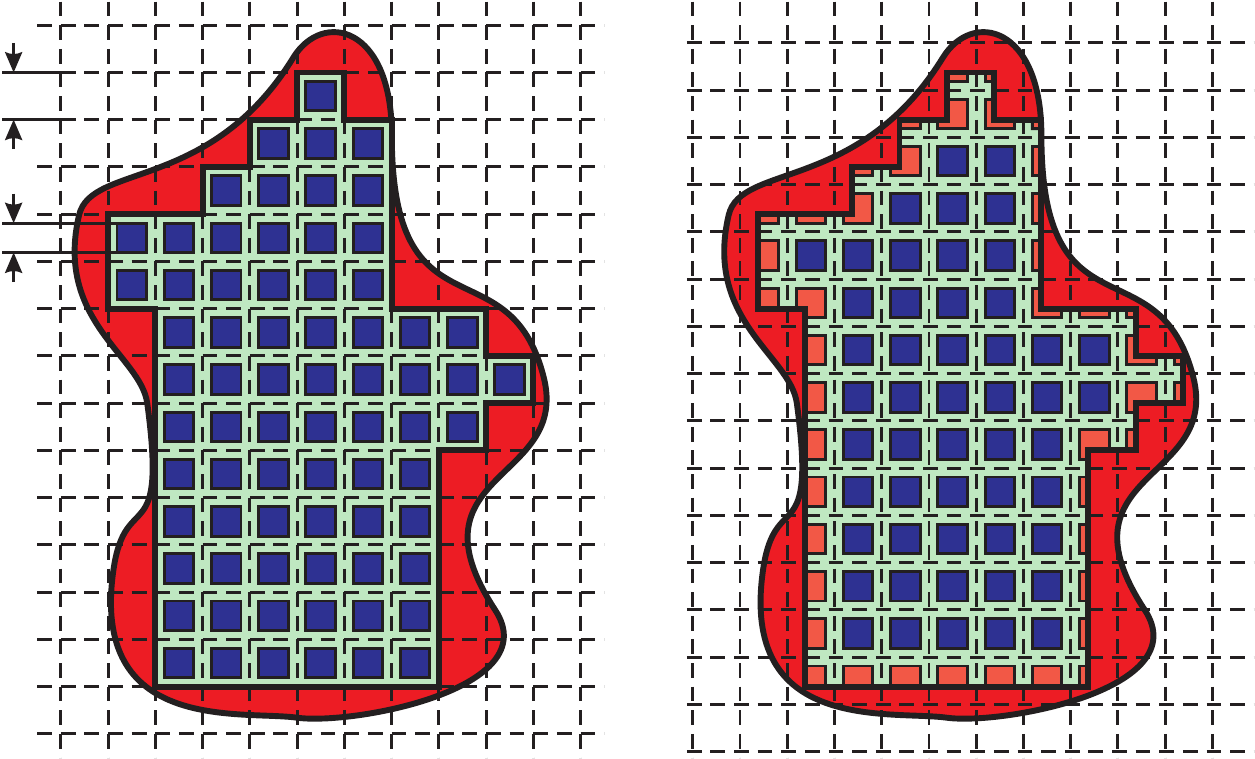}};
    \begin{scope}[x={(image.south east)},y={(image.north west)}]
     \draw (0.012,0.911) node[below] {$L$};
     \draw (-0.005,0.725) node[below] {$\ell$};
    \end{scope}
\end{tikzpicture}
  \caption{The cutting of the set $\Omega_N$ to smaller cubes with corridors on the left and shifted grid on the right. The blue colour corresponds to the cubes in $\Lambda_{N,\tau}$, orange colour to the cubes intersecting $\Omega_N\setminus \Lambda_N$, and the red colour to $ \rho \1_{\Omega_N \setminus \Lambda_N}$. The density $\beta_\tau$ corresponds to the union of the orange cubes and the red region.
  \label{fig:existence}}
  \end{figure}

	Note that because of the corridors between the cubes $ C_{\ell} + kL + \tau $, the distance between the supports of the bulk term $ \sum_{k \in \mathcal{A}_{\tau}} \rho_n \myp{\nonarg - kL - \tau} $ and the boundary term $ \beta_{\tau} $ is at least $ \lambda/2 $.
	Furthermore, we have by the regularity~\eqref{eq:regularity} of $ \Omega_N $,
	\begin{align*}
		\abs{\supp \beta_{\tau}}
		\leq{} \abs{ \Omega_N \setminus \Lambda_{N,\tau}}
		\leq{} \abs{ \Set{x \in \R^d \given \rd \myp{x, \partial \Omega_N} \leq \tfrac{3}{2}  \sqrt{d} L} }
		\leq C\abs{\Omega_N}^{1-\frac1d}L,
	\end{align*}
	in particular, $ \abs{\Lambda_{N,\tau}} = \abs{\Omega_N} + o \myp{\abs{\Omega_N}} $.
	Furthermore, by definition of $ \beta_{\tau} $ we have for any of the translated cubes $ C_L + kL + \tau $,
	\begin{equation}
	\label{eq:betabound}
		\int_{C_L + kL + \tau} \beta_{\tau} \myp{x} \id x
		\leq \int \rho_n \myp{x} + \rho \1_{C_L} \myp{x} \id x
		\leq 2 n.
	\end{equation}

	Now, in each subcube $C_L + kL + \tau$ of $ \Lambda_{N,\tau} $, we place a copy $ \bP_{n,kL+\tau} $ of the $ n $-particle minimizer $ \bP_n $ with density $ \rho_n $, and in the remaining boundary region $ \Omega_N \setminus \Lambda_{N,\tau} $ we simply place any state $ \bP_{\tau} $ with density $ \beta_{\tau} $.
	By our regularity assumptions on $\Omega_N$, it follows from \cite[Lemma $1$]{Fisher-64} that the diameter of $\Omega_N$ is of order $|\Omega_N|^{1/d}$, hence the radius $R(x)$ for $\beta_{\tau}$, defined in~\eqref{eq:Rdef}, satisfies $R(x)\leq C|\Omega_N|^{1/d}	$.
	The free energy of the boundary state $\bP_\tau$ can thus be estimated using our upper bound~\eqref{eq:upper_bd_F_T} by
	\begin{equation}
	\label{eq:trialboundary}
		\mathcal{F}_T \myp{\bP_{\tau}}\leq C \rho L\abs{\Omega_N}^{1-\frac1d} \log|\Omega_N|.
	\end{equation}
	The final trial state $ \bP_N $ on $\Omega_N$ is then defined by
	\begin{equation}
	\label{eq:thermostate}
		\bP_N
		= \frac{1}{\abs{C_L}} \int_{C_L} \widetilde{\bP}_{\tau} \otimes \bP_{\tau} \id \tau
		:= \frac{1}{\abs{C_L}} \int_{C_L} \myp[\bigg]{ \bigotimes_{k \in \mathcal{A}_{\tau}} \bP_{n, kL+\tau} } \otimes \bP_{\tau} \id \tau,
	\end{equation}
	where $ \widetilde{\bP}_{\tau} $ has density
	\begin{equation*}
		\rho_{\widetilde{\bP}_{\tau}} \myp{x} = \sum\limits_{k \in \mathcal{A}_{\tau}} \rho_n \myp{x- kL - \tau}.
	\end{equation*}

	\medskip

\noindent\textbf{Energy bounds.}
	The energy of the state $ \bP_N $ splits into a bulk term, a boundary term, and an interaction term between the two,
	\begin{equation}
	\label{eq:trialenergy}
		\mathcal{F}_T \myp{\bP_N}
		= \frac{1}{\abs{C_L}} \int_{C_L} \mathcal{F}_T \myp{\widetilde{\bP}_\tau} + \mathcal{F}_T \myp{\bP_{\tau}} + 2 D_w \myp{\rho_{\widetilde{\bP}_{\tau}}, \beta_{\tau}} \id \tau,
	\end{equation}
	where
	\begin{equation*}
		D_w \myp{\nu, \mu} := \frac{1}{2} \iint \nu \myp{x} \mu \myp{y} w \myp{x-y} \id x \id y.
	\end{equation*}
	Denoting $ \rho_{C_k} = \rho_n \myp{\nonarg - kL} $, $ k \in \Z^d $, the energy of the bulk term is
	\begin{align}
		\mathcal{F}_T \myp{\widetilde{\bP}_{\tau}}
		={}& \frac{\abs{\Lambda_{N,\tau}}}{L^d} \mathcal{F}_T \myp{\bP_n} + \sum\limits_{\substack{k,k' \in \mathcal{A}_{\tau} \\ k \neq k'}} 2 D_w \myp{\rho_n \myp{\nonarg - kL - \tau}, \rho_n \myp{\nonarg -k'L - \tau}} \nn \\
		={}& \frac{\abs{\Lambda_{N,\tau}}}{\ell^d} \frac{\rho}{\rho + \eps} F_T \myp{n, C_{\ell}} + \sum\limits_{\substack{k,k' \in \mathcal{A}_{\tau} \\ k \neq k'}} \iint \rho_{C_k} \myp{x} \rho_{C_{k'}} \myp{y} w \myp{x-y} \id x \id y,
	\label{eq:trialbulk}
	\end{align}
	where $ F_T \myp{n, C_{\ell}} $ is the minimal energy for $ n $ particles in $ C_{\ell} $, without restrictions on the density.
	To estimate the interaction between the cubes, we introduce the ``tail" of the interaction,
	\begin{equation*}
		\phi \myp{t} := \frac{\kappa}{1+ \lambda^s} \1_{(-\infty,\lambda]} \myp{t} + \frac{\kappa}{1 + \abs{t}^s} \1_{\myp{\lambda,\infty}} \myp{t}
		\qquad t \in \R,
	\end{equation*}
	and the minimal distance between particles in two different cubes
	\begin{equation}
	\label{eq:deltadist}
		\delta_{kk'} := \inf_{\substack{x \in C_{\ell} + kL \\ y \in C_{\ell} + k'L}} \abs{x-y}.
	\end{equation}
	Note that because of the corridors between the cubes, we have
	\begin{equation}
	\label{eq:deltabound}
		\delta_{0k} \geq \lambda + L \myp{\abs{k}-1}
		\geq \lambda + \abs{x} - L \myp{1+ \sqrt{d}}
	\end{equation}
	for all $ k \neq 0 $ and $ x \in C_L + kL $.
	In particular, since we always have $ \delta_{0k} \geq \lambda $, we can estimate
	\begin{align*}
		\MoveEqLeft[6] \sum\limits_{\substack{k,k' \in \mathcal{A}_{\tau} \\ k \neq k'}} \iint \rho_{C_k} \myp{x} \rho_{C_{k'}} \myp{y} w \myp{x-y} \id x \id y \\
		\leq{}& \sum\limits_{k \in \mathcal{A}_{\tau}} \iint \rho_{C_k} \myp{x} \sum\limits_{\substack{k' \in \Z^d \\ k \neq k'}} \rho_{C_{k'}} \myp{y} \phi \myp{\delta_{kk'}} \id x \id y \\
		=& \#\mathcal{A}_{\tau} n^2 \sum\limits_{\substack{k \in \Z^d \\ k \neq 0}} \phi \myp{\delta_{0k}}
		\leq{} \rho^2 \abs{\Lambda_{N,\tau}} L^d \sum\limits_{\substack{k \in \Z^d \\ k \neq 0}} \phi \myp{\delta_{0k}} \\
		\leq{}& \rho^2 \abs{\Lambda_{N,\tau}} \int_{\R^d} \phi \myp[\big]{\abs{x} - L \myp{1+ \sqrt{d}}} \id x,
	\end{align*}
	where we used \eqref{eq:deltabound} in the last inequality.
	We have
	\begin{align*}
		\int_{\R^d} \phi \myp[\big]{\abs{x} - L \myp{1+ \sqrt{d}}} \id x
		={}& \frac{\abs{\mathbb{S}^{d-1}}}{d} \phi \myp{0} L^d \myp{1+\sqrt{d}}^d \\
		&+ \abs{\mathbb{S}^{d-1}} \int_0^{\infty} \phi \myp{r} \myp[\big]{r + L \myp{1+\sqrt{d}}}^{d-1} \id r.
	\end{align*}
	Splitting the integral at $ r = L \myp{1+ \sqrt{d}} $ gives
	\begin{align*}
		\MoveEqLeft[4] \abs{\mathbb{S}^{d-1}} \int_0^{\infty} \phi \myp{r} \myp[\big]{r + L \myp{1+\sqrt{d}}}^{d-1} \id r \\
		\leq{}& \frac{\abs{\mathbb{S}^{d-1}}}{d} \phi \myp{0} \myp{2^d -1} L^d \myp{1+\sqrt{d}}^d + 2^{d-1} \int_{\abs{x} \geq L \myp{1+ \sqrt{d}}} \phi \myp{ \abs{x}} \id x,
	\end{align*}
	where for instance we can bound
	\begin{equation*}
		\int_{\abs{x} \geq L \myp{1+ \sqrt{d}}} \phi \leq \kappa \abs{\mathbb{S}^{d-1}} \int_{L \myp{1+\sqrt{d}}}^{\infty} r^{d-s-1} \id r = \frac{\kappa \abs{\mathbb{S}^{d-1}}}{\myp{s-d} \myp{L \myp{1+\sqrt{d}}}^{s-d}}.
	\end{equation*}
	Collecting the bounds, we have shown
	\begin{align}
		\MoveEqLeft[4] \int_{\R^d} \phi \myp[\big]{\abs{x} - L \myp{1+ \sqrt{d}}} \id x \nn \\
		\leq{}& 2^{d-1} \kappa \abs{\mathbb{S}^{d-1}} \myp[\bigg]{ \frac{2}{d} \frac{\myp{L \myp{1+\sqrt{d}}}^d}{1 + \lambda^s} + \frac{1}{\myp{s-d} \myp{L \myp{1+\sqrt{d}}}^{s-d}}},
	\label{eq:w2bound}
	\end{align}
	and hence the interaction energy between the cubes in the bulk of $ \Omega_N $ is bounded by
	\begin{align}
		\MoveEqLeft[4] \frac{1}{\abs{\Omega_N}} \sum\limits_{\substack{k,k' \in \mathcal{A}_{\tau} \\ k \neq k'}} \iint \rho_{C_k} \myp{x} \rho_{C_{k'}} \myp{y} w \myp{x-y} \id x \id y \nn \\
		\leq{}& \rho^2 2^{d-1} \kappa \abs{\mathbb{S}^{d-1}} \myp[\bigg]{ \frac{2}{d} \frac{\myp{L \myp{1+\sqrt{d}}}^d}{1 + \lambda^s} + \frac{1}{\myp{s-d} \myp{L \myp{1+\sqrt{d}}}^{s-d}}},
	\label{eq:bulkbound}
	\end{align}
	where the right hand side tends to zero when $ L \to \infty $, since $ \lambda $ is of order $ L $, and $ s > d $.

	What remains now is to estimate the interaction $ D_w \myp{\rho_{\widetilde{\bP}_{\tau}}, \beta_{\tau}} $ between the bulk and the boundary in \eqref{eq:trialenergy}.
	Denote
	\begin{equation*}
		\mathcal{B}_{N,\tau} := \Set{k \in \Z^d \given \myp{C_L + kL + \tau} \cap \myp{\Omega_N \setminus \Lambda_N} \neq \emptyset }.
	\end{equation*}
	Then $ \mathcal{A}_{N,\tau} \cap \mathcal{B}_{N,\tau} = \emptyset $, and since the cubes $ \myp{C_L + kL + \tau}_{k \in \mathcal{A}_{N,\tau} \cup \mathcal{B}_{N,\tau}} $ by construction constitute a cover of $ \Omega_N $ for any $ \tau \in C_L $, we have
	\begin{equation*}
		\Omega_N \setminus \Lambda_{N,\tau}
		\subseteq \bigcup\limits_{k \in \mathcal{B}_{N,\tau}} \myp{C_L + kL + \tau},
	\end{equation*}
	where by regularity of $ \Omega_N $,
	\begin{equation*}
		\abs[\Big]{\bigcup\limits_{k \in \mathcal{B}_{N,\tau}} \myp{C_L + kL}}
		\leq \abs{ \Set{ x \in \R^d \given \rd \myp{x, \partial \Omega_N} \leq \tfrac{3}{2} \sqrt{d} L } }
		\leq C\abs{\Omega_N}^{1-\frac1d}  L.
	\end{equation*}
	To estimate $ D_w \myp{\rho_{\widetilde{\bP}_{\tau}}, \beta_{\tau}} $, we will use the slightly smaller distance
	\begin{equation*}
		{\widetilde{\delta}}_{kk'}
		:={} \inf_{\substack{x \in C_{\ell} + kL + \tau \\ y \in C_L + k'L + \tau }} \abs{x-y}
		\leq{} \inf_{\substack{x \in C_{\ell} + kL + \tau \\ y \in \supp \beta_\tau \cap \myp{C_L + k'L + \tau} }} \abs{x-y},
	\end{equation*}
	for $ k \in \mathcal{A}_{N,\tau} $ and $ k' \in \mathcal{B}_{N,\tau} $, because of the presence of the term $ \rho \1_{\Omega_N \setminus \Lambda_N} $ in the definition of $ \beta_{\tau} $.
	Note that in spite of this, the distance between the supports of the bulk density $ \rho_{\widetilde{\bP}_\tau} $ and $ \beta_{\tau} $ is still at least $ \lambda/2 $, and we have the bound $ \widetilde{\delta}_{0k} \geq \abs{x} - L \myp{1+ \sqrt{d}} $.
	Hence we can estimate the interaction term in \eqref{eq:trialenergy},
	\begin{align}
		2D_w \myp{\rho_{\widetilde{\bP}_\tau}, \beta_{\tau}}
		={}& \sum\limits_{\substack{k' \in \mathcal{B}_{N,\tau} \\ k \in \mathcal{A}_{N,\tau}}} \int_{C_L + k'L + \tau} \beta_{\tau} \myp{y} \int \rho_n \myp{x - kL - \tau} w \myp{x-y} \id x \id y \nn \\
		\leq{}& \sum\limits_{k' \in \mathcal{B}_{N,\tau}} \int_{C_L + k'L + \tau} \beta_{\tau} \myp{y} \id y \sum\limits_{\substack{k \in \Z^d \\ k \neq k'}} n \phi \myp{ \widetilde{\delta}_{kk'} } \nn \\
		\leq{}& \# \mathcal{B}_{N,\tau} 2 n^2 \sum\limits_{\substack{k \in \Z^d \\ k \neq 0}} \phi \myp{\widetilde{\delta}_{0k}} \nn \\
		\leq{}& C \rho^2 \abs{\Omega_N}^{1-\frac1d} L \int_{\R^d} \phi \myp[\big]{\abs{x} - L \myp{1+ \sqrt{d}}} \id x,
	\label{eq:trialinteract}
	\end{align}
	where we also used \eqref{eq:betabound}.
	The integral on the right hand side is bounded by \eqref{eq:w2bound}, so this interaction term is of lower order than $ \abs{\Omega_N} $.
	Finally, combining \eqref{eq:trialenergy} with \eqref{eq:trialbulk}, \eqref{eq:trialboundary}, \eqref{eq:bulkbound}, and \eqref{eq:trialinteract}, we conclude that
	\begin{align*}
		\lim_{L \to \infty} \limsup_{N \to \infty} \frac{\mathcal{F}_T \myp{\bP_N}}{\abs{\Omega_N}}
		\leq{}& \lim_{L \to \infty} \frac{\rho}{\rho + \eps} \frac{F_T \myp{n, C_{\ell}}}{\ell^d} + o\myp{1}_{L \to \infty} \\
		={}& \frac{\rho}{\rho + \eps} f_T \myp{\rho + \eps}.
	\end{align*}
	Since $ \eps $ is arbitrary, and $ f_T $ is continuous, this concludes the proof of Theorem~\ref{thm:thermolim}.\qed

	\bigskip

The limit in the canonical case implies the same in the grand-canonical case. We state it here, although we will prove more later.

\begin{corollary}
\label{cor:thermolim}
	Under the same assumptions as \cref{thm:thermolim}, the corresponding grand-canonical thermodynamic limit exists, and is the same as in the canonical case,
	\begin{equation}
		\lim_{N \to \infty} \frac{G_T \myb{\rho \1_{\Omega_N}}}{\abs{\Omega_N}}
		= f_T \myp{\rho}.
	\end{equation}
	Here, we do not need to assume that $ \rho \abs{\Omega_N} $ is an integer.
\end{corollary}

\begin{proof}
	For the upper bound, we can use the same proof as in \cref{thm:thermolim} with only minor modifications.
	Because the particle number is now not fixed, there is no need to put any restrictions on the side length $ L $ of the cubic lattice covering $ \R^d $.
	Instead of using $ \rho_n $ to construct the partition of unity with corridors in \eqref{eq:rhores}, we take the density of the grand-canonical Gibbs state $ \bP_{T, \mu_{\ell}\myp{\rho+\eps},\ell} $ in $ C_{\ell} $, and also use this Gibbs state in the construction of the trial state \eqref{eq:thermostate}.
	Here, the chemical potential $ \mu_{\ell} \myp{\rho+\eps} $ is chosen to maximize \eqref{eq:fl_def} at density $ \rho+\eps $ and side length $ \ell $, in particular,
	\begin{equation*}
		\mathcal{G}_T \myp{\bP_{T, \mu_{\ell}\myp{\rho+\eps},\ell}} = \ell^d f_T \myp{\rho+\eps,\ell},
	\end{equation*}
	as explained in \eqref{eq:avg_particles}-\eqref{eq:fT_variational}.
	Going through the proof exactly as before now yields
	\begin{equation*}
		\limsup_{N \to \infty} \frac{G_T \myb{\rho \1_{\Omega_N}}}{\abs{\Omega_N}}
		\leq{} \lim_{\ell \to \infty} \frac{\rho}{\rho+\eps} \frac{f_T \myp{\rho+\eps,\ell}}{\ell^d} + o\myp{1}_{\ell \to \infty}
		= \frac{\rho}{\rho+\eps} f_T \myp{\rho+\eps}
	\end{equation*}
	for any sufficiently small $ \eps > 0 $.

	For the corresponding lower bound on $ G_T \myb{\rho \1_{\Omega_N}} $, we introduce the chemical potential $ \mu=f'_T(\rho)$ and note that for any grand-canonical state $ \bP $ with $ \rho_{\bP} = \rho \1_{\Omega_N} $, we have
	\begin{equation*}
		\mathcal{G}_T \myp{\bP}
		= \mathcal{G}_T\myp{\bP}-\mu\cN(\bP) + \mu \rho \abs{\Omega_N}.
	\end{equation*}
	It follows that
	\begin{equation*}
		\liminf_{N \to \infty} \frac{G_T \myb{\rho \1_{\Omega_N}}}{\abs{\Omega_N}}
		\geq{} \liminf_{N \to \infty} \frac{G_T \myp{\mu, \Omega_N}}{\abs{\Omega_N}} + \mu \rho
		={} g_T \myp{\mu} + \mu \rho=f_T(\rho),
	\end{equation*}
	where we recall $ g_T $ is the usual grand-canonical thermodynamic limit \eqref{eq:fgc}. The last equality is due to our choice $\mu=f'_T(\rho)$.
\end{proof}

\begin{remark}[Hard-core]
Corollary~\ref{cor:thermolim} is also valid in the hard-core case $\alpha=+\ii$, under the additional assumption that $ \rho_0 <\rho_c(d) r_0^{-d} $, where $\rho_c(d)$ is the packing density. Due to the corridors we slightly increase the density in the small boxes and therefore need the strict inequality. In the grand-canonical case we do not need to separate the boundary part as we did in Figure~\ref{fig:existence} and the proof goes as in~\cite[Thm.~14]{JexLewMad-23}.
To make the proof of \cref{thm:thermolim} work for hard cores in the canonical case, we would need an upper bound on the free energy of the boundary part, but it is not clear if the corresponding density is representable (see the discussion in~\cite[Sec.~2.4]{JexLewMad-23}.
\end{remark}

\section{Grand-canonical convergence rates}
\label{sec:gcconvrate}
The aim of this section is to provide \emph{a priori} bounds for the convergence rate of the usual grand-canonical free energy for cubes in the thermodynamic limit. That is, for cubes $ C_L = [ -L/2 , L/2 )^d $, temperature $ T \geq 0 $, and any chemical potential $ \mu \in \R $, we estimate
\begin{equation*}
	\frac{G_T \myp{\mu, C_L}}{L^d}-g_T \myp{\mu}.
\end{equation*}

A central tool for controlling error terms is a bound due to Ruelle \cite{Ruelle-70} which allows one to uniformly control the local average (square) number of particles $ \expec[\big]{n_Q^2}_{T,\mu,\Omega} $ in a cube $ Q $, for a Gibbs state in the set $ \Omega \subseteq \R^d $.
One version of the bound takes the form
\begin{equation}
\label{eq:ruelle}
	\boxed{\expec{n_Q}_{T, \mu, \Omega}\leq \expec[\big]{n_Q^2}_{T, \mu, \Omega}^{1/2}
	\leq  \abs{Q} \xi_{T,\mu},}
\end{equation}
for all cubes $Q$ of side length at least $L_0$ (which is independent of $\mu$ and $T$), and where
\begin{equation}
\xi_{T,\mu}=\begin{cases}
C_Te^{\frac{\mu}{2T}} \myp[\big]{1+e^{\frac{\mu d}{2T\eps}} } &\text{for $T>0$,}\\
C_0(C_0+\mu)_+^{1+\frac{d}{\eps}}&\text{for $T=0$.}
              \end{cases}
 \label{eq:def_xi}
\end{equation}
We recall that $\pscal{\nonarg}_{T,\mu,\Omega}$ denotes the expectation against the Gibbs state at temperature $T$, chemical potential $\mu$ in a domain $\Omega$. At $T=0$ we just take any minimizer $\{\bar x_1,...,\bar x_N\}$ of the free energy
$$\min_{n\geq0}\min_{x_1,...,x_n\in\Omega} \myt[\bigg]{\sum_{1\leq j<k\leq n}w(x_j-x_k)-\mu n }$$
and $\expec[\big]{n_Q^2}_{0, \mu, \Omega}:=\#\{\bar x_j\in Q\}^2$ is simply the square of the number of points in $Q$.
In~\eqref{eq:def_xi}, $\eps=\min(1,s-d)/2$ and the constant $ C_T $ depends only on $ T $ and the interaction $ w $ between the particles, and not on the cube $ Q $, or the larger domain $ \Omega $. The bound~\eqref{eq:ruelle} is well known at low activity $z=e^{\mu/T}\ll1$ and can be found in~\cite[Chap.~4]{Ruelle}. The bound for larger activities is given in Appendix~\ref{app:Ruelle}, in Corollary~\ref{cor:Ruelle_polynomial} for $T>0$ and Corollary~\ref{cor:T0} for $T=0$.

We express the convergence rate in two different ways. The following proposition provides an additive error term. Later in Corollary~\ref{cor:gcconvrate_mu} we instead express the rate using a shift of the chemical potential $\mu$, which turns out to be convenient for our purposes.

\begin{proposition}[Grand-canonical convergence rate]
\label{prop:gcconvrate}
	Let $ T\geq  0 $ and $ \mu \in \R $, and suppose that the interaction $ w $ satisfies the conditions in \cref{de:shortrangenew}.
	Then there exists a constant $ c > 0 $ depending only on $ d$, $ T $, and the interaction $ w $, such that
	\begin{equation}
	\label{eq:gcconvrate}
		\left|\frac{G_T \myp{\mu,C_{\ell}}}{\ell^d}-g_T \myp{\mu}\right|
		\leq{}  c\, \xi_{T,\mu}^2\eps_\ell
	\end{equation}
	for any $ \ell > 0 $ sufficiently large (independently of $\mu$ and $T$), where
\begin{equation}
\eps_\ell:= \begin{cases}
		\ell^{-1}&\text{if $s>d+1$,}\\
		\ell^{-1}\log\ell&\text{if $s=d+1$,}\\
		\ell^{d-s}&\text{if $d<s<d+1$.}
			\end{cases}
\label{eq:eps_ell}
\end{equation}
\end{proposition}

The inequality~\eqref{eq:gcconvrate} provides the expected error bound on $G_T(\mu,C_\ell)$ in terms of the surface area, when $w$ decays fast enough, that is, $s>d+1$. When $d<s\leq d+1$ the bound gets worse. The proof will be based on the following simple lemma, which will be used to estimate the interaction between subsystems.

\begin{lemma}
	\label{lem:integralbound}
	For any $ s > d $, there is a constant $ c > 0 $, depending only on $ s $ and $ d $, such that for any $ \ell \geq2 $ and $ x_0 \in C_{\ell} $,
	\begin{equation}
		\int_{ \myp{C_{\ell} }^c} \frac{\id y}{1 + \abs{x_0-y}^s}
		\leq \frac{c}{1+\rd \myp{x_0, \partial C_{\ell}}^{s-d}},
		\label{eq:bound_interaction_cube_outside}
	\end{equation}
	and
	\begin{equation*}
		\int_{C_{\ell}} \int_{ \myp{C_{\ell} }^c} \frac{\id y}{1 + \abs{x-y}^s}\id x
		\leq c\ell^d\eps_\ell,
	\end{equation*}
	with $\eps_\ell$ as in~\eqref{eq:eps_ell}.
\end{lemma}

\begin{proof}
	We denote $ \eta = \rd \myp{x_0, \partial C_{\ell}} $ and assume first that $\eta\leq1$. Then we simply bound
$$\int_{ \myp{C_{\ell}}^c } \frac{\id y}{1 + \abs{x_0-y}^s}\leq \int_{\R^d} \frac{\id y}{1 + \abs{y}^s}\leq \frac{2}{1+\eta^{s-d}}\int_{\R^d} \frac{\id y}{1 + \abs{y}^s}. $$
When $\eta\geq1$, we note that $ \myp{C_{\ell}}^c + x_0 \subseteq \myp{C_{\eta}}^c \subseteq \Set{\abs{y} \geq \frac{\eta}{2}} $, so
	\begin{align*}
		\int_{ \myp{C_{\ell}}^c } \frac{\id y}{1 + \abs{x_0-y}^s}
		={}& \int_{ \myp{C_{\ell}}^c + x_0 } \frac{\id y}{1 + \abs{y}^s}
		\leq{} \int_{ \abs{y} \geq \frac{\eta}{2}} \frac{\id y}{\abs{y}^s}  \\
		={}& \abs{\mathbb{S}^{d-1}} \int_{\frac{\eta}{2}}^{\infty} \frac{\id r}{r^{s-d+1}}
		={} \frac{c_1}{\eta^{s-d}}\leq \frac{2c_1}{1+\eta^{s-d}},
	\end{align*}
	where the last bound is because $\eta\geq1$.

To prove the second inequality, we write the boundary as a union of faces $\partial C_\ell=\cup_{j=1}^{2d} F_j$ and we use the previous bound to get
$$\int_{C_{\ell}} \int_{ \myp{C_{\ell} }^c} \frac{\id y}{1 + \abs{x-y}^s}\id x\leq c\sum_{j=1}^{2d}\int_{C_{\ell}}\frac{\id x}{1+\rd (x,F_j)^{s-d}}=2cd\ell^{d-1}\int_0^\ell\frac{\id x_1}{1+x_1^{s-d}}.$$
The last equality is because the integral involving $F_j$ is independent of the face, by symmetry. The last integral is bounded for $s>d+1$, behaves as $\log\ell$ for $s=d+1$ and as $\ell^{d+1-s}$ for $d<s<d+1$, hence the result follows.
\end{proof}

We are now ready to provide the following.

\begin{proof}[Proof of \cref{prop:gcconvrate}]
	\textbf{Lower bound on $ G_T $.}
	Consider a large cube $ C_L = [ -L/2, L/2 )^d $ of side length $ L > 0 $.
	Write $ C_L $ as a union of smaller cubes $ C_{\widetilde{\ell}} $,
	\begin{equation*}
		C_L = \bigcup_{k \in \mathcal{A}} C_{\widetilde{\ell}} + \widetilde{\ell} k,
	\end{equation*}
	for some appropriate subset $ \mathcal{A} \subseteq \Z^d $.
	We insert corridors of size $ \delta > 0 $ between the smaller cubes by shrinking them a little and instead considering the union of the cubes $ C_{\ell}^k := C_{\ell} + \widetilde{\ell} k $, where $ \ell = \widetilde{\ell} - \delta $.
	In each of the smaller cubes $ C_{\ell}^k $ we place the usual Gibbs state $ \bP_{T,\mu,C_{\ell}^k} $ and consider in the large cube $ C_L $ the (grand-canonical) tensor product
	\begin{equation*}
		\bP_L := \bigotimes_{k \in \mathcal{A}} \bP_{T,\mu,C_{\ell}^k}.
	\end{equation*}
	We use $\bP_L$ as a trial state for the grand-canonical problem in $C_L$ at the same chemical potential $\mu$. The free energy then satisfies
	\begin{align}
		\mathcal{G}_T^{\mu} \myp{\bP_L}
		={}& \sum\limits_{k \in \mathcal{A}} G_T \myp{\mu, C_{\ell}} + \sum\limits_{\substack{k,m\in \mathcal{A} \\k \neq m}} 2 D_w \myp{\rho_{\ell,k} , \rho_{\ell,m}}  \nn \\
		\leq{}& \frac{L^d}{\myp{\ell+\delta}^d} G_T \myp{\mu, C_{\ell}} + \frac{L^d}{\myp{\ell+\delta}^d} \sum\limits_{\substack{k \in \Z^d \\ k \neq 0}} 2 \abs{ D_{w} \myp{\rho_{\ell,0}, \rho_{\ell,k}} },
	\label{eq:gcrate1}
	\end{align}
	where $ \rho_{\ell,k} = \rho_{\ell} \myp{\cdot - \widetilde{\ell} k} $ denotes the one-body density of the Gibbs state $ \bP_{T,\mu,C_{\ell}} $ placed in the cube $ C_{\ell}^k $.
	We will proceed to estimate the interaction term above, and then take $ L $ to infinity at fixed $ \ell $ to obtain a lower bound on $ G_T \myp{\mu, C_{\ell}} $.

	By taking $ \delta \geq r_0 $, none of the error terms see the core of $ w $, so by dividing each of the small cubes into even smaller cubes of side length $ r \sim L_0 $, i.e.,
	\begin{equation*}
		C_{\ell}^k
		= \bigcup_{m \in \widetilde{\mathcal{A}}} C_r^{k,m}
		:= \bigcup_{m \in \widetilde{\mathcal{A}}} \myp{C_r + r m + \widetilde{\ell} k},
	\end{equation*}
	for some appropriate subset $ \widetilde{\mathcal{A}} \subseteq \Z^d $, we obtain
	\begin{align}
		2 \abs{ D_{w} \myp{\rho_{\ell,0}, \rho_{\ell,k}} }
		\leq{}& \iint \frac{\kappa}{1+ \abs{x-y}^s} \rho_{\ell,0} \myp{x} \rho_{\ell, k} \myp{y} \id x \id y \nn \\
		\leq{}& \sum\limits_{j,m \in \widetilde{\mathcal{A}}} \frac{\kappa}{1+ \rd \myp[\big]{C_r^{0,j}, C_r^{k,m}}^s } \int_{C_r^{0,j}} \rho_{\ell,0} \int_{C_r^{k,m}} \rho_{\ell,k} \nn \\
		\leq{}& C r^d \xi_{T,\mu}^2 \sum\limits_{j,m \in \widetilde{\mathcal{A}}} \frac{1}{1+ \rd \myp[\big]{C_r^{0,j}, C_r^{k,m}}^s },
	\label{eq:ruelle_interaction}
	\end{align}
	from the Ruelle bound \eqref{eq:ruelle}.
	Note that we can choose a constant $ \widetilde{c} $ such that for all $ t \geq 0$,
	\begin{equation*}
		1+\myp{t + 2 \sqrt{d} r}^s
		\leq \tilde c(1+t^s),
	\end{equation*}
	and denote by $ C_1 $, $ C_2 $ any two open, disjoint cubes of side length $ r $.
	Then, since
	\begin{equation*}
		\max_{x \in C_1, y \in C_2} \abs{x-y}
		\leq \rd \myp{C_1, C_2} + 2 \sqrt{d} r,
	\end{equation*}
	it follows that (by taking $ t = \rd \myp{C_1, C_2} $),
	\begin{equation}
	\label{eq:integralbound2}
		\frac{1}{1 + \rd \myp{C_1, C_2}^s}
		\leq \frac{\widetilde{c}}{1+ \max\limits_{x \in C_1 ,y \in C_2} \abs{x-y}^s}
		\leq \frac{\tilde c}{r^{2d}} \int_{C_1} \int_{C_2} \frac{\id x \id y}{1 + \abs{x-y}^s}.
	\end{equation}
Continuing \eqref{eq:ruelle_interaction} and summing over $ k $, we now obtain
\begin{equation}
\sum\limits_{j,m \in \widetilde{\mathcal{A}}} \frac{1}{1+ \rd \myp[\big]{C_r^{0,j}, C_r^{k,m}}^s }\leq \frac{\tilde c}{r^{2d}}\int_{C^0_\ell}\int_{C^k_\ell}\frac{\id x \id y}{1 + \abs{x-y}^s}.
\label{eq:interactbound1}
\end{equation}
Summing finally then over $k$, we find for the interaction between the cube $C^0_\ell$ and the rest of the system
\begin{equation}
\sum_{\substack{k \in \Z^d \\ k \neq 0}} 2 \abs{ D_{w} \myp{\rho_{\ell,0}, \rho_{\ell,k}}} \leq C \xi_{T,\mu}^2\int_{C^0_\ell}\int_{(C^0_\ell)^c}\frac{\id x \id y}{1 + \abs{x-y}^s}\leq C \xi_{T,\mu}^2\ell^d\eps_\ell,
\label{eq:estim_int_rho}
\end{equation}
by \cref{lem:integralbound}.
After taking $L\to\ii$ in~\eqref{eq:gcrate1} we end up with
	$$(1+\delta/\ell)^d g_T \myp{\mu}\leq{} \frac{G_T \myp{\mu,C_{\ell}}}{\ell^d}+C\xi_{T,\mu}^2\eps_\ell.$$
The stability of $w$ and the free energy of the non-interacting gas imply the simple lower bound
$$0\geq g_T \myp{\mu}\geq -Te^{\frac1T(\mu+\kappa)}\geq -C_T\xi_{\mu,T}^2.$$
from which we deduce that the correction $\delta |g_T(\mu)|/\ell$ can be absorbed into $C\xi_{T,\mu}^2\eps_\ell$. At $T=0$ we also have $g_0(\mu)\geq -C\xi_{0,\mu}^2$. Hence we obtain
$$g_T \myp{\mu}-C\xi_{T,\mu}^2\eps_\ell\leq{} \frac{G_T \myp{\mu,C_{\ell}}}{\ell^d}$$
as we wanted.

\medskip

	\noindent\textbf{Upper bound on $ G_T $.}
	We consider again a large cube $ C_L = [ -L/2; L/2 )^d $ and divide it into a union of smaller cubes $ C_{\ell} $,
	\begin{equation*}
		C_L = \bigcup_{k \in \mathcal{A}} C_{\ell}^{k}
		= \bigcup_{k\in \mathcal{A}} \myp{ C_{\ell} + \ell k }.
	\end{equation*}
	For a state $ \bP $ on $ C_L $ we denote by $ \bP_k := \bP_{\1_{C_{\ell}^k}} $ its geometric localization to the cube $ C_{\ell}^k $  (see \cite[Appendix A]{HaiLewSol_2-09} and \cite{Lewin-11}). Note that the expected number of particles behaves additively,
	\begin{equation*}
		\mathcal{N}(\bP)
		= \int \rho_{\bP}
		= \sum\limits_{k \in \mathcal{A}} \int_{C_{\ell}^k} \rho_{\bP_k}
		= \sum\limits_{k \in \mathcal{A}} \mathcal{N} \myp{\bP_k},
	\end{equation*}
	and the interaction energy of $ \bP $ splits into local terms and cross terms,
	\begin{align*}
		\mathcal{U} \myp{\bP}
		={}& \frac{1}{2} \iint w \myp{x-y} \rho_{\bP}^{\myp{2}} \myp{x,y} \id x \id y \\
		={}& \sum\limits_{k,m \in \mathcal{A}} \frac{1}{2} \iint w \myp{x-y} \1_{C_{\ell}^k} \myp{x} \1_{C_{\ell}^m} \myp{y} \rho_{\bP}^{\myp{2}} \id x \id y \\
		={}& \sum\limits_{k \in \mathcal{A}} \mathcal{U} \myp{\bP_k} + \sum\limits_{\substack{k,m \in \mathcal{A} \\ k \neq m}} \frac{1}{2} \iint_{C_{\ell}^k \times C_{\ell}^m} w \myp{x-y} \rho_{\bP}^{\myp{2}} \myp{x,y} \id x \id y.
	\end{align*}
	By the sub-additivity property of the entropy for $T>0$, and taking $ \bP $ to be the Gibbs state in $ C_L $, we conclude that
	\begin{equation}
		G_T \myp{\mu, C_L}
		\geq{} L^d \frac{G_T \myp{\mu, C_{\ell}}}{\ell^d} + \sum\limits_{\substack{k,m \in \mathcal{A} \\ k \neq m}} \expec{I_{k,m}}_{\bP},
	\label{eq:geolowbound}
	\end{equation}
	with
	\begin{equation*}
		\expec{I_{k,m}}_{\bP}
		:={} \frac{1}{2} \iint w \myp{x-y} \1_{C_{\ell}^k} \myp{x} \1_{C_{\ell}^{m}} \myp{y} \rho_{\bP}^{\myp{2}} \myp{x,y} \id x \id y.
	\end{equation*}
	Hence, it only remains to provide a suitable lower bound on the sum of cross terms $ \expec{I_{k,m}}_{\bP} $ in the thermodynamic limit.
	For this, we again divide the cubes $ C_{\ell}^k $ into smaller cubes of fixed side length $ r $,
	\begin{equation*}
		C_{\ell}^k
		= \bigcup_{\gamma \in \widetilde{\mathcal{A}}} C_{r}^{k,\gamma}
		= \bigcup_{\gamma \in \widetilde{\mathcal{A}}} \myp{ C_{r} + r \gamma + \ell k },
	\end{equation*}
	and recall that we can bound the interaction from below by $ w \myp{x} \geq - \frac{\kappa}{1 + \abs{x}^s} $, so that
	\begin{align}
		- 2 \expec{I_{k,m}}_{\bP}
		\leq{}& \sum\limits_{\gamma, \gamma' \in \widetilde{\mathcal{A}}}  \pscal{ \sum\limits_{i < j} \frac{\kappa}{1+\abs{x_i-x_j}^s} \1_{C_r^{k,\gamma} } \myp{x_i} \1_{C_r^{m, \gamma'}} \myp{x_j}}_{T,\mu,C_L} \nn \\
		\leq{}& \sum\limits_{\gamma, \gamma' \in \widetilde{\mathcal{A}}} \frac{\kappa}{1+ \rd \myp[\big]{C_r^{k,\gamma}, C_r^{m,\gamma'}}^s } \expec[\big]{ n_{C_r^{k,\gamma}}n_{C_r^{m,\gamma'}} }_{T,\mu,C_L} \nn \\
		\leq{}& \sum\limits_{\gamma, \gamma' \in \widetilde{\mathcal{A}}} \frac{\kappa}{1+ \rd \myp[\big]{C_r^{k,\gamma}, C_r^{m,\gamma'}}^s } \expec[\big]{n_{C_r^{k,\gamma}}^2}_{T, \mu, C_L}^{\frac{1}{2}} \expec[\big]{n_{C_r^{m,\gamma'}}^2}_{T, \mu, C_L}^{\frac{1}{2}} \nn \\
		\leq{}& C \xi_{T,\mu}^2 r^{2d} \sum\limits_{\gamma, \gamma' \in \widetilde{\mathcal{A}}} \frac{1}{1+ \rd \myp[\big]{C_r^{k,\gamma}, C_r^{m,\gamma'}}^s },
	\label{eq:interactbound}
	\end{align}
	where we have again used the Ruelle bound \eqref{eq:ruelle}. By~\eqref{eq:interactbound1} and Lemma~\ref{lem:integralbound}, we obtain
$$\sum\limits_{m \in \mathcal{A}\setminus\{k\}} \expec{I_{k,m}}_{\bP}\geq{} - C \xi_{T,\mu}^2\ell^d\eps_\ell$$
where $C$ now contains the factor $r^{2d}$. Finally, returning to \eqref{eq:geolowbound}, we can now divide by $ L^d $ and let $ L \to \infty $ to conclude that
	\begin{equation*}
		\frac{G_T \myp{\mu, C_{\ell}}}{\ell^d}
		\leq{} g_T \myp{\mu} +C \xi_{T,\mu}^2 \eps_\ell.
	\end{equation*}
	which concludes the proof of Proposition~\ref{prop:gcconvrate}.
\end{proof}

In our setting, a slightly different way of stating the convergence rate will be useful. Instead of shifting the total free energy, we modify the chemical potential $\mu$. This works well at $T>0$ but we get a slightly worse lower bound for $T=0$, which is probably an artefact of our proof.

\begin{corollary}[Shifting the chemical potential]\label{cor:gcconvrate_mu}
Under the same assumptions as in Proposition~\ref{prop:gcconvrate}, we have
\begin{equation}
	g_T \myp[\big]{\mu+C_T(1+\xi^2_{T,\mu})\eps_\ell }
	\leq \frac{G_T \myp{\mu, C_{\ell}}}{\ell^d}
	\leq g_T \myp[\big]{\mu-C_T(1+\xi^2_{T,\mu})\eps_\ell }
\label{eq:shiftmu_T}
\end{equation}
for $T>0$, and
\begin{equation}
g_0\left(\mu+\frac{C\xi_{0,\mu}}{\ell^{\frac{s-d}{s-d+1}}}\right)\leq  \frac{G_0 \myp{\mu, C_{\ell}}}{\ell^d}\leq g_0\big(\mu-C\xi_{0,\mu}\eps_\ell\big).
\label{eq:shiftmu_T0}
\end{equation}
\end{corollary}

\begin{proof}[Proof of \cref{cor:gcconvrate_mu}]
Let $\rho:=-\frac{\partial}{\partial\mu_+} g_T(\mu)$ be the (right) density corresponding to the chemical potential $\mu$. If $\rho\geq1$ we simply estimate $\xi_{T,\mu}^2\leq \xi_{T,\mu}^2\rho$. For $\rho\leq1$ we use Corollary~\ref{thm:mubounds} (in the limit $L\to\ii$), which provides
$$\xi_{T,\mu}^2\leq C_T\rho,\qquad\text{for $\rho\leq1$ and $T>0$}$$
where $C_T$ only depends on $T$. All in all, we can thus bound when $T>0$
$$g_T(\mu)-C\xi_{T,\mu}^2\eps_\ell\geq g_T(\mu)+C(1+\xi_{T,\mu}^2)\eps_\ell\frac{\partial}{\partial\mu_+}g_T(\mu)\geq g_T\left(\mu+C_T(1+\xi_{T,\mu}^2)\eps_\ell\right)$$
and
$$g_T(\mu)+C\xi_{T,\mu}^2\eps_\ell\leq g_T(\mu)-C(1+\xi_{T,\mu}^2)\eps_\ell\frac{\partial}{\partial\mu_+}g_T(\mu)\leq g_T\left(\mu-C_T(1+\xi_{T,\mu}^2)\eps_\ell\right),$$
by concavity of $\mu\mapsto g_T(\mu)$. This is the claimed estimate~\eqref{eq:shiftmu_T}.

At $T=0$ the situation is more complicated, since we have no good control on $\xi_{0,\mu}$ at low density. Recall that $\mu\mapsto g_0(\mu)$ vanishes over some interval $(-\ii,\mu_c)$ for some unknown $\mu_c=f_0'(0)$. We need to revisit the proof of Proposition~\ref{prop:gcconvrate} and we will get a slightly worse lower bound.

\medskip

\noindent\textbf{Lower bound on $ G_0$.} We argue as before, except that we use $\bP_L$ as a trial state for the grand-canonical problem in $C_L$ at a modified chemical potential $\mu + \nu_\ell$, with $ \nu_\ell $ to be chosen later. We will also have to take the length of the corridors $\delta$ large. 	The free energy then satisfies
$$\mathcal{G}_0^{\mu+ \nu_\ell} \myp{\bP_L}\leq{} \frac{L^d}{\myp{\ell+\delta}^d} G_0 \myp{\mu, C_{\ell}} + \frac{L^d}{\myp{\ell+\delta}^d} \sum\limits_{\substack{k \in \Z^d \\ k \neq 0}} 2 \abs{ D_{w} \myp{\rho_{\ell,0}, \rho_{\ell,k}} } - \nu_\ell \cN \myp{\bP_L}.$$
For simplicity, in the interaction of one cube to the rest of the system we have added all the cubes in $\R^d$. The main idea is to use the last term to control the interaction, but this requires an estimate in terms of $\cN \myp{\bP_L}$. We thus use the Ruelle bound for the cubes outside of $C_\ell^0$ but not for $C^0_\ell$ itself, and obtain by~\eqref{eq:integralbound2} and Lemma~\ref{lem:integralbound}
\begin{align*}
\abs{ D_{w} \myp{\rho_{\ell,0}, \rho_{\ell,k}} }
&\leq C r^d \xi_{0,\mu} \sum\limits_{j \in \widetilde{\mathcal{A}}} \expec[\big]{n_{C_r^{0,j}}}_{0,\mu,C_{\ell}^0} \sum\limits_{m \in \widetilde{\mathcal{A}}} \frac{1}{1+ \rd \myp[\big]{C_r^{0,j}, C_r^{k,m}}^s }\\
&\leq C\xi_{0,\mu} \sum\limits_{j \in \widetilde{\mathcal{A}}} \expec[\big]{n_{C_r^{0,j}}}_{0,\mu,C_{\ell}^0} \int_{(C_{\ell+\delta})^c}\frac{\rd y}{1+ |X_j-y|^s }\\
&\leq C\xi_{0,\mu} \sum\limits_{j \in \widetilde{\mathcal{A}}} \frac{\expec[\big]{n_{C_r^{0,j}}}_{0,\mu,C_{\ell}^0}}{\rd(X_j,\partial C_{\ell+\delta})^{s-d}}
\end{align*}
where $X_j$ denotes the center of the cube $C_r^{0,j}$. As a sub-optimal but simple bound we can use that $\rd(X_j,\partial C_{\ell+\delta})\geq \delta$ and thus get after summing
$$\mathcal{G}_0^{\mu+ \nu_\ell} \myp{\bP_L}\leq{} \frac{L^d}{\myp{\ell+\delta}^d} G_0 \myp{\mu, C_{\ell}} + \myp[\Big]{ C\frac{\xi_{0,\mu}}{\delta^{s-d}}- \nu_\ell } \cN \myp{\bP_L}.$$
Choosing $\nu_\ell=C\xi_{0,\mu}/{\delta^{s-d}}$ and passing to the limit $L\to\ii$, we arrive at
$$(1+\delta/\ell)^d
g_0 \myp[\Big]{ \mu+\frac{C\xi_{0,\mu}}{\delta^{s-d}} } \leq{\ell^{-d}} G_0 \myp{\mu, C_{\ell}}.$$
On the other hand, using the stability of $w$, we have
\begin{equation}
(1+t)g_0(\mu)\geq g_0(\mu+t\mu+t\kappa),\qquad\forall \mu\in\R,\ t>0,
\label{eq:simple_lower_bd_g_0}
\end{equation}
and we thus obtain
$${\ell^{-d}} G_0 \myp{\mu, C_{\ell}}\geq g_0\left(\mu+\frac{C\xi_{0,\mu}}{\delta^{s-d}} +\Big(\mu+\frac{C\xi_{0,\mu}}{\delta^{s-d}}+\kappa\Big)\frac{\delta}{\ell}\right).$$
The optimal choice here is $ \delta = \ell^{\frac{1}{s-d+1}} $, which leads to the lower bound in~\eqref{eq:shiftmu_T0}, after changing the constant $C_0$ in the definition \eqref{eq:def_xi} of $\xi_{0,\mu}$.

\medskip

\noindent\textbf{Upper bound on $ G_0$.} It turns out that we can get the optimal rate. We argue exactly as in the proof of Proposition~\ref{prop:gcconvrate}, taking this time $\bP$ a minimizer for the chemical potential $\mu-\nu_\ell$. We also split $\R^d=\cup C_\ell^k$ using cubes \emph{a priori} unrelated to the large cube $C_L$ (the latter is not necessarily an exact union of the smaller cubes). In the estimate~\eqref{eq:interactbound} we use the important fact that
$$\expec[\big]{ n_{C_r^{k,\gamma}}n_{C_r^{m,\gamma'}}}_{0,\mu,C_L}=\expec[\big]{ n_{C_r^{k,\gamma}}}_{0,\mu,C_L}\expec[\big]{n_{C_r^{m,\gamma'}}}_{0,\mu,C_L}$$
since $\bP$ is here a delta measure as a minimizer, hence everything is deterministic. In the following we simply suppress the expectations. The same arguments as for the lower bound then provide
\begin{align*}
	\MoveEqLeft[2] G_0 \myp{\mu-\nu_\ell, C_L}\\
	\geq{}& L^d \frac{G_0 \myp{\mu, C_{\ell}}}{\ell^d} - C\xi_{0,\mu}\sum_{k} \sum\limits_{j \in \widetilde{\mathcal{A}}} \frac{n_{C_r^{k,j}}}{1+\rd(X_j,\partial C_{\ell}^k)^{s-d}}+\nu_\ell\cN(\bP)\\
	={}& L^d \frac{G_0 \myp{\mu, C_{\ell}}}{\ell^d} +\int_{C_L} \myp[\bigg]{ \nu_\ell- C\xi_{0,\mu}\sum_{k} \sum\limits_{j \in \widetilde{\mathcal{A}}} \frac{\1_{C_r^{k,j}} \myp{x}}{1+\rd(X_j,\partial C_{\ell}^k)^{s-d}} } \rho_\bP \myp{x} \id x\\
	\geq{}& L^d \frac{G_0 \myp{\mu, C_{\ell}}}{\ell^d} +\int_{C_L}\left(\nu_\ell- \frac{C'\xi_{0,\mu}}{1+\rd(x, \cup_k \partial C_{\ell}^k)^{s-d}}\right)\rho_\bP(x) \id x.
\end{align*}
Now we use that the right side simplifies if we average over the position of the tiling of size $\ell$, at fixed $\bP$. Under this procedure, the periodic function
$$x\mapsto \frac{1}{1+\rd(x, \cup_k \partial C_{\ell}^k)^{s-d}}$$
is replaced by a constant over $\R^d$, which equals its average over one cube. This can be bounded by
$$\frac1{\ell^d}\int_{C^0_\ell}\frac{\dx }{1+\rd(x, \partial C_{\ell}^0)^{s-d}}\leq C\eps_\ell,$$
due to the proof of Lemma~\ref{lem:integralbound}. Thus we have proved that
$$G_0 \myp{\mu-\nu_\ell, C_L}\geq L^d \frac{G_0 \myp{\mu, C_{\ell}}}{\ell^d} +\left(\nu_\ell- C'\xi_{0,\mu}\eps_\ell\right)\cN(\bP)$$
and we can conclude using $\nu_\ell=C'\xi_{0,\mu}\eps_\ell$ and taking $L\to\ii$.
\end{proof}

We will need the following tool.

\begin{lemma}[Cutting and translating states]
\label{lem:measurecut}
	Let $ \Omega = \Omega_1 \cup \Omega_2 \subseteq \R^d $ be any disjoint union of Borel subsets, and consider a tiling $\R^d=\bigcup_{k\in\Z^d} L_0(k+Q)$ of cubes of side length $L_0$, with $Q=(0,1]^d$. Let $ \bP $ be a grand-canonical state on $ \Omega $.
For any $ v \in \R^d $ satisfying $ \Omega_1 \cap \myp{\Omega_2 + v}= \emptyset $, we define a map $ T: \Omega \rightarrow \Omega_v := \Omega_1 \cup \myp{\Omega_2 + v} $ by $ T\myp{x} := x \1_{\Omega_1} \myp{x} + \myp{x+v} \1_{\Omega_2} \myp{x} $, and a state $ \bP_v $ supported on $ \Omega_v $ by $ \myp{\bP_v}_n \myp{x} := \bP_n \myp{\myp{T^{-1}}^{\otimes n} \myp{x} } $.
	Then the $ k $-particle densities of $ \bP_v $ are given by
	\begin{equation}
	\label{eq:measurecutdensity}
		\rho_{\bP_v}^{\myp{k}} \myp{x}
		:= \sum\limits_{n \geq k} \frac{1}{\myp{n-k}!} \int_{\Omega_v^{n-k}} \bP_v \myp{x,y} \id y
		= \rho_{\bP} \myp{\myp{T^{-1}}^{\otimes k} \myp{x}}.
	\end{equation}
	Furthermore, if $ \rd \myp{\Omega_1, \Omega_2 + v} \geq 2r_0+2\sqrt{d}L_0 $, then the free energy of the state $ \bP_v $ satisfies
	\begin{equation}
	\label{eq:measurecutenergy_gc}
		\mathcal{G}_T \myp{\bP_v}
		\leq \mathcal{G}_T \myp{\bP} + 2\kappa \mathcal{N} \myp{\bP} + \frac{C\sum_{k}\big\langle n_{k}^2\big\rangle_\bP}{\rd \myp{\Omega_1, \Omega_2+v}^{s-d}},
	\end{equation}
	where $n_{k}$ denotes the number of particles in the cube $L_0(k+Q)$.
\end{lemma}

Recall that $r_0$ is the range of the core of the interaction $w$, as defined in Assumption~\ref{de:shortrangenew}.

\begin{proof}
	The fact that the $ k $-particle densities satisfy \eqref{eq:measurecutdensity} follows immediately from $ T $ being measure preserving, which also implies
	\begin{equation*}
		\mathcal{S} \myp{\bP_v} = \mathcal{S} \myp{\bP}.
	\end{equation*}
	Hence, to finish the proof, we only need to provide a bound on the interaction energy $ \mathcal{U} \myp{\bP_v} $.
	Applying \eqref{eq:measurecutdensity}, we have
	\begin{align*}
		\mathcal{U} \myp{\bP_v}
		={}& \frac{1}{2} \iint_{\Omega_v^2} w \myp{x-y} \rho_{\bP_v}^{\myp{2}} \myp{x,y} \id x \id y \\
		={}& \mathcal{U} \myp{\bP} - \iint_{\Omega_1 \times \Omega_2 } w \myp{x - y} \rho_{\bP}^{\myp{2}} \myp{x, y} \id x \id y\\
 		&+ \iint_{\Omega_1 \times \myp{\Omega_2 + v}} w \myp{x - y} \rho_{\bP}^{\myp{2}} \myp{x, y-v} \id x \id y.
	\end{align*}
To estimate the second term, we use the stability of $ w $ in the form
	\begin{align*}
		\MoveEqLeft[4] \iint_{\Omega_1 \times \Omega_2 } w \myp{x - y} \rho_{\bP}^{\myp{2}} \myp{x, y} \id x \id y \\
		={}& 2 \iint_{\Omega_1\times\Omega_2} \sum\limits_{n \geq 2} \frac{1}{n!} \int_{\Omega^{n-2}} \sum\limits_{j < k} w \myp{x_j - x_k} \bP_n \myp{x} \id x \\
		\geq{}& 2 \sum\limits_{n \geq 2} \frac{1}{n!} \int_{\Omega^n} \sum\limits_{j < k} w_2 \myp{x_j - x_k} \bP_n \myp{x} \id x
		\geq{} - 2 \kappa \mathcal{N} \myp{\bP}.
	\end{align*}
To estimate the third term, we argue similarly as in~\eqref{eq:interactbound}, leading to
	\begin{align*}
		\MoveEqLeft[6] \iint_{\Omega_1 \times \myp{\Omega_2 + v}} w \myp{x - y} \rho_{\bP}^{\myp{2}} \myp{x, y-v} \id x \id y\\
		\leq{}& C\sum_{k,m} \int_{Q_k\cap \Omega_1}\int_{Q_m\cap \Omega_2}\frac{\dx\,\dy}{1+|x-y-v|^s}\pscal{ n_k n_m}_{\bP}\nn \\
		\leq{}& C\sum_{k,m} \int_{Q_k\cap \Omega_1}\int_{Q_m\cap \Omega_2}\frac{\dx\,\dy}{1+|x-y-v|^s}\left(\pscal{ n_k^2}_\bP+\pscal{n_m^2}_{\bP}\right)\nn \\
		\leq{}& \frac{C}{\rd(\Omega_1,\Omega_2+v)^{s-d}}\sum_k \pscal{ n_k^2}_\bP,
	\end{align*}
by Lemma~\ref{lem:integralbound}. This concludes the proof of \eqref{eq:measurecutenergy_gc}.
\end{proof}

We now provide a bound on the convergence rate at constant density.

\begin{proposition}[Convergence rate at fixed density]
\label{prop:gcconvrate_fixed}
	Let $ T \geq 0 $ and $ \rho > 0 $. If $T=0$ we also assume that $s>d+1$. There is a constant $ C > 1 $ depending only on $ L_0 $, $ T $, and $ w $, such that for $ \ell $ sufficiently large,
	\begin{equation}
 \left|\frac{G_T \myb{\rho \1_{C_{\ell}}}}{\ell^d}-f_T \myp{\rho}\right|\leq \xi(\rho)\eta_\ell,
		\label{eq:upper_CV_rate_rho}
	\end{equation}
	with
	$$\xi(\rho):=C\begin{cases}
		\rho e^{C\rho^{\gamma-1}} &\text{if $T>0$,}\\
 		\sqrt{\rho}\left(1+\rho^{\gamma-1}\right)^{2+\frac{2d}{\eps}} &\text{if $T=0$,}
		\end{cases}$$
		and
$$\eta_\ell:=\begin{cases}
		\ell^{-1/2}&\text{if $s>d+1$,}\\
		\ell^{-1/2}\sqrt{\log\ell} &\text{if $s=d+1$,}\\
		\ell^{-\frac{s-d}{1+s-d}}&\text{if $d<s<d+1$\,.}
		\end{cases}$$
		When $\gamma=2$, $\rho^{\gamma-1}$ has to be replaced by $\rho\log(2+\rho)$ in the definition of $\xi(\rho)$.
		Here, $\eps=\min(1,s-d)/2$ is the same as in \eqref{eq:def_xi}.
	\end{proposition}

\begin{proof}
First we quickly discuss the lower bound. We introduce $\mu=\mu(\rho)$ (at $T=0$ we take the largest admissible $\mu$) and write
\begin{align*}
\ell^{-d}G_T[\rho\1_{C_\ell}]&\geq \ell^{-d}G_T(\mu,C_\ell)+\mu\rho\\
&\geq g_T(\mu)+\mu\rho -C\xi_{T,\mu}^2\eps_\ell=f_T(\rho)-C\xi_{T,\mu}^2\eps_\ell
\end{align*}
by Proposition~\ref{prop:gcconvrate}. Recall that $\xi_{T,\mu}$ is defined in~\eqref{eq:def_xi}. When $T>0$, we have $\xi_{T,\mu}^2\leq\xi(\rho)$ for a large enough constant $C$, due to the bounds~\eqref{eq:muupbounds} on the chemical potential. This provides an estimate better than the one stated, with $\eps_\ell$ in place of $\eta_\ell$. When $T=0$ we only get
$$\ell^{-d}G_0[\rho\1_{C_\ell}]\geq f_0(\rho)-C\left(1+\rho^{\gamma-1}\right)^{2+\frac{2d}{\eps}}\eps_\ell.$$
When $\gamma=2$, $\rho^{\gamma-1}$ is replaced by $\rho\log(2+\rho)$. This does not have the claimed behavior $\sqrt\rho$ at low density. However, from the universal bounds in~\cite[Thm.~11]{JexLewMad-23} we see that
$$\abs[\big]{ \ell^{-d} G_0 \myb{\rho \1_{C_{\ell}}} - f_0 \myp{\rho} } \leq C\rho(1+\rho^{\gamma-1})$$
(with an additional logarithm for $\gamma=2$).
We can therefore introduce a power of $\rho$ at the expense of decreasing the power of $\ell$:
\begin{equation}
\ell^{-d} G_0 \myb{\rho \1_{C_{\ell}}} - f_0 \myp{\rho}\geq -C\sqrt{\eps_\ell}\sqrt\rho\left(1+\rho^{\gamma-1}\right)^{\frac32+\frac{d}{\eps}},
 \label{eq:CV_rate_T0}
\end{equation}
which can be bounded by $\xi(\rho)$. Note that $\sqrt{\eps_\ell}=\eta_\ell$ since $s>d+1$ by assumption.

Next we turn to the upper bound. This time, we split the cube of interest into smaller cubes. For consistency with the previous proofs we thus call the length of interest $L$ and the smaller length $\ell$. Let $ \delta > 0$ and the length $\ell + \delta =: \widetilde{\ell} < L $. We consider the tiling of space
$$\R^d=\bigcup_{k\in\Z^d}C_{\widetilde\ell}^k,\qquad C^k_{\widetilde\ell}:=C_{\widetilde{\ell}} + k\widetilde{\ell}.$$
Consider a function $\rho_\ell$ in $C_{\tilde\ell}$ supported in the slightly smaller cube $C_\ell$, and such that $ \int_{C_\ell} \rho_\ell = \rho \widetilde{\ell}^d $. Repeat this function in each cube $C_{\widetilde\ell}^k$ to obtain an $\tilde\ell$--periodic function over $\R^d$. When we average this function over translations of the tiling, we obtain a constant $\rho$ over the whole space:
	\begin{equation*}
		\rho
		= \frac{1}{\widetilde{\ell}^d} \int_{C_{\widetilde{\ell}}} \sum\limits_{k \in \Z^d} \rho_{\ell} \myp{x-k \widetilde{\ell} - \tau} \id \tau,\qquad\forall x\in\R^d.
	\end{equation*}
The idea of the proof is to interpret this as a partition of unity and write
\begin{equation}
\label{eq:densitydecomp}
\rho \1_{C_{L}} \myp{x}
= \frac{1}{\widetilde{\ell}^d} \int_{C_{\widetilde{\ell}}} \sum\limits_{k \in \Z^d} \rho_\ell \myp{x-k \widetilde{\ell} - \tau} \1_{C_{L}} \myp{x} \id \tau.
\end{equation}
We choose the size of the corridors, $ \delta $, to be large enough but fixed. Let then $\bP_\ell$ be the Gibbs state (or a minimizer at $T=0$) for $G_T(\mu_\ell,C_\ell)$, where $\mu_\ell$ is chosen so that
\begin{equation*}
\cN(\bP_\ell)=\rho \widetilde{\ell}^d.
\end{equation*}
In other words, $\bP_\ell$ has the exact average density $\rho$ in the bigger cube $C_{\tilde\ell}$. We take $\rho_\ell:=\rho_{\bP_\ell}$. For every $\tau$, we will use $\bP_\ell$ to construct a trial state $\bQ_\tau$ of density
$$\rho_{\bQ_\tau}(x):=\sum\limits_{k \in \Z^d} \rho_\ell \myp{x-k \widetilde{\ell} - \tau} \1_{C_{L}} \myp{x}=:\sum_{k\in\Z^d}\rho_\tau^k(x).$$
Then the state
$$\bQ:=\frac{1}{\widetilde{\ell}^d} \int_{C_{\widetilde{\ell}}}\bQ_\tau\,\rd \tau$$
has the desired density $\rho\1_{C_L}$. We now explain how to construct $\bQ_\tau$. For every fixed $\tau$, we call $\cI_\tau\subset\Z^d$ the indices of the cubes which are completely inside $C_L$:
$$\cI_\tau:=\{k\in\Z^d\ :\ C_{\widetilde\ell}^k+\tau\subset C_L\}.$$
We call $\cB_\tau$ the set of indices so that $C_{\widetilde\ell}^k+\tau$ intersects the boundary of $C_L$. We use the fact that the union of these boundary cubes
$$\bigcup_{k\in\cB_\tau}C_L\cap(C^k_{\tilde\ell}+\tau)$$
is made of a certain number of \emph{complete} cubes which have been split into finitely many parts as displayed in Figure~\ref{fig:cubes}. In other words, any piece of cube at the boundary can be merged with some other pieces to build an entire cube. We can thus write
$$\bigcup_{k\in\cB_\tau}C_L\cap(C^k_{\tilde\ell}+\tau)=\bigcup_{k\in\cB'_\tau}T_k(C^k_{\tilde\ell}+\tau)$$
where $T_k$ is a translation map as in Lemma~\ref{lem:measurecut} (or, rather, the composition of finitely many such maps depending on the number of pieces), and $\cB'_\tau$ is a subset of $\cB_\tau$.

  \begin{figure}
\begin{tikzpicture}
    \node[anchor=south west,inner sep=0] (image) at (0,0) {\includegraphics[width=.9\textwidth]{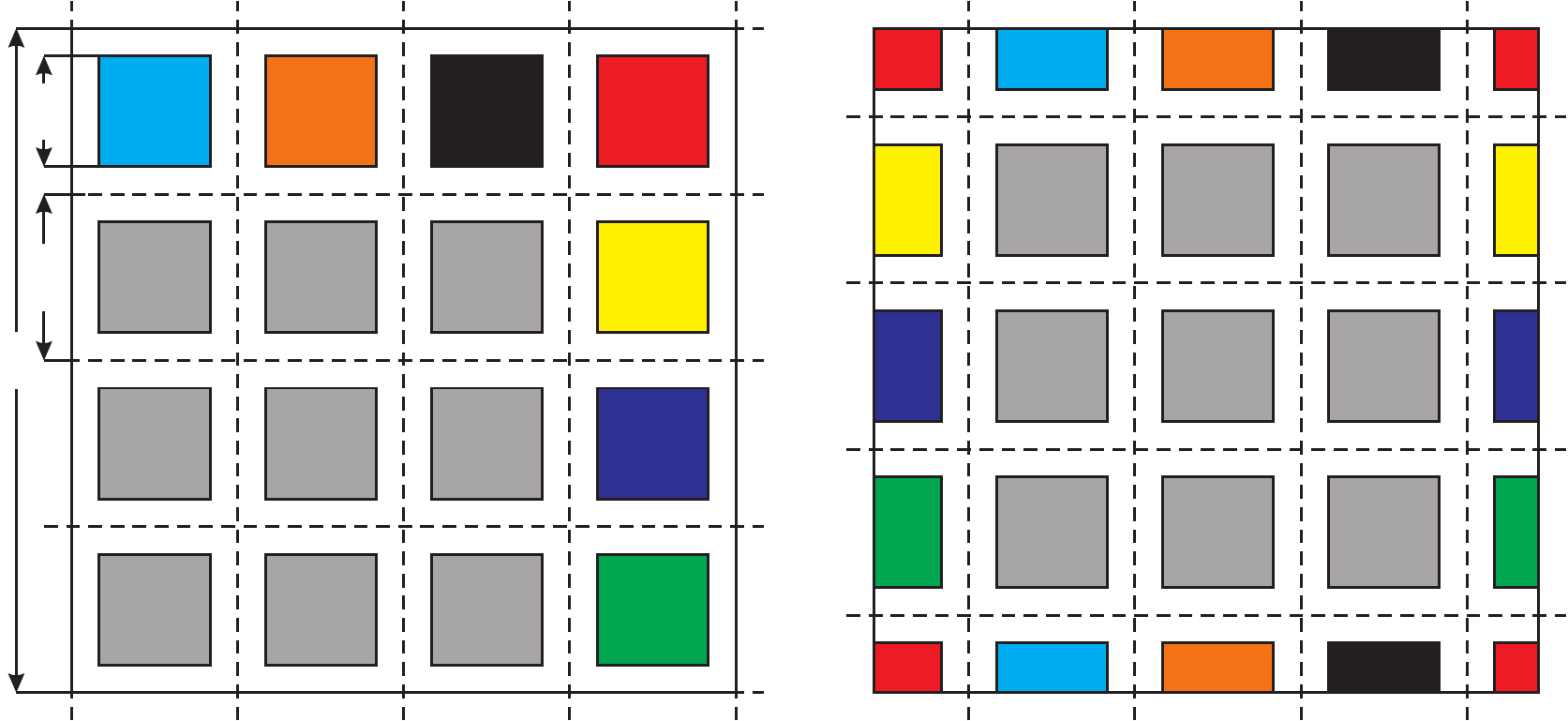}};
    \begin{scope}[x={(image.south east)},y={(image.north west)}]
        \draw (0.0285,0.895) node[below] {$\ell$};
       \draw (0.0285,0.675) node[below] {$\tilde \ell$};
       \draw (0.012,0.55) node[below] {$L$};
    \end{scope}
\end{tikzpicture}
  \caption{Splitting of the large cube $C_L$ into smaller cubes $C_{\tilde\ell}$ with corridors, for $\tau=0$ (left figure) and a shift $\tau\neq0$ (right figure). The grey color corresponds to the cubes which stay inside when translated by $\tau$. The other colors display the other pieces at the boundary. The sets of a given color can be merged to form an entire cube. For those sets, the trial state is constructed using \cref{lem:measurecut}.
  \label{fig:cubes}}
  \end{figure}
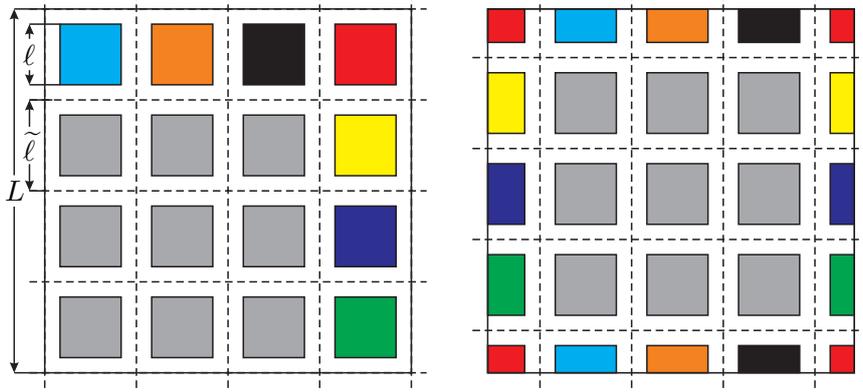

Next we use as trial state for a given shift $\tau$
$$\bQ_\tau=\bigotimes_{k\in\cI_\tau}\bP_\ell^k\otimes \bigotimes_{k\in\cB'_\tau}\tilde\bP_\ell^k$$
with $\bP_\ell^k:= \bP_\ell+k\tilde\ell+\tau$ and $\tilde\bP_\ell^k$ constructed from $\bP_\ell$ using Lemma~\ref{lem:measurecut}. This state has the desired density $\rho_\tau^k$ by construction. From the subadditivity of the entropy, we have
$$G_T[\rho\1_{C_L}]\leq \frac1{\tilde\ell^d}\int_{C_{\tilde\ell}}\cG_T(\bQ_\tau)\,\rd \tau$$
and
$$\cG_T(\bQ_\tau)=\sum_{k\in\cI_\tau}\cG_T(\bP_\ell)+\sum_{k\in\cB'_\tau}\cG_T(\tilde \bP^k_\ell)+2\sum_{k\neq k'\in\cI_\tau\cup\cB'_\tau}D(\rho_\tau^k,\rho_\tau^{k'}).$$
Using Lemma~\ref{lem:measurecut} and the fact that the pieces of a cube at the boundary are at a distance of order $L$ from each other, we have for $k\in\cB'_\tau$
$$\cG_T(\tilde\bP^k_\ell)\leq\cG_T(\bP_\ell)+ \tilde\ell^d2\kappa \rho + \frac{C}{L^{s-d}}\sum_{j\in \Z^d}\pscal{n_j^2}_{\bP_\ell}$$
where $n_j$ denotes the number of particles in the cube $L_0(j+Q)$, with $L_0$ the side length for which the Ruelle bound~\eqref{eq:ruelle} holds. When $T>0$, the Ruelle bound provides
$$\sum_{j}\pscal{n_j^2}_{\bP_\ell}\leq C\ell^d\xi_{T,\mu_\ell}^2\leq C\rho e^{C\rho^{\gamma-1}}\tilde \ell^d,$$
using the definition of $\xi_{T,\mu}$ and the bounds~\eqref{eq:muupbounds} on the chemical potential. There is an additional logarithm when $\gamma=2$. When $T=0$ the state $\bP_\ell$ can be taken to be a delta measure, hence
$$\pscal{n_j^2}_{\bP_\ell}=\pscal{n_j}^2_{\bP_\ell}\leq  \xi_{0,\mu_\ell}\pscal{n_j}_{\bP_\ell}$$
so that
$$\sum_{j}\pscal{n_j^2}_{\bP_\ell}\leq C\rho \left(1+\rho^{\gamma-1}\right)^{1+\frac{d}{\eps}}\tilde \ell^d.$$
On the other hand, the interaction between any cube and the rest can be estimated exactly like in~\eqref{eq:estim_int_rho} in the proof of Proposition~\ref{prop:gcconvrate} by
$$\sum_{k'\in\cI_\tau\cup\cB_\tau'\setminus\{k\}}D(\rho_\tau^k,\rho_\tau^{k'})\leq C\xi_{T,\mu_\ell}^2\ell^d\eps_\ell.$$
Note that this uses the presence of the corridors to estimate $w\leq \kappa(1+|x|^s)^{-1}$. Thus we have proved that
$$
\frac{\cG_T(\bQ_\tau)}{L^d}\leq \frac{\cG_T(\bP_\ell)}{\tilde\ell^d}+\zeta(\rho)\begin{cases}
\rho\left(\frac{\ell}{L}+\eps_\ell\right)&\text{if $T>0$,}\\
\rho\frac{\ell}{L}+\eps_\ell &\text{if $T=0$.}
\end{cases}
 $$
where
$$\zeta(\rho)=\begin{cases}
Ce^{C\rho^{\gamma-1}}&\text{if $T>0$,}\\
C\left(1+\rho^{\gamma-1}\right)^{2+\frac{2d}{\eps}}&\text{if $T=0$,}
    \end{cases}$$
with an additional logarithm for $\gamma=2$. We have used here that the volume of the union of the cubes in $\cB'_\tau$ can be controlled by $L^{d-1}\ell$, since $\ell$ and $\tilde\ell$ are comparable.
Next we study the free energy of $\bP_\ell$ and recall that
$$		\tilde{\ell}^{-d} \mathcal{G}_T \myp{\bP_\ell}
		={} (1+\delta/\ell)^{-d} f_T \big(\rho(1+\delta/\ell)^d,\ell\big).$$
We apply \cref{cor:gcconvrate_mu} and obtain
\begin{align*}
f_T (\rho(1+\delta/\ell)^d,\ell)
={}&g_T(\mu_\ell,C_\ell)+\mu_\ell\rho(1+\delta/\ell)^d\nn\\
\leq{}& g_T(\mu_\ell-\zeta(\rho)\eps_\ell)+\mu_\ell\rho(1+\delta/\ell)^d.
\end{align*}
so that
\begin{equation}
\tilde{\ell}^{-d} \mathcal{G}_T \myp{\bP_\ell}
\leq{} (1+\delta/\ell)^{-d}g_T(\mu_\ell-\zeta(\rho)\eps_\ell)+\mu_\ell\rho.
\label{eq:estim_rho_f_T}
\end{equation}
If $T>0$ we use that $g_T(\mu)\geq -Te^{\frac{1}{T}(\mu+\kappa)}$ due to the stability of $w$ and the free energy of the non-interacting gas. This gives
\begin{align*}
(1+\delta/\ell)^{-d}g_T(\mu_\ell-\zeta(\rho)\eps_\ell)&\leq (1-C\delta/\ell)g_T(\mu_\ell-\zeta(\rho)\eps_\ell)\\
&\leq g_T(\mu_\ell-\zeta(\rho)\eps_\ell)+C\frac{\delta}\ell Te^{\mu_\ell/T}\\
&\leq g_T(\mu_\ell-\zeta(\rho)\eps_\ell)+\eps_\ell\rho\zeta(\rho),
\end{align*}
from the bounds on $\mu_\ell$ in Proposition~\ref{thm:mubounds} and the fact that $1/\ell\leq \eps_\ell$, after increasing the constant in $\zeta(\rho)$. From the duality formula $f_T(\rho)=\sup_\nu \{g_T(\nu)+\nu\rho\}$, we then obtain from~\eqref{eq:estim_rho_f_T}
\begin{equation}
\tilde{\ell}^{-d} \mathcal{G}_T \myp{\bP_\ell}
\leq f_T(\rho)+2\zeta(\rho)\eps_\ell\rho.
\label{eq:bound_f_T_positive}
\end{equation}
At $T=0$ we argue slightly differently and immediately use the duality formula in~\eqref{eq:estim_rho_f_T} to get
\begin{align*}
\tilde{\ell}^{-d} \mathcal{G}_0 \myp{\bP_\ell}
\leq{}& (1+\delta/\ell)^{-d}g_0(\mu_\ell-\zeta(\rho)\eps_\ell)+\mu_\ell\rho\\
\leq{}& (1+\delta/\ell)^{-d}f_0(\rho)+\frac{\zeta(\rho)\rho\eps_\ell}{(1+\delta/\ell)^d}+\mu_\ell\rho \myp[\Big]{ 1-\frac1{(1+\delta/\ell)^d} }\\
\leq{}& f_0(\rho)+C\frac\delta\ell f_0(\rho)_-+\zeta(\rho)\rho\eps_\ell+C\frac\delta\ell (\mu_\ell)_+\rho.
\end{align*}
Using that $f_0(\rho)\geq -\kappa\rho$ by stability and the upper bounds on $\mu_\ell$ in Proposition~\ref{thm:mubounds}, we get the same as~\eqref{eq:bound_f_T_positive} at $T=0$.

As a conclusion, after increasing the constant in $\zeta(\rho)$, we have proved in all cases
\begin{equation}
\frac{\cG_T(\bQ_\tau)}{L^d}\leq f_T(\rho)+C\zeta(\rho)\begin{cases}
\rho\left(\frac{\ell}{L}+\eps_\ell\right)&\text{if $T>0$,}\\
\rho\frac{\ell}{L}+\eps_\ell &\text{if $T=0$.}\end{cases}
 \label{eq:bound_sliding_cutting2}
\end{equation}
After optimizing we are led to choosing
$$\ell=\begin{cases}
		\sqrt{L}&\text{for $s>d+1$,}\\
		\sqrt{L\log L}&\text{for $s=d+1$,}\\
		L^{\frac1{1+s-d}}&\text{for $d<s<d+1$}
		\end{cases}$$
for $T>0$ or $T=0$ with $\rho\geq1$, and
$$\ell=\begin{cases}
		\sqrt{L/\rho}&\text{for $s>d+1$,}\\
		\sqrt{L/\rho\log (L/\rho)}&\text{for $s=d+1$,}\\
		(L/\rho)^{\frac1{1+s-d}}&\text{for $d<s<d+1$}
		\end{cases}$$
	for $T=0$ and $\rho\leq1$. This is how we get $\eta_L$ for the error term. When $T=0$ and $s=d+1$ we have an additional $\log\rho$ which is why we have assumed for simplicity that $s>d+1$ at zero temperature, in the statement. This concludes the proof of Proposition~\ref{prop:gcconvrate_fixed}.
\end{proof}

\section{Proof of Theorem~\ref{thm:lda} on the local density approximation}
\label{sec:lda}
\subsection{Lower bound}
In this subsection, we prove the lower bound on $ G_T \myb{\rho} $ in \cref{thm:lda}.

\begin{proposition}[LDA from below]
\label{thm:ldalow}
	Let $ M > 0 $, and $ p \geq 1 $, and let $ w $ be a short-range interaction (satisfying the conditions of \cref{de:shortrangenew}).
	Furthermore, take any $ b \geq 0 $ satisfying
	\begin{equation}
	\label{eq:b_cond}
		b >1+\min(1,s-d)\left( 1-\frac{1}{2p} \right).
	\end{equation}
	For any $T \geq 0$, there exists a constant $ C > 0 $ depending on $ M $, $ p $, $ b $, $ T $, $d$, and $ w $, such that
	\begin{equation}
		G_T \myb{\rho} - \int_{\R^d} f_T \myp{\rho \myp{x}} \id x
		\geq{} - C \sqrt{\eps_\ell} \myp[\bigg]{ \int_{\R^d} \sqrt{\rho} + \ell^{bp} \int_{\R^d}\delta\rho_\ell(z)^p\,\rd z},
	\label{eq:ldalow}
	\end{equation}
	for any $ \ell > 0 $, and any density $ 0 \leq \rho \in L^1 \myp{\R^d} $ satisfying $ \normt{\rho}{\infty} \leq M $ and $ \sqrt{\rho} \in L^1 \myp{\R^d} $.
	Here $\eps_{\ell}$ is defined by \eqref{eq:eps_ell}.
\end{proposition}

The upper bound stated in Proposition~\ref{thm:ldaup} below will require a stronger condition on $b$. In Theorem~\ref{thm:lda} we have just taken the worse of the two.

The main difficulty in proving Proposition~\ref{thm:ldalow} is that we have little information on the optimal state~$\bQ$ minimizing $G_T[\rho]$. In particular we have no local Ruelle-type bounds on~$\bQ$ (except on the average local number which is given by $\rho$ by definition). To circumvent this problem, we will argue by Legendre-Fenchel duality and just remove the constraint on the density at the expense of adding an external potential. More precisely, we use that
$$G_T \myb{\rho}\geq \min_{\bP}\left\{\cF_T(\bP)+\int_{\R^d}V(x)\rho_\bP(x)\,\rd x\right\}-\int_{\R^d}V(x)\rho(x)\,\rd x$$
for any external potential $V$. There is equality if we maximize over $V$ but here we just pick a well-chosen $V$ to obtain the lower bound~\eqref{eq:ldalow}. To construct this external potential we tile $\R^d$ with cubes of side length $\ell$. In each cube we take $V$ to be an approximation of the expected opposite chemical potential $V\approx -\mu$ with $\mu=f'_T(\rho)$. More precisely, in the cubes $C_k$ where $\rho$ varies slowly enough (that is, $\delta\rho_\ell$ is small enough compared to the other terms), we just take the potential to be constant, for instance $-f_T'(\max_{C_k}\rho)$. In the cubes where $\rho$ varies too much, we take $V(x)=-f'_T(\rho(x))$. However, we will need to modify this choice of $ V $ a bit to also handle the zero temperature case. The proof requires to estimate the interactions between the cubes, for which we need the Ruelle bounds of Appendix~\ref{app:Ruelle}.

\begin{proof}
Due to the universal bounds on $G_T \myb{\rho}$ proved in~\cite{JexLewMad-23}, and the bounds on $f_T$ in \cref{thm:tdlimbounds}, the estimate~\eqref{eq:ldalow} holds when $\ell$ is finite, even without the term involving $\delta\rho_\ell$. We may thus assume in the whole proof that $\ell$ is large enough.

We write the whole proof for $T>0$ and explain at the end how to adapt it for $T=0$. As announced, we split $ \R^d $ into a union of cubes $ C_k = [ -\ell/2, \ell/2 )^d +  \ell k $, $ k \in \Z^d $ of side length $ \ell > 0 $.
	Denote by $ V $ the one-body potential
	\begin{equation}
	\label{eq:vdef}
		V = \sum\limits_{k \in \Z^d} v_k \1_{C_k},\qquad 	v_k = -\mu_k + K\eps_{\ell}
	\end{equation}
	where the functions $ \mu_k $ will be chosen below. The constant $K$ depends on $M$ and we explain later how to choose it so that the shift appearing in the lower bound on the convergence rate in Corollary~\ref{cor:gcconvrate_mu} is bounded above by $K\eps_\ell$. We let $ \bQ = \myp{\bQ_n} $ be a grand-canonical state minimizing the grand-canonical free energy $G_T[\rho]$ at fixed density $ \rho $.
Then we can rewrite and bound from below
	\begin{align}
		G_T \myb{\rho}
		={}& G_T \myb{\rho} + \int V \rho - \int V \rho \nn \\
		={}& \mathcal{U} \myp{\bQ} + \sum\limits_{n \geq 1} \int_{\R^{dn}} \sum\limits_{i=1}^n V \myp{x_i} \id \bQ_n \myp{x} - T \mathcal{S} \myp{\bQ} - \int V \rho \nn \\
		\geq{}& \inf_{\bP = \myp{\bP_n}} \Set[\big]{ \mathcal{U} \myp{\bP} + \cV \myp{\bP} - T \mathcal{S} \myp{\bP} } - \int V \rho,
	\label{eq:ldalow1}
	\end{align}
	where $ \cV \myp{\bP} $ denotes the grand-canonical energy of the potential $ V $ in the state $ \bP $.
	In the minimization problem above, we have removed the restriction that the one-body density of the state $ \bP $ must be equal to $ \rho $, meaning that the minimizer is just a Gibbs state (with external potential $V$).

\medskip
\noindent\textbf{Choosing the $ \mu_k $ and localizing.}
	Since the map $\rho\mapsto\mu(\rho)$ is well defined for $ T>0 $, but its inverse might not exist in case of phase transitions, it is easier to think in terms of densities. For any $k$, we will thus choose some function $\rho_k$ in $C_k$ and then take
	$$\boxed{\mu_k(x):=f'_T(\rho_k(x))=\mu(\rho_k(x)).}$$
	We will either pick $\rho_k = \max_{C_k} \rho $ constant (in which case $\mu_k$ is also constant), or just $\rho_k=\rho\1_{C_k}$. By definition, we have in all cases
	\begin{equation}
	\label{eq:mukdef}
		f_T \myp{\rho_k \myp{x} } = g_T \myp{\mu_k \myp{x} } + \mu_k \myp{x} \rho_k \myp{x},\qquad\forall x\in C_k.
	\end{equation}
Note that in either case $ \mu_k $ is universally bounded from above in terms of $ M \geq \normt{\rho}{\ii}$, by~\eqref{eq:muupbounds}. Hence by \eqref{eq:def_xi} we must have $\xi_{T,\mu_k}\leq C_M$ for all $k$. We choose the constant $K$ in~\eqref{eq:vdef} so that the shift appearing in Corollary~\ref{cor:gcconvrate_mu} is controlled by $K\eps_\ell$:
$$K\geq C_T(1+\xi^2_{T,\mu_k}),\qquad\forall k.$$

	Let us now denote by $ \bP $ the Gibbs state in the external potential $V$, that is, minimizing \eqref{eq:ldalow1}. By geometric localization and sub-additivity of the entropy, we obtain in the same way as \eqref{eq:geolowbound},
	\begin{align}
		G_T \myb{\rho}
		\geq{} &\mathcal{U} \myp{\bP} + \cV \myp{\bP} - T \mathcal{S} \myp{\bP} - \int_{\R^d} V \rho \nn \\
		\geq{}& \sum\limits_{k \in \Z^d} G_T^{v_k}(C_k) + \int_{C_k} \mu_k \rho + \sum\limits_{\substack{k,m \in \Z^d \\ k \neq m}} \expec{I_{k,m}}_{\bP} - K\eps_\ell \int_{\R^d} \rho,
	\label{eq:ldalow2}
	\end{align}
	where $ G_T^{v} \myp{C_k} $ denotes the minimal grand-canonical energy in $ C_k $ with external potential $ v $, and
	\begin{equation*}
		\expec{I_{k,m}}_{\bP}
		:={} \frac{1}{2} \iint_{C_k \times C_m} w \myp{x-y} \rho_{\bP}^{\myp{2}} \myp{x,y} \id x \id y
	\end{equation*}
	denotes the interaction between the cubes $ C_k $ and $ C_m $.
	We will first choose the $ \rho_k $ and bound $ G_T^{v_k}(C_k) $ from below, depending on the behaviour of $ \rho $ in each cube, and then treat the interaction terms at the end.

	For the remainder of the proof, we denote by
	\begin{equation*}
		r_k:=\min_{C_k}\rho,
		\qquad R_k:=\max_{C_k}\rho
	\end{equation*}
	the (essential) minimum and maximum values of $ \rho $ in the cube $ C_k $.

\medskip
\noindent\textbf{Free energy in the simple cubes.}
	Now, to decide whether to use $ \rho_k = \rho $ or $ \rho_k = R_k $, we consider all the $ k $'s for which
\begin{equation}
		\int_{C_k} \rho \myp{x} \id x
		\leq{} \sqrt{\eps_\ell} \left(\int_{C_k} \sqrt{\rho \myp{x}} \id x + \ell^{bp+d} \delta\rho_\ell \myp{\ell k}^p\right).
\label{eq:simple_cubes}
\end{equation}
	where $p\geq 1$ and $ b $ satisfying \eqref{eq:b_cond} are fixed as in the statement.
	In other words we look at the cubes where $\int_{C_k}\rho$ is already controlled by the error term we are aiming for. A similar argument was used in the quantum case in~\cite{LewLieSei-20}. Note that the set of such ``simple'' cubes where \eqref{eq:simple_cubes} holds contains all the cubes where the density is uniformly small.
	In fact, if $R_k\leq \eps_\ell$,
	then
	\begin{equation*}
		\int_{C_k} \rho \myp{x} \id x
		\leq \sqrt{\eps_\ell} \int_{C_k} \sqrt{\rho \myp{x}} \id x,
	\end{equation*}
	so~\eqref{eq:simple_cubes} holds, even with just the square root term.

For all those simple cubes, we choose $\rho_k=\rho\1_{C_k}$ and use the bound on $\int_{C_k}\rho$ to control the whole free energy. Using $ G_T^v \myp{C_k} \geq - C \int_{C_k} e^{-\frac{v}{T}} $ along with \cref{thm:mubounds}, we get for $ T> 0 $,
	\begin{align}
		G_T^{v_k}\myp{C_k} + \int_{C_k} \mu_k \rho - \int_{C_k} f_T \myp{\rho}
		\geq{}& - C \int_{C_k} e^{\frac{\mu_k-K\eps_{\ell}}{T}}+ \int_{C_k} \myp{T \log \rho_k - C} \rho  \nn \\
		&- \int_{C_k} C \rho + T \rho \log \rho \nn \\
		\geq{}& - C\int_{C_k} \rho.
	\label{eq:estim_simple}
	\end{align}
We have used here that $|\mu \myp{\rho}-T\log\rho|\leq C_{T,M}$ for $ \rho \leq M $ by~\eqref{eq:muupbounds} and~\eqref{eq:mulowbound}, as well as
$$f_T(\rho)\leq C_T(\rho^\gamma+\rho)+T\rho\log\rho\leq C_{T,M}\rho+T\rho\log\rho$$
by~\eqref{eq:fcanupbound}. It is thus important that $\rho_k(x)=\rho(x)$ in these cubes, so that the two logarithms cancel each other. Combining with \eqref{eq:simple_cubes}, we arrive at
	\begin{equation}
		G_T^{v_k}\myp{C_k} + \int_{C_k} \mu_k \rho - \int_{C_k} f_T \myp{\rho}
		\geq{} -C\sqrt{\eps_\ell} \left( \int_{C_k} \sqrt{\rho} + \ell^{bp+d} \delta \rho_{\ell} \myp{\ell k}^p\right)
	\label{eq:estim_simple_final}
	\end{equation}
	for the simple cubes. We obtain the desired error term for the local free energy, after summing over $k$.

\medskip
\noindent\textbf{Free energy in the main cubes.}
	For the remaining cubes we have
	\begin{equation}
		\int_{C_k} \rho \myp{x} \id x
		> \sqrt{\eps_\ell} \left( \int_{C_k} \sqrt{\rho \myp{x}} \id x + \ell^{bp+d} \delta\rho_\ell \myp{\ell k}^p\right).
	\label{eq:complicated_cubes}
	\end{equation}
	We call these the ``main cubes'' since, as we will see, the bound~\eqref{eq:complicated_cubes} implies that the density is slowly varying in $C_k$, hence this is where the Local Density Approximation is efficient. In these cubes we take $\rho_k$ to be a constant:
	$$\rho_k=R_k=\max_{C_k}\rho.$$
	We recall that the corresponding $ \mu_k $ satisfies~\eqref{eq:mukdef}. Since $ \rho_k $ is constant, $ v_k = -\mu_k+K \eps_{\ell} $ is constant as well, so we can directly use the convergence rate from \cref{cor:gcconvrate_mu},
	\begin{align*}
		G_T^{v_k} \myp{C_k}={}G_T \myp{-v_k, C_k}
		\geq{}& \ell^d g_T \myp[\big]{\mu_k-K\eps_\ell + C_T(1+\xi^2_{T,\mu_k-K\eps_\ell})\eps_\ell} \\
		\geq{}& \ell^d g_T \myp[\big]{\mu_k-K\eps_\ell + C_T(1+\xi^2_{T,\mu_k})\eps_\ell} \\
		\geq{}& \ell^d g_T \myp{\mu_k}= \ell^d \big(f_T \myp{\rho_k} - \mu_k \rho_k\big).
	\end{align*}
	We used that $ \xi_{T,\mu} $ and $ g_T $ are respectively increasing and non-increasing in $ \mu $.
	Using that $ \mu_k \leq C_T $ and $\rho\leq \rho_k=R_k$, we obtain
	\begin{align}
		G_T^{v_k} \myp{C_k} + \int_{C_k} \mu_k \rho
		\geq{}& \ell^d f_T \myp{\rho_k} + \int_{C_k} \mu_k \myp{\rho-\rho_k} 
		\nn \\
		\geq{}& \ell^d f_T \myp{R_k} -C \int_{C_k}  (R_k-\rho).
	\label{eq:estim_precise}
	\end{align}
	By the argument just below \eqref{eq:simple_cubes}, necessarily
	\begin{equation}
		R_k \geq \eps_\ell
	\label{eq:pty_R_k}
	\end{equation}
	for all the cubes satisfying \eqref{eq:complicated_cubes}.
	Using~\eqref{eq:complicated_cubes}, we then have
	\begin{align}
		R_k-r_k
		\leq{}& \ell^{1-b-\frac{d}p}\eps_\ell^{-\frac{1}{2p}} \myp[\Big]{ \int_{C_k} \rho }^{\frac{1}{p}}
		\leq{} \ell^{1-b}\eps_\ell^{-\frac{1}{2p}} R_k^{\frac1p} \nn \\
		\leq{}& \ell^{1-b}\eps_\ell^{-\frac{1}{2p}} \left(\eps_\ell^{-1}\right)^{\frac{p-1}{p}}R_k
		=\ell^{1-b}\eps_\ell^{-1+\frac{1}{2p}} R_k.
	\label{eq:varies_little}
	\end{align}
	We require that $b$ is so large that the coefficient of the right side tends to zero
	\begin{equation}
		\ell^{1-b}\eps_\ell^{-1+\frac{1}{2p}} \to 0,
	\label{eq:cond_b_ell1}
	\end{equation}
	which can be checked to be true under the condition \eqref{eq:b_cond}.
	Hence, we obtain for $\ell$ large enough
	\begin{equation*}
		R_k-r_k\leq \frac{ R_k}{2},
	\end{equation*}
	which gives
	\begin{equation}
		\eps_\ell \leq R_k \leq 2r_k\leq 2\rho\leq 2R_k,
	\label{eq:main_cubes_unif_bounds}
	\end{equation}
	that is, the (essential) maximum and minimum of $ \rho $ are comparable and not too small in the box $C_k$.

	With this we can go back to~\eqref{eq:estim_precise}. By convexity of $f_T$ we have
	\begin{align*}
	f_T(R_k)
	\geq{}& f_T(\rho)+\mu(\rho)(R_k-\rho)\\
	\geq{}& f_T \myp{\rho} + (T\log\rho-C)(R_k-\rho)
	\geq f_T \myp{\rho}-C\log\ell(R_k-\rho),
	\end{align*}
	due to the lower bound~\eqref{eq:mulowbound} on $\mu(\rho)$ and~\eqref{eq:main_cubes_unif_bounds}. We used here that $\eps_\ell$ is a power of $\ell$ (with an additional logarithm when $s=d+1$).  As usual, the constant depends on $T$ and $M\geq\|\rho\|_\ii$. We have thus proved
	\begin{equation}
		G_T^{v_k} \myp{C_k} + \int_{C_k} \mu_k \rho
		\geq  \int_{C_k}f_T(\rho)-C\log\ell\int_{C_k}(R_k-\rho).
		\label{eq:lipschitz_bound}
	\end{equation}
	On the other hand, we have for $ p > 1 $, using \eqref{eq:pty_R_k},
	\begin{align}
		\log\ell  \int_{C_k} \myp{R_k-\rho} 
		\leq{}&  \ell^{d+1} \delta \rho_{\ell} \myp{\ell k}\log\ell  \nn \\
		={}& \ell^{1-b} \eps_\ell^{\frac{1}{2p}-\frac{1}{2}}\log\ell \myp[\big]{ \sqrt{\eps_\ell}\ell^{d} }^{\frac{p-1}p}\left(\ell^{bp+d}\delta\rho_\ell(\ell k)^p\right)^{\frac1p}\nn\\
		\leq{}& \sqrt{2} \ell^{1-b} \eps_\ell^{\frac{1}{2p}-\frac{1}{2}}\log\ell \myp[\bigg]{\int_{C_k}\sqrt\rho }^{\frac{p-1}p} \myp[\Big]{\ell^{bp+d}\delta\rho_\ell(\ell k)^p }^{\frac1p} \nn \\
		\leq{}& C\ell^{1-b} \eps_\ell^{\frac{1}{2p}-\frac{1}{2}}\log\ell \myp[\bigg]{\int_{C_k}\sqrt\rho+\ell^{bp+d}\delta\rho_\ell(\ell k)^p }.
	\label{eq:log_bound}
	\end{align}
	In the second inequality we used that
	$$\sqrt{\eps_\ell} \ell^{d} \leq \int_{C_k}\sqrt{R_k} \leq \sqrt2\int_{C_k}\sqrt{\rho}$$
	by~\eqref{eq:main_cubes_unif_bounds}. We require that
	\begin{equation}
		\ell^{1-b}\eps_\ell^{\frac{1}{2p}-\frac{1}{2}} \log\ell \leq C\sqrt{\eps_\ell},
	\label{eq:cond_b_ell2}
	\end{equation}
	which implies our other condition~\eqref{eq:cond_b_ell1}, due to the additional logarithm.
	From the definition of $\eps_\ell$ in~\eqref{eq:eps_ell}, the condition~\eqref{eq:cond_b_ell2} is true when
	\begin{equation}
		b > \begin{cases}
			2-\frac{1}{2p}	&\text{if $s\geq d+1$,}\\
			1+(s-d) \myp[\big]{1-\frac{1}{2p}}	&\text{if $d<s<d+1$.}
		\end{cases}
	\end{equation}
	The strict inequality is used to control the logarithm. This is exactly the condition~\eqref{eq:b_cond} from the statement. When $p=1$ we also get the desired bound, using the strict inequality in the condition on $b$.
	This means that for the ``main'' cubes satisfying \eqref{eq:complicated_cubes}, we obtain from \eqref{eq:estim_precise}
	\begin{equation}
		 G_T^{v_k} \myp{C_k} + \int_{C_k} \mu_k \rho \geq  \int_{C_k} f_T \myp{\rho \myp{x}} \id x
		- C\sqrt{\eps_\ell} \myp[\bigg]{ \int_{C_k} \sqrt{\rho} +\ell^{bp+d} \delta \rho_{\ell} \myp{\ell k}^p }.
	\label{eq:estim_complicated_final}
	\end{equation}
	Combining \eqref{eq:ldalow2} with \eqref{eq:estim_simple_final} and \eqref{eq:estim_complicated_final}, we have proven after summing over $k$
	\begin{align}
		\MoveEqLeft[6] G_T \myb{\rho} - \int_{\R^d} f_T \myp{\rho \myp{x} } \id x \nn \\
		\geq{}& \sum\limits_{\substack{k, m \in \Z^d \\ k \neq m}} \expec[\big]{I_{k,m}}_{\bP} - C\sqrt{\eps_\ell} \myp[\bigg]{ \int_{\R^d} \sqrt{\rho} +\ell^{bp+d} \delta \rho_{\ell} \myp{\ell k}^p }.
	\label{eq:intermediate_lowbound}
	\end{align}
	It remains to provide a bound on the interaction terms.

\medskip
\noindent\textbf{Bound on the interactions at $T>0$.}
We claim that
	\begin{equation}
		\sum\limits_{\substack{k,m \in \Z^d \\ k \neq m}} \expec{I_{k,m}}_{\bP}
		\geq{} - C\sqrt{\eps_\ell} \myp[\bigg]{ \int_{\R^d} \rho + \int_{\R^d} \sqrt{\rho } + \ell^{bp + d} \sum\limits_{k \in \Z^d} \delta\rho_{\ell}\myp{\ell k}^p }.
		\label{eq:ldalowinteraction}
	\end{equation}
The proof will make use of a local version of the Ruelle bound \eqref{eq:ruelle} in the presence of an external potential $ V $, which is bounded from below by a constant $-\mu_0$.
Precisely, there exist $ C, L_0 > 0 $, such that for any cube $ Q $ of side length $ L_0 $, we have
\begin{equation}
	\label{eq:ruelleloc}
	\expec[\big]{n_{Q}^2}_{T,V}
	\leq{} C |Q|\int_{Q} e^{-\frac{1}{T} V \myp{x}} \id x\left(1+e^{\frac{d\mu_0}{\eps T}}\right).
\end{equation}
The side length $ L_0 $ of the cube $ Q $ can be chosen arbitrarily large and the constant $ C $ depends on $L_0$ and the temperature $T$. The estimate~\eqref{eq:ruelleloc} is in Corollary~\ref{thm:Ruelle_V} in Appendix~\ref{app:Ruelle}. When $\int_{Q} e^{-\frac{1}{T} V \myp{x}}$ is small, it can be proved the same as in~\cite[Chap.~4]{Ruelle} and~\cite{JanKunTsa-22}.

First, we recall that the potential $ V $ is given by \eqref{eq:vdef}, and that $ \bP $ is the Gibbs state minimizing the problem in \eqref{eq:ldalow1}.
	Since $ \ell \geq L_0 $, we can choose an $ r\in \myb{L_0; 2L_0} $ such that each $ C_k $ is exactly a disjoint union of cubes of side length $ r $,
	\begin{equation*}
		C_k = \bigcup_{\gamma \in \mathcal{A}} Q_{k,\gamma}.
	\end{equation*}
	Since $ V = - \mu_k + K\eps_\ell $ in each of the smaller cubes $ Q_{k, \gamma} $, we get by throwing away the positive part of the interaction and repeating the calculation in~\eqref{eq:interactbound},
	\begin{align}
		\sum_{k\neq m}\expec{I_{k,m}}_{\bP}
		\geq{}& - \sum_{k\neq m}\sum\limits_{\gamma, \gamma' \in \mathcal{A}} \frac{\kappa}{1+ \rd \myp[\big]{Q_{k,\gamma}, Q_{m,\gamma'}}^s } \expec[\big]{n_{Q_{k,\gamma}}^2}_{T, V}^{\frac{1}{2}} \expec[\big]{n_{Q_{m,\gamma'}}^2}_{T, V}^{\frac{1}{2}} \nn\\
		\geq{}& - \sum_{k\neq m}\sum\limits_{\gamma, \gamma' \in \mathcal{A}} \frac{\kappa}{1+ \rd \myp[\big]{Q_{k,\gamma}, Q_{m,\gamma'}}^s } \expec[\big]{n_{Q_{k,\gamma}}^2}_{T, V}\nn\\
		\geq{}& - c_M \sum_{k}\sum_{\gamma\in \mathcal{A}} \int_{Q_{k,\gamma}} \int_{(C_k)^c} \frac{\dx\,\dy}{1 + \abs{x-y}^s} \int_{Q_{k,\gamma}} \rho_k.
		\label{eq:estim_interaction}
		\end{align}
Here we have used that
$$2\expec[\big]{n_{Q_{k,\gamma}}^2}_{T, V}^{\frac{1}{2}} \expec[\big]{n_{Q_{m,\gamma'}}^2}_{T, V}^{\frac{1}{2}}\leq \expec[\big]{n_{Q_{k,\gamma}}^2}_{T, V}+ \expec[\big]{n_{Q_{m,\gamma'}}^2}_{T, V}$$
and then
$$\expec[\big]{n_{Q_{k,\gamma}}^2}_{T, V}\leq C\int_{Q_{k,\gamma}}e^{\mu_k/T}\leq C\int_{Q_{k,\gamma}}\rho_k$$
due to the Ruelle bound~\eqref{eq:ruelleloc}, and the estimates on $\mu_k$ in terms of $\rho_k$ from~\cref{thm:mubounds}. We also used~\eqref{eq:integralbound2} to replace $1/(1+\rd(Q_{k,\gamma},Q_{m,\gamma'})^s)$ by an integral.

Next we have two cases. In the ``simple'' cubes $C_k$ we just estimate
$$\int_{(C_k)^c} \frac{\dx\,\dy}{1 + \abs{x-y}^s}\leq \int_{\R^d} \frac{\dx\,\dy}{1 + \abs{x-y}^s}\leq C$$
and obtain
$$\sum_{\gamma\in \mathcal{A}} \int_{Q_{k,\gamma}} \int_{(C_k)^c} \frac{\dx\,\dy}{1 + \abs{x-y}^s} \int_{Q_{k,\gamma}} \rho_k\leq C\int_{C_k}\rho$$
which can be bounded by the desired error terms due to the definition~\eqref{eq:simple_cubes} of the simple cubes. In the ``main'' cubes, we use that $\rho_k\equiv R_k$ is constant and thus get
\begin{align*}
\MoveEqLeft[4] \sum_{\gamma\in \mathcal{A}} \int_{Q_{k,\gamma}} \int_{(C_k)^c} \frac{\dx\,\dy}{1 + \abs{x-y}^s} \int_{Q_{k,\gamma}} \rho_k\\
={}& R_k r^d\int_{C_k}\int_{(C_k)^c} \frac{\dx\,\dy}{1 + \abs{x-y}^s}
\leq c r^d\eps_\ell \ell^d R_k\leq c 2r^d\eps_\ell\int_{C_k}\rho,
\end{align*}
by Lemma~\ref{lem:integralbound} and the fact that $\rho\geq R_k/2$ on $C_k$. We recall that $r$ is the (fixed) side length of the small cubes $Q_{k,\gamma}$. We have therefore proved~\eqref{eq:ldalowinteraction}.

	\medskip
\noindent\textbf{Concluding the proof for $T>0$.}
	Inserting the estimate~\eqref{eq:ldalowinteraction} into~\eqref{eq:intermediate_lowbound}, we finally conclude that
	\begin{multline}
		G_T \myb{\rho} - \int_{\R^d} f_T \myp{\rho \myp{x} } \id x
		\geq{} -C\sqrt{\eps_\ell} \myp[\bigg]{ \int_{\R^d} \rho + \int_{\R^d} \sqrt{\rho } + \ell^{bp+d} \sum\limits_{k\in \Z^d} \delta\rho_{\ell}\myp{\ell k}^p }.
	\label{eq:ldalow_sum}
	\end{multline}
	The terms in this inequality are all invariant under translations except for the last sum, due to our initial choice of the tiling of space.
	However, since~\eqref{eq:ldalow_sum} holds for all densities $\rho$ we can freely average the bound over translations $\rho(\cdot +\tau)$ with $\tau\in C_\ell$, which amounts to averaging over the position of the tiling.
	This is how we obtain the bound~\eqref{eq:ldalow} in Proposition~\ref{thm:ldalow} for $T>0$.

\medskip
\noindent\textbf{Zero temperature case.}
We conclude the proof by explaining how to treat the case $T=0$.
Since $\rho\mapsto \mu(\rho)$ is not necessarily single-valued, we can for instance instead take $\mu_k$ to be the largest value of $\mu$ such that~\eqref{eq:mukdef} holds. In addition, we do not insert any shift and instead define $v_k=-\mu_k$ for all $k$.

We take the same definition of the ``simple'' and ``main'' cubes as in the $T>0$ case, i.e. the cubes satisfying \eqref{eq:simple_cubes} and \eqref{eq:complicated_cubes}, respectively.
In the ``main'' cubes (those for which~\eqref{eq:complicated_cubes} holds) we take $\mu_k$ to be maximal such that \eqref{eq:mukdef} holds with $\rho_k=R_k:=\max_{C_k}\rho$. In the ``simple'' cubes (where we again have $R_k\leq\eps_\ell$) we do not take $\rho_k=\rho$ as we did for $T>0$, but rather choose $ \mu_k = -C $ sufficiently negative, such that
$$ G_0^{v_k}\myp{C_k}  = 0,\qquad \pscal{n_{C_k}}_\bP=0.$$
In other words, we enforce that there is no particle at all in the simple cubes. For the first condition it suffices that $C\geq\kappa$, the stability constant. For the second condition (recall that $\bP$ is the minimizer with the external potential $V$), we have to use the Ruelle bound~\eqref{eq:main_estim_Ruelle_thm_V_T0} which provides the existence of such a constant $C$, depending on $M$.

For the simple cubes, we are thus left with
\begin{multline*}
G_0^{v_k}\myp{C_k} + \int_{C_k} \mu_k \rho - \int_{C_k} f_0 \myp{\rho}= -C\int_{C_k} \rho - \int_{C_k} f_0 \myp{\rho}\\
\geq -C_M\int_{C_k} \rho\geq-C_M\sqrt{\eps_\ell}\left( \int_{C_k} \sqrt{\rho \myp{x}} \id x + \ell^{bp+d} \delta\rho_\ell \myp{\ell k}^p\right),
\end{multline*}
since at $T=0$ we have $|f_0(\rho)|\leq C_M\rho$ by~\cref{thm:tdlimbounds}.
In the main cubes we again have $ R_k\geq\eps_\ell$ and~\eqref{eq:varies_little}, so~\eqref{eq:main_cubes_unif_bounds} still holds for $\ell$ sufficiently large. However, because $ f_0'\myp{\rho} \geq -C $,~\eqref{eq:log_bound} is replaced by
\begin{equation}
\int_{C_k}(R_k-\rho)\leq C\ell^{1-b}\eps_\ell^{-\frac12+\frac{1}{2p}}\left(\int_{C_k}\sqrt\rho+\ell^{bp+d}\delta\rho_\ell( \ell k)^p\right),
\label{eq:log_bound_T0}
\end{equation}
that is, the logarithm in \eqref{eq:log_bound} disappears.
The condition we need on $b$ is therefore the same \eqref{eq:cond_b_ell2} as in the $T > 0$ case (but without the $\log \ell$),
which is true under~\eqref{eq:b_cond}.
This allows us to argue as for $T>0$ to estimate the local energy for the main cubes.
One difference is that we use instead the \emph{additive} convergence rate in~\eqref{eq:gcconvrate}, which gives
\begin{multline*}
G_0^{v_k}\myp{C_k} + \int_{C_k} \mu_k \rho - \int_{C_k} f_0 \myp{\rho}\\
\geq -C\int_{C_k} (R_k-\rho) -C\ell^d\eps_\ell\geq -C\int_{C_k} (R_k-\rho) -2C\sqrt{\eps_\ell}\int_{C_k} \sqrt{\rho},
\end{multline*}
where in the last bound we have used that $\rho\geq R_k/2\geq \eps_\ell/2$ in the box $C_k$.
Inserting~\eqref{eq:log_bound_T0} leads to the expected error term after summing over $k$.

Finally, the interaction is estimated similarly as for $T>0$, but we use that there is just no ``simple'' cube to consider, since there is no particle at all there. Like for Corollary~\ref{cor:gcconvrate_mu}, the main difficulty is that $\xi_{0,\mu}$ cannot easily be controlled in terms of $\rho$ at low density, as we used at $T>0$ in~\eqref{eq:estim_interaction}. Therefore we use instead that $\rho_k$ is uniformly bounded over the main cubes and get similarly as in the proof of Proposition~\ref{prop:gcconvrate}
\begin{align*}
\sum\limits_{\substack{k,m \in \Z^d \\ k \neq m}}\expec{I_{k,m}}_{\bP}\geq -C\!\!\sum_{\text{main cubes}}\ell^{d}\eps_\ell&\geq -2C\!\!\sqrt{\eps_\ell}\sum_{\text{main cubes}}\int_{C_k}\sqrt\rho\\
&\geq  -2C\sqrt{\eps_\ell}\int_{\R^d}\sqrt \rho
\end{align*}
as desired. This concludes the proof of Proposition~\ref{thm:ldalow}.
\end{proof}

\subsection{Upper bound}

\begin{proposition}[LDA from above]
	\label{thm:ldaup}
	Let $ M > 0 $ and $ p \geq 1 $, and let $ w $ be an interaction satisfying \cref{de:shortrangenew}.
	Let $T\geq0$ and assume furthermore that $s > d+1$ if $T=0$. Let
	\begin{equation}
	\label{eq:b_cond2}
		b > \begin{cases}
				2 - \frac{1}{2p} & \text{if } p \geq 2, \\
				\frac{3}{2}+\frac{1}{2p} & \text{if } 1 \leq p < 2.
			\end{cases}
	\end{equation}
	There exists a constant $ C > 0 $ depending on $ M,p,b,T,d $, and $ w $, such that
	\begin{equation}
		G_T \myb{\rho} - \int_{\R^d} f_T \myp{\rho \myp{x}} \id x \\
		\leq  C\eta_\ell\myp[\bigg]{\int_{\R^d}\sqrt \rho + \ell^{bp} \int_{\R^d} \delta\rho_\ell(z)^p\,\rd z},
	\label{eq:lda_up}
	\end{equation}
	for any $ \ell > 0 $, and any density $ 0 \leq \rho \in L^1 \myp{\R^d} $ satisfying $ \normt{\rho}{\infty} \leq M $ and $ \sqrt{\rho} \in L^1 \myp{\R^d} $.
\end{proposition}

Recall that $\eta_\ell$ is defined in~\cref{prop:gcconvrate_fixed}. For $s>d+1$ we just have $\eta_\ell=\sqrt{\eps_\ell}=1/\sqrt\ell$ so that~\eqref{eq:lda_up} takes the same form as the lower bound in Proposition~\ref{thm:ldalow}. The constraint on $b$ is slightly worse when $1\leq p<2$, however. For $T=0$ we can also handle $d<s\leq d+1$ but we get an error worse than $\eta_\ell$. For simplicity we do not give the details.

\begin{proof}
The proof uses very similar arguments as for the lower bound, but with difficulties at different places. The interaction is much easier to treat but the local free energy is slightly trickier to estimate. Again, due to the universal bounds proved in~\cite{JexLewMad-23}, the estimate~\eqref{eq:ldalow} holds when $\ell$ is finite, even without the term involving $\delta\rho_\ell$. We can thus assume that $\ell$ is large enough.

The first step is to localize the problem into cubes of side length $ \ell > 0 $, with corridors. Following~\cite{LewLieSei-20}, we pick a $\chi$ supported on the cube $Q_\ell:=[-\ell/2,\ell/2)^d$ with $\int\chi=\ell^d$ and use the relation
\begin{equation}
\frac{1}{\ell^d} \int_{Q_{\ell}} \sum_{k \in \Z^d} \chi(x-\ell k-\tau) \id \tau=1,\qquad\forall x\in\R^d,
 \label{eq:partition_unity}
\end{equation}
which we interpret as a continuous partition of unity. We automatically have corridors if $\chi$ is supported well inside the cube and thus take
$$\chi:=\frac{\1_{C_0}}{\myp{1-\delta/\ell}^d},\qquad C_0:=\left[-\frac{\ell-\delta}2,\frac{\ell-\delta}2\right)^d.$$
where $ \delta \geq r_0 $ (the range of the core of $ w $). Then we let $\chi_{k}:=\chi(\cdot-\ell k)$ and deduce that
	\begin{equation*}
		\rho \myp{x}
		= \frac{1}{\ell^d} \int_{Q_\ell} \sum_{k \in \Z^d} \chi_{k} \myp{x-\tau} \rho \myp{x} \id \tau
		=: \frac{1}{\ell^d} \int_{Q_\ell} \sum_{k \in \Z^d} \rho_{k,\tau} \myp{x} \id \tau.
	\end{equation*}
	For each $ k \in \Z^d $ and $ \tau \in Q_{\ell} $, we let $ \bP_{k,\tau} $ be the state minimizing the energy at fixed density $ \rho_{k,\tau} $ and take as a trial state
	\begin{equation*}
		\bP := \frac{1}{\ell^d} \int_{Q_{\ell}} \bigotimes\limits_{k \in \Z^d} \bP_{k,\tau} \id \tau.
	\end{equation*}
	Then, from the concavity of the entropy we have
	$$G_T \myb{\rho}\leq \mathcal{G}_T \myp{\bP}\leq \frac{1}{\ell^d} \int_{Q_{\ell}}  \bigg\{\sum\limits_{k \in \Z^d} \mathcal{G}_T \myp{\bP_{k,\tau}} + \sum\limits_{\substack{k,m \in \Z^d \\ k \neq m}} 2 D_w \myp{\rho_{k,\tau}, \rho_{m,\tau}}\bigg\}\rd \tau.$$
From~\eqref{eq:partition_unity}, we can also write
$$ \int_{\R^d} f_T \myp{\rho \myp{x}} \id x=\frac{1}{\ell^d} \int_{Q_{\ell}} \int_{\R^d} \sum\limits_{k \in \Z^d} \chi_{k} \myp{x-\tau} f_T \big(\rho \myp{x}\big) \id x\,\rd \tau.$$
Since everything will be done at fixed $ \tau $, for simplicity we suppress $\tau$ from the notation. We also denote by $C_k=C_{0}+\ell k$ the translated cube. With these notations we have to estimate
$$G_T[\rho_{k,\tau}]-\int \chi_{k} \myp{\cdot -\tau} f_T \myp{\rho}=G_T\left[\frac{\rho\1_{C_k}}{(1-\delta/\ell)^d}\right]-\frac1{(1-\delta/\ell)^d}\int_{C_k}  f_T \myp{\rho}.$$

As for the lower bound, we then split the cubes into two categories, depending whether $\rho$ varies too much or not. We thus consider all the cubes satisfying the simple estimate
\begin{equation}
	\label{eq:up_simple_cubes}
		\int_{C_k} \rho
		\leq{} \eta_\ell \left( \int_{C_k} \sqrt{\rho} + \ell^{bp+d} \delta\rho_{\ell}\myp{\ell k}^p\right),
	\end{equation}
where $\eta_\ell$ is from Proposition~\ref{prop:gcconvrate_fixed}, which we call the ``simple cubes''. In those cubes we use the upper bound \eqref{eq:upper_bd_G_T_simple} on $ G_T $ and the lower bound \eqref{eq:fcanlowbound} on $ f_T $ to obtain the simple estimate
	\begin{align}
		\MoveEqLeft[3] G_T\left[\rho_k\right]-\frac1{(1-\delta/\ell)^d}\int_{C_k}  f_T \myp{\rho} \nn \\
		\leq{}& C \int_{C_k} \rho_k + T \int_{C_k} \rho_k \log \rho_k + \int_{C_k} \frac{1}{\myp{1-\delta / \ell}^d} \myp[\big]{ C \rho - T \rho \log \rho} \nn \\
		={}& 2C \int_{C_k} \rho_k + T \int_{C_k} \rho_k \myp[\big]{ \log \rho_k - \log \rho}\nn\\
		={}&\frac{2C-Td\log(1-\delta/\ell)}{(1-\delta/\ell)^d} \int_{C_k} \rho\leq C\int_{C_k}\rho.
	\label{eq:up_estim_simple}
	\end{align}
From the definition~\eqref{eq:up_simple_cubes} of the simple cubes, we get the desired error term for those cubes.

Next we look at the ``main cubes'' for which
\begin{equation}
	\label{eq:up_complicated_cubes}
		\int_{C_k} \rho
		\geq{} \eta_\ell\left( \int_{C_k} \sqrt{\rho} + \ell^{bp+d} \delta\rho_{\ell}\myp{\ell k}^p\right),
	\end{equation}
We denote again
	\begin{equation*}
		r_k:=\min_{C_k} \rho,
		\qquad R_k:=\max_{C_k} \rho.
	\end{equation*}
In these cubes, we must have
	\begin{equation}
	\label{eq:up_pty_R_k}
		\sqrt{R_k}\geq \eta_\ell.
	\end{equation}
	The difference $ R_k - r_k $ is bounded in terms of $ R_k $ exactly as in \eqref{eq:varies_little},
	\begin{align}
	\label{eq:up_varies_little}
		R_k - r_k 	
		\leq{}& \ell^{1-b}\eta_\ell^{-2+\frac1p} R_k,
	\end{align}
	while on the other hand, the same argument also gives
	\begin{equation*}
		R_k - r_k
		\leq{}	\begin{cases}
					\ell^{1-b} \eta_{\ell}^{-1+\frac{1}{p}} \sqrt{R_k} & \text{if } p \geq 2 \\
					M^{\frac{2-p}{2p}} \ell^{1-b} \eta_{\ell}^{-\frac{1}{p}} \sqrt{R_k} & \text{if } 1 \leq p < 2.
				\end{cases}
	\end{equation*}
	If we require that
	\begin{equation*}
		\max\myp[\Big]{\ell^{1-b} \eta_\ell^{-1+\frac1p}, \ell^{1-b} \eta_{\ell}^{-\frac{1}{p}}}
		\leq{} \frac{C \eta_{\ell}}{\log \ell},
	\label{eq:cond_b_ell1_up}
	\end{equation*}
	which is true for any value of $s > d$ under our choice of $b$ in~\eqref{eq:b_cond2}, then we obtain in either case that the right hand side of \eqref{eq:up_varies_little} tends to zero for large $\ell$, and
	\begin{equation}
	\label{eq:up_varies_little2}
		R_k - r_k
		\leq{} \frac{C \eta_{\ell} \sqrt{R_k}}{\log \ell}.
	\end{equation}
	Note the additional $\eta_\ell/\log\ell$ on the right of~\eqref{eq:up_varies_little2} compared to the similar estimate~\eqref{eq:varies_little} in the proof of the lower bound.
	From \eqref{eq:up_varies_little} we obtain again
	\begin{equation*}
		\eta_\ell^2 \leq{} R_k \leq{} 2 r_k \leq 2 \rho \leq 2R_k
	\end{equation*}
	for $ \ell $ large enough.
	Furthermore, \eqref{eq:up_varies_little2} implies the pointwise bound
	\begin{equation}
		\rho_k - r_k
		={} \frac{1}{\myp{1- \delta / \ell}^d} \rho - r_k
		\leq{} \rho - r_k + c \frac{\delta}{\ell} \rho
		\leq{} C\frac{\eta_\ell \sqrt{R_k}}{\log\ell},
		\label{eq:estim_r_k_rho}
	\end{equation}
	since $\log\ell/\ell\leq \eta_\ell$.
	In the ``main'' cubes satisfying \eqref{eq:up_complicated_cubes}, we will replace $ \rho_k $ by the minimum $ r_k $ of $ \rho $ in the cube $ C_k $, using the sub-additivity bound in \cref{prop:subadditivity}.
	More precisely, choosing $ \eps = {\eta_\ell}/{\log\ell}$ in \eqref{eq:subadditivity} yields in the case $ \alpha \neq d $,
	\begin{align}
		G_T \myb{\rho_k}
		={}& G_T \myb{\myp{r_k + \myp{\rho_k - r_k}} \1_{C_k}} \nn \\
		\leq{}& G_T \myb{r_k \1_{C_k}} + C \myp[\Big]{ 1+\log \myp[\Big]{ \frac{\eta_\ell}{\log\ell} }_- } \frac{\eta_\ell}{\log\ell} \sqrt{R_k} \ell^d\nn \\
		\leq{}& G_T \myb{r_k \1_{C_k}} + C\eta_\ell\int_{C_k} \sqrt{\rho} 
		\leq{} \abs{C_k} f_T \myp{r_k} + C\eta_\ell\int_{C_k} \sqrt{\rho}.
	\label{eq:up_estim_precise}
	\end{align}
	In the last line we have used the upper bound on the convergence rate from Proposition~\ref{prop:gcconvrate_fixed} (and hence also the assumption $s > d+1$ when $T = 0$).
	The same bound can be obtained in the $ \alpha = d $ case by using \eqref{eq:subadditivity2}.
	Similarly, the lower bound \eqref{eq:fcanlowbound} on $ f_T $ and $\rho \geq \eta_{\ell}^2 / 2$ in $C_k$ imply
	\begin{equation*}
		- \int_{C_k} f_T \myp{\rho}\leq C\begin{cases}
											\log\ell \int_{C_k}\rho&\text{if $T>0$,}\\
											\int_{C_k}\sqrt\rho&\text{if $T=0$.}
		                                 \end{cases}
	\end{equation*}
	Using $\log\ell/\ell\leq\eta_\ell$, we obtain in either case
	\begin{equation*}
		G_T \myb{\rho_k} - \frac{1}{\myp{1-\delta / \ell}^d} \int_{C_k} f_T \myp{\rho}
		\leq{} \int_{C_k}(f_T \myp{r_k} - f_T \myp{\rho}) +C\eta_\ell \int_{C_k}\sqrt{\rho}.
	\end{equation*}
	By \eqref{eq:up_pty_R_k} and the bounds on $ f_T'(\rho) $, we have as in \eqref{eq:estim_r_k_rho},
	\begin{equation*}
		\int_{C_k} f_T \myp{r_k} - f_T \myp{\rho}
		\leq{} C \log\ell  \int_{C_k} \myp{\rho - r_k}
		\leq{} C \eta_\ell\int_{C_k} \rho.
	\end{equation*}
	We conclude that for any ``main'' cube $ C_k $ where \eqref{eq:up_complicated_cubes} holds, we have
	\begin{equation*}
		G_T \myb{\rho_k} - \frac{1}{\myp{1-\delta / \ell}^d} \int_{C_k} f_T \myp{\rho}
		\leq{} C \eta_\ell \int_{C_k}\sqrt{\rho}.
	\end{equation*}

It remains to deal with the interaction terms. Due to the corridors, we have
$$D_w \myp{\rho_{k,\tau}, \rho_{m,\tau}}\leq \kappa\iint  \frac{\rho_{k,\tau}(x)\rho_{m,\tau}(y)}{1+|x-y|^s}\dx\,\dy.$$
Hence, using the fact that $\rho\leq M$, we have
$$\sum_{m\in\Z^d\setminus\{k\}}D_w \myp{\rho_{k}, \rho_{m}}\leq C\int_{C_k}\int_{(C_k)^c}  \frac{\rho(x)}{1+|x-y|^s}\dy\,\dx.$$
In the simple cubes we just bound the previous integral by $\int_{C_k}\rho$ and then use~\eqref{eq:up_simple_cubes} to obtain $\eta_\ell\int_{C_k}\sqrt\rho$. In the main cubes satisfying~\eqref{eq:up_complicated_cubes}, we rather use Lemma~\ref{lem:integralbound} and obtain
$$\int_{C_k}\int_{(C_k)^c}  \frac{\rho(x)}{1+|x-y|^s}\dy\,\dx\leq C\eps_\ell \ell^d R_k\leq 2C\eps_\ell \int_{C_k}\rho.$$
We arrive at the desired estimate after integrating over $\tau$.
\end{proof}

\section{Uniform bounds on energy per unit volume}
\label{sec:tdlim_proofs}

This section contains the proofs of Proposition~\ref{thm:tdlimbounds} and Corollaries~\ref{thm:mubounds} and~\ref{thm:gcbounds}.

\begin{proof}[Proof of \cref{thm:tdlimbounds}]
	We note first that because $ f_T \myp{\rho,L} $ is given by the minimization problem \eqref{eq:fT_variational}, the upper bounds \eqref{eq:fcanupbound} follow easily from the universal bounds \eqref{eq:upper_bd_G_T1} and \eqref{eq:upper_bd_G_T2} by simply using the density $ \rho \1_{C_L} $,
	\begin{equation*}
		f_T \myp{\rho,L}
		={} \inf_{ \substack{ \supp \bP \subseteq C_L \\ \mathcal{N}(\bP) = \rho L^d }} \frac{\mathcal{G}_T \myp{\bP}}{L^d}
		\leq{} \min_{\rho_{\bP} = \rho \1_{C_L}} \frac{\mathcal{G}_T \myp{\bP}}{L^d}
		={} \frac{G_T \myb{ \rho \1_{C_L}}}{L^d}.
	\end{equation*}
	Thus we need only to prove the lower bound \eqref{eq:fcanlowbound}.
	For this, we use \eqref{eq:fT_variational} and split the minimization problem
	\begin{equation*}
		f_T \myp{\rho,L}
		= \inf_{ \substack{ \supp \bP \subseteq C_L \\ \mathcal{N}(\bP) = \rho L^d }} \frac{\mathcal{G}_T \myp{\bP}}{L^d}
		\geq \inf_{ \substack{ \supp \bP \subseteq C_L \\ \mathcal{N}(\bP) = \rho L^d }} \frac{\cU \myp{\bP}}{L^d} - \sup_{ \substack{ \supp \bP \subseteq C_L \\ \mathcal{N}(\bP) = \rho L^d }} \frac{T\cS \myp{\bP}}{L^d},
	\end{equation*}
	where $ \mathcal{S} \myp{\bP} $ under the stated conditions is maximized by the Poisson state \cite[Lemma $ 6.1 $]{MarLewNen-22_ppt} with density $ \rho \1_{C_L} $, with
	\begin{equation}
	\label{eq:tdlim_entropy}
		\sup_{ \substack{ \supp \bP \subseteq C_L \\ \mathcal{N}(\bP) = \rho L^d }} \cS \myp{\bP}
		\leq L^d \myp{\rho - \rho \log \rho}.
	\end{equation}
	Thus we proceed to derive a bound on the interaction energy $ \mathcal{U} \myp{\bP} $.
	We write the cube $ C_L $ as a union of smaller cubes of side length $ 0 < r < \sqrt{d}^{-1} r_0 $,
	\begin{equation*}
		C_L = \bigcup\limits_{k \in \mathcal{A}} C_k.
	\end{equation*}
	Given a configuration $ x = \Set{x_1, \dotsc, x_n} $ of $ n $ particles in $ C_L $, we denote by $ n_k = \# \myp{x \cap C_k} $ the number of particles in the cube $ C_k $.
	It follows from \cref{de:shortrangenew} that for any $ n \geq 1 $,
	\begin{align}
		\sum\limits_{i < j}^n w \myp{x_i - x_j}
		\geq{}& \sum\limits_{i < j}^n \frac{r_0^{\alpha}}{\kappa} \frac{\1 \myp{\abs{x_i-x_j}<r_0}}{\abs{x_i-x_j}^{\alpha}} - \kappa n \nn \\
		\geq{}& \sum\limits_{k \in \mathcal{A}} \sum\limits_{\substack{x_i, x_j \in C_k \\ i \neq j}} \frac{1}{\kappa d^{\alpha/2}} \frac{r^{\alpha}}{\abs{x_i - x_j}^{\alpha}} - \kappa n \nn \\
		\geq{}& \sum\limits_{k \in \mathcal{A}} \frac{c}{\kappa} n_k^{\gamma} - \myp[\Big]{\kappa + \frac{c}{\kappa}} n.
	\label{eq:low_interact}
	\end{align}
	We have used here the well-known inequality
	\begin{equation}
		\min_{x_i \in C_r} \sum\limits_{i < j}^n \frac{r^{\alpha}}{\abs{x_i - x_j}^{\alpha}}
		\geq c n^{\gamma} - c n
	\end{equation}
	which holds for $ \alpha \neq d $ in any cube of side length $r$, see for instance~\cite[Lemma 1]{Lewin-22} and~\cite{HarSaf-05} (when $\alpha=d$ there is even a lower bound involving $n^2\log(n)$).
	It follows for any state $ \bP $ in $ C_L $ with $ \mathcal{N} \myp{\bP} = \rho L^d $ that we can estimate, using Jensen's inequality,
	\begin{align*}
		\mathcal{U} \myp{\bP}
		\geq{}& \sum\limits_{n \geq 0} \frac{1}{n!} \int_{C_L^n} \sum\limits_{k \in \mathcal{A}} \frac{c}{\kappa} n_k^{\gamma} - \myp[\Big]{\kappa + \frac{c}{\kappa}} n \id \bP_n \\
		\geq{}& \sum\limits_{n \geq 0} \frac{1}{n!} \int_{C_L^n} \frac{c}{\kappa} \abs{\mathcal{A}}^{1-\gamma} n^{\gamma} \id \bP_n - \myp[\Big]{\kappa + \frac{c}{\kappa}} \mathcal{N}\myp{\bP} \\
		\geq{}& \frac{c}{\kappa} \abs{\mathcal{A}}^{1-\gamma} \mathcal{N}\myp{\bP}^{\gamma} - \myp[\Big]{\kappa + \frac{c}{\kappa}} \mathcal{N}\myp{\bP}
		={} \frac{c r^{d \myp{\gamma-1}}}{\kappa} L^d \rho^{\gamma} - \myp[\Big]{\kappa + \frac{c}{\kappa}} L^d \rho.
	\end{align*}
	Combining with \eqref{eq:tdlim_entropy} and choosing for instance $ r \geq \frac{r_0}{2 \sqrt{d}} $, we finally conclude for $ \alpha \neq d $ that
	\begin{equation*}
		f_T \myp{\rho,L} \geq \frac{c r_0^{d \myp{\gamma-1}}}{\kappa} \rho^{\gamma} - \myp[\big]{\kappa + \frac{c}{\kappa} +T}\rho + T \rho \log \rho.
	\end{equation*}

	In the case $ \alpha = d $, we get in the same way,
	\begin{align*}
		\mathcal{U} \myp{\bP}
		\geq{}& \sum\limits_{n \geq 0} \frac{1}{n!} \int_{C_L^n} \sum\limits_{k \in \mathcal{A}} \frac{c}{\kappa} n_k^2 \log n_k - \kappa n \id \bP_n \\
		\geq{}& \frac{c \abs{\mathcal{A}}}{\kappa} \myp[\bigg]{\frac{\mathcal{N}\myp{\bP}}{\abs{\mathcal{A}}}}^2 \log \myp[\bigg]{\frac{\mathcal{N}\myp{\bP}}{\abs{\mathcal{A}}}}_+ - \kappa \mathcal{N} \myp{\bP}
		={} \frac{c r^d L^d}{\kappa} \rho^2 \myp[\big]{\log r^d \rho}_+ - \kappa L^d \rho.
	\end{align*}
	Again choosing $ r \geq \frac{r_0}{2 \sqrt{d}} $, we conclude that \eqref{eq:fcanlowbound} holds.
\end{proof}

\begin{proof}[Proof of \cref{thm:mubounds}]
	Since $ f_T \myp{\rho,L} $ is a convex function, we have for all $ \rho, \widetilde{\rho} \geq 0 $,
	\begin{equation*}
		f_T \myp{\widetilde{\rho},L}
		\geq f_T \myp{\rho,L} + f_T' \myp{\rho,L} \myp{\widetilde{\rho}- \rho}.
	\end{equation*}
	Taking $ \widetilde{\rho} = 0 $ and using $ f_T \myp{0,L}= 0 $ immediately gives $ \mu_L \myp{\rho}=f_T' \myp{\rho,L} \geq \frac{f_T \myp{\rho,L}}{\rho} $, so combining with the lower bounds on $ f_T $ from \cref{thm:tdlimbounds} gives the stated lower bounds \eqref{eq:mulowbound} on $ \mu_L \myp{\rho} $.

	For the upper bounds, we simply take $ \widetilde{\rho} = 2 \rho $ and obtain
	\begin{equation*}
		\mu_L \myp{\rho}=f_T' \myp{\rho,L} \leq \frac{f_T \myp{2\rho,L}-f_T \myp{\rho,L}}{\rho}.
	\end{equation*}
	Using both the upper and lower bounds on $ f_T $ in \cref{thm:tdlimbounds} yields the upper bounds \eqref{eq:muupbounds}. At $T=0$ the argument is the same, using any derivative in the interval $\myb[\big]{ \frac{\partial}{\partial \rho_-} f_0(\rho,L), \frac{\partial}{\partial \rho_+} f_0(\rho,L) }$.
\end{proof}

\begin{proof}[Proof of \cref{thm:gcbounds}]
	Dividing $ C_L $ into smaller cubes of side length $ r \leq \sqrt{d}^{-1} r_0 $, we follow the approach of the proof of the lower bound in \cref{thm:tdlimbounds} and obtain as in \eqref{eq:low_interact} for $ \alpha \neq d $,
	\begin{align}
		\sum\limits_{i < j}^n w \myp{x_i - x_j} - \mu n
		\geq{}& \sum\limits_{k \in \mathcal{A}} \frac{c}{\kappa} n_k^{\gamma} - \myp[\big]{\kappa + \frac{c}{\kappa}+ \mu} n_k \nn\\
		\geq{}& - L^d K \myp[\big]{\kappa + \frac{c}{\kappa} + \mu}_+^{\frac{\gamma}{\gamma-1}},
	\label{eq:bound_using_EN}
	\end{align}
	where the constant $ K $ depends on $ \kappa $, $ \alpha $, $ d $, and $ r_0 $.
	We use the bound~\eqref{eq:bound_using_EN} when $ C := \kappa + \frac{c}{\kappa} \geq - \mu $.
	If $ C + \mu \leq 0 $, then we just remove the core of the interaction, so all in all we obtain
	\begin{equation*}
		\sum\limits_{i < j}^n w \myp{x_i - x_j} - \mu n
		\geq - L^d K \myp{C + \mu}_+^{\frac{\gamma}{\gamma-1}} + \myp{C+\mu}_- n.
	\end{equation*}
	Using this, it follows for any grand-canonical probability $ \bP = \myp{\bP_n} $ supported in $ C_L $, that
	\begin{align}
		\cG_T \myp{\bP}-\mu\cN(\bP)
		\geq{}& - L^d K \myp{C+ \mu}_+^{\frac{\gamma}{\gamma-1}} + \myp{C+ \mu}_- \mathcal{N} \myp{\bP} - T \mathcal{S} \myp{\bP} \nn \\
		\geq{}& - L^d K \myp{C+ \mu}_+^{\frac{\gamma}{\gamma-1}} - L^d T e^{- \frac{1}{T} \myp{C+\mu}_-},
	\label{eq:fgclowbound1}
	\end{align}
	where we used that the energy of a non-interacting system in a bounded region $ \Omega \subseteq \R^d $ is minimized by
	\begin{equation*}
		- \widetilde{\mu} \mathcal{N} \myp{\bP} - T \mathcal{S} \myp{\bP}
		\geq - T \log \myp[\bigg]{ \sum\limits_{n \geq 0} \frac{1}{n!} \int_{\Omega^n} e^{ \frac{\widetilde{\mu} n}{T}} \id \bP_n}
		= - T e^{\frac{\widetilde{\mu}}{T}} \abs{\Omega}.
	\end{equation*}
	The bound \eqref{eq:fgclowbound} then follows directly from \eqref{eq:fgclowbound1} by minimizing over $ \bP $.

	In the case $ \alpha = d $, the bound in \eqref{eq:bound_using_EN} is replaced by
	\begin{align*}
		\sum\limits_{i < j}^n w \myp{x_i - x_j} - \mu n
		\geq{}& \sum\limits_{k \in \mathcal{A}} \frac{c}{\kappa} n_k^2 \log n_k - \myp[\big]{\kappa + \mu} n_k.
	\end{align*}
    For $\lambda>0$, the function $x\in\R_+\mapsto x^2\log x-\lambda x$ is first decreasing and then increasing. It attains its minimum at $\bar x=e^{W-1/2}=\lambda/(2W)$ where $W=W_0(\lambda\sqrt{e}/2)\geq0$, with $W_0$ the principal branch of the Lambert $ W $-function, solving $ W_0 \myp{y} e^{W_0 \myp{y}} = y $ for $ y \geq 0 $. The minimum equals
    $$\bar x^2\log \bar x-\lambda \bar x=-\frac{\lambda^2}{8W^2}(2W+1).$$
    We have $\bar x>1$ if and only if $W>1/2$ and in this case we obtain
    $$\bar x^2\log \bar x-\lambda \bar x\geq -\frac{\lambda^2}{2W}.$$
    If $\bar x\leq 1$ then since the function is increasing on the right of $\bar x$, we conclude that $x^2\log x-\lambda  x\geq -\lambda$ for all $x\geq1$. Thus for integers we can conclude that
    $$\min_{n\in\N} \big\{n^2\log n-\lambda n\big\} \geq -\lambda -\frac{\lambda^2}{2W_0(\lambda\sqrt e/2)},\qquad \forall \lambda>0.$$
    Using that $W_0(y)\sim_0 y$ and $W_0(y)\sim_{+\ii}\log (y)$, we obtain
    $$\min_{n\in\N} \big\{n^2\log n-\lambda n\big\} \geq -C\lambda\left(1+\frac{\lambda}{\log(2+\lambda)}\right),\qquad \forall \lambda>0,$$
    for some universal constant $C$.
	Inserting this bound in our problem and arguing as in the $ \alpha \neq d $ case, we end up with a bound of the form
	\begin{equation*}
		g_T \myp{\mu,L}
		\geq{} - K \myp{\kappa + \mu}_+\left(1+\frac{\myp{\kappa + \mu}_+}{\log(2+\myp{\kappa + \mu}_+)}\right)- T e^{- \frac{1}{T} \myp{\kappa+\mu}_-}.
	\end{equation*}
	Changing $\kappa+\mu$ into $C+\mu$ for some $C>\kappa$, we obtain the simpler bound in the statement.

	Since $ g_T $ is concave in $ \mu $, it is also the Legendre transform of $ f_T $, that is, for $ \alpha \neq d $,
	\begin{equation*}
		g_T \myp{\mu,L}
		= \inf_{\rho \geq 0} \Set{f_T \myp{\rho,L} - \mu \rho}
		\leq{} C \rho^{\gamma} + C \myp{1+T} \rho + T \rho \log \rho - \mu \rho,
	\end{equation*}
	for any $ \rho \geq 0 $, where the inequality follows from \eqref{eq:fcanupbound}.
	When $ \alpha = d $, the same bound holds with $ \rho^{\gamma} $ replaced by $ \rho^2 \myp{\log \rho}_+ $.
	When $ \mu \leq \widetilde{C} := C \myp{1+T} + T + C $, we choose $ \rho_{\mu} = e^{-\frac{1}{T} \myp{\mu - \widetilde{C}}_- } \leq 1 $ to obtain
	\begin{equation*}
		g_T \myp{\mu,L}
		\leq C \rho_{\mu}^{\gamma} - \myp{C+T}\rho_{\mu}
		\leq -T \rho_{\mu}
		= - T e^{-\frac{1}{T} \myp{\mu - \widetilde{C}}_- }.
	\end{equation*}
	The same bound holds when $ \alpha = d $ because $ \rho_{\mu} \leq 1 $ implies $ \rho_{\mu}^2 \myp{ \log \rho_{\mu}}_+ = 0 $.
	In the case when $ \mu \geq \widetilde{C} $ and $ \alpha \neq d $, we can for instance take $ \rho_{\mu} $ to satisfy
	\begin{equation*}
		\rho_{\mu}^{\gamma-1} = \frac{1}{\gamma C + T} \myp{\mu - \widetilde{C}} + 1,
	\end{equation*}
	or equivalently, $ \mu = \myp{\gamma C + T} \rho_{\mu}^{\gamma-1} + C \myp{1+T} - \myp{\gamma -1}C $.
	Using that $ \rho_{\mu} \geq 1 $, this leads to the bound
	\begin{align*}
		g_T \myp{\mu,L}
		\leq{}& - \myp{\myp{\gamma-1} C + T} \rho_{\mu}^{\gamma} + \myp{\gamma - 1}C \rho_{\mu} + T \rho_{\mu} \log \rho_{\mu} \\
		\leq{}& - \myp{\gamma-1}C \rho_{\mu} \myp{\rho_{\mu}^{\gamma-1} - 1} - T \rho_{\mu} \\
		={}& - \frac{\myp{\gamma-1}C}{\gamma C + T} \rho_{\mu} \myp{\mu-\widetilde{C}} - T \rho_{\mu}
		\leq{} - \frac{\myp{\gamma-1}C}{\gamma C + T} \myp{\mu - \widetilde{C}}^{\frac{\gamma}{\gamma-1}} - T.
	\end{align*}
	The argument is similar when $\alpha=d$, using the Lambert $W$-function.
	For instance, taking $\rho_{\mu} = \exp\myp[\big]{W_0 \myp{\frac{\sqrt{e}}{2 \myp{C+T}} \myp{\mu - C \myp{1+T} - T} }} \geq 1 $ immediately gives
	\begin{align*}
		g_T \myp{\mu,L}
		\leq{}& \myp{C+T} \rho_{\mu}^2 \myp{\log \rho_{\mu}}_+ - \myp{\mu - C \myp{1+T} - T} \rho_{\mu} - T \\
		={}& - \frac{2 \sqrt{e} - e}{4 \myp{C+T}} \frac{\myp{\mu - C \myp{1+T} - T}^2}{W_0 \myp{\frac{\sqrt{e}}{2 \myp{C+T}} \myp{\mu - C \myp{1+T} - T}}} - T,
	\end{align*}
	which can then easily be put in the same form as in the statement.
\end{proof}

\appendix
\section{Ruelle local bounds for superstable interactions}
\label{app:Ruelle}

 \subsection{Ruelle bounds for homogeneous systems}

Ruelle bounds~\cite{Ruelle-70} are local estimates on Gibbs point processes, which are uniform with respect to the size of the system. They have played a very important role in statistical mechanics and Ruelle's original article~\cite{Ruelle-70} is considered as a historical breakthrough. Unfortunately, the dependence in terms of the temperature $T$ and the chemical potential $\mu$ (resp. density $\rho$) was not made explicit in Ruelle's paper. It is the goal of this appendix to state such quantitative bounds. The proof is essentially that of Ruelle in~\cite{Ruelle-70} with some small variations and we provide it for completeness.

\begin{theorem}[Quantitative Ruelle bounds]\label{thm:Ruelle}
Assume that $w=w_1+w_2$ where

\smallskip

\noindent $\bullet$ $w_1\geq0$ with $w_1(x)\geq A>0$ for $|x|\leq\delta$;

\smallskip

\noindent $\bullet$ $w_2$ is stable,
$$\sum_{1\leq j<k\leq n}w_2(x_j-x_k)\geq -\kappa n,\qquad \forall x_1,...,x_n\subset\R^d,\ \forall n\geq2$$
and lower regular
\begin{equation}
 w_2(x)\geq -\phi(|x|),
 \label{eq:assumption_w_2}
\end{equation}
where $x\mapsto \phi(|x|)$ is a non-negative radial-decreasing function in $L^1\cap L^\ii(\R^d)$.

For $x=\{x_1,...,x_n\}\subset\R^d$, let
$$H(x):=\sum_{1\leq j<k\leq n}w(x_j-x_k),\qquad h(x):=\sum_{1\leq j<k\leq n}w_1(x_j-x_k)$$
be the total and repulsive part of the interaction. Denote by $h_Q(x)=h(x\cap Q)$ its restriction to any set $Q\subset\R^d$.
There exists a function $\psi:\N\mapsto \R_+$ tending to $+\ii$ and a constant $\zeta$ (both depending only on $w$) such that
\begin{equation}
\expec[\big]{e^{\frac\beta2 h_{Q}}}_{\beta,\mu,\Omega} 
\leq \exp\myp[\big]{L^de^{\beta(\mu+\zeta)}} \myp[\bigg]{ e^{\zeta\beta\psi(L)L^d}+\sum_{\substack{\ell\in\N\\ \ell> L}}e^{-\beta\big(\frac{\psi(\ell)}{\zeta}-\zeta e^{\beta\mu}\big)\ell^d} }
\label{eq:main_estim_Ruelle_thm}
\end{equation}
for any $\beta>0$, any $\mu\in\R$, any bounded domain $\Omega$ and any cube $Q$ of side length $L\geq \zeta$. Here $\pscal{\nonarg}_{\beta,\mu,\Omega}$ denotes the expectation against the grand-canonical Gibbs point process with energy $H$, chemical potential $\mu$ and inverse temperature $\beta$, in a domain $\Omega\subset\R^d$.
\end{theorem}

The estimate~\eqref{eq:main_estim_Ruelle_thm} is not at all expected to be optimal and it is only displayed for concreteness. In the proof we explain in detail how to choose $\psi$ and the large constant $\zeta$ (the latter is the maximum of several explicit  constants). The precise final estimate we obtain is in~\eqref{eq:estim_exp} below. Note that~\eqref{eq:main_estim_Ruelle_thm} behaves rather badly when the temperature tends to 0, that is, $\beta\to\ii$. In fact, the proof also works at $T=0$ and the corresponding bound is provided in Corollary~\ref{cor:T0} below.

The function $\psi$ depends on the decay of $\phi$ at infinity and, in principle, it can diverge very slowly. If
$\phi(x)=B(1+|x|)^{-s}$ for some $s>d$, we can just take $\psi(\ell)= \ell^\eps$ for any $0<\eps<\min(1,s-d)$ or $\psi(\ell)=\log\ell$. The series on the right of~\eqref{eq:main_estim_Ruelle_thm} is convergent because $\psi(\ell)\to+\ii$, hence the power in the exponential is strictly negative for $\ell$ large enough.\footnote{In our proof we do not get the full sum over $\ell> L$ on the right side of~\eqref{eq:main_estim_Ruelle_thm}, but only over a subsequence $\{\ell_j\}$ of integers tending exponentially fast to infinity. The details are provided below.} Note, however, that the value of the sum can be quite large when the activity $z=e^{\beta\mu}$ is large. We do not expect this to be optimal at all. Better estimates exist at low activity $z=e^{\beta\mu}\ll1$, using expansion methods~\cite{Ruelle}.

Since $\psi$ diverges, we may choose $L$ large enough (depending on $\beta$ and $\mu$) so that the term in the parenthesis in~\eqref{eq:main_estim_Ruelle_thm} is $\leq e^{2\zeta\beta\psi(L)L^d}$. We then end up with the simpler estimate
\begin{equation}
\expec[\big]{e^{\frac\beta2 h_{Q}}}_{\beta,\mu,\Omega} 
\leq \exp\myp[\big]{ L^d(e^{\beta(\mu+\zeta)}+2\zeta\beta\psi(L)) }.
\label{eq:main_estim_Ruelle_thm_simplified}
\end{equation}
As we will see in~\eqref{eq:estim_h_cube} below, there exists a constant $c>0$ (depending on $\delta$) such that $h_{Q}\geq (A/2)(cL^{-d}n_Q^2-n_Q)$ where $n_Q$ is the number of points in $Q$ and we obtain
\begin{equation}
\expec[\big]{n_Q^{2k}}_{\beta,\mu,\Omega}\leq C_kL^{2dk}(1+\psi(L))
 \label{eq:estim_moments}
\end{equation}
on the local moments of the Gibbs state, with $C_k$ depending on $w,\beta,\mu$. Due to the divergence of $\psi(L)$, this bound is definitely not optimal for large values of $L$. But one should not think of taking $L$ too large here. Once we obtain~\eqref{eq:estim_moments} for one cube of side length $L_0$, it immediately follows for a bigger cube by covering it with the smaller cubes. This way we obtain the expected estimate involving the volume of the large cube. The estimate~\eqref{eq:estim_moments} implies that the correlation functions are bounded in $L^1_{\rm unif}(\R^d)$, independently of the size of $\Omega$. In~\cite{Ruelle-70}, Ruelle proves $L^\ii(\R^d)$ bounds on the correlation functions, which are discussed later in Remark~\ref{rmk:correlation_bounds}. Averaged bounds of the kind of~\eqref{eq:main_estim_Ruelle_thm} and~\eqref{eq:estim_moments} are usually enough in applications.

In the case that $\phi$ decays polynomially at infinity, we can make the bound more explicit after choosing $\psi(\ell)=\ell^\eps$.

\begin{corollary}[Ruelle bounds for polynomially decreasing potentials]\label{cor:Ruelle_polynomial}
Assume that $w$ satisfies the assumptions of Theorem~\ref{thm:Ruelle} with $\phi(|x|)=B(1+|x|)^{-s}$ for some $s>d$. Let $0<\eps<\min(1,s-d)$. Then we have
\begin{equation}
\expec[\big]{e^{\frac\beta2 h_{Q}}}_{\beta,\mu,\Omega}
\leq e^{\zeta\beta\cR},
\qquad \pscal{n_{Q}^2}_{\beta,\mu,\Omega}
\leq L^d\,\min\big\{ze^{\zeta\beta(2\cR+1)}\,;\,\zeta\cR\big\}
\label{eq:main_estim_Ruelle_polynomial}
\end{equation}
for any cube $Q$ of side length $L\geq \zeta$, where $z=e^{\beta\mu}$ and
$$\cR:=\frac{zL^de^{\beta\zeta}}{\beta\zeta}+L^{d+\eps}+ \myp[\Big]{ \frac{z}\beta }^{1+\frac{d}\eps}+\frac{(\log\beta)_-+1}\beta.$$
\end{corollary}

Applying the above bound for $L_0=\zeta$ and covering any larger cube by cubes of side length $L_0$, we conclude that
\begin{equation}
\boxed{\pscal{n_{Q}^2}_{\beta,\mu,\Omega}\leq C_{\beta,\eps}|Q|^2\,z\big(1+z^{\frac{d}\eps}\big)}
 \label{eq:ruelle_app}
\end{equation}
for a constant $C_{\beta,\eps}$ depending on $w$ and $\beta>0$ and for every large enough cube $Q$. The proof of Corollary~\ref{cor:Ruelle_polynomial} is provided in Section~\ref{sec:proof_cor_Ruelle}.

\subsection{Proof of Theorem~\ref{thm:Ruelle}}
We write the whole proof assuming for simplicity $\delta=\sqrt{d}$. The general case follows by scaling, that is, after applying the inequality to $w'(x):=w(\delta x/\sqrt{d})$ which has $\delta'=\sqrt{d}$, taking  $\Omega'=(\sqrt{d}/\delta)\Omega$ and $e^{\beta\mu'}=(\delta/\sqrt{d}) e^{\beta\mu}$. This changes the constant $\zeta$. The function $\psi$ also needs to be rescaled appropriately, since it depends on $\phi$.

\medskip

\noindent\textbf{Step 1: Definition of a splitting of space.}
First we split our space into the small cubes $C_i:=i+\left[0,1\right)^d$ with $i\in\Z^d$ and denote by
$n_i(x):=\#C_i\cap x$
the number of points in the small cube $C_i$. Those have diameter $\delta=\sqrt{d}$ and therefore we have $w_1(x-y)\geq A$ for $x,y$ in such a cube, hence
\begin{equation}
h_{C_i}\geq \frac{An_i(n_i-1)}{2}.
\label{eq:estim_h_cube}
\end{equation}
This estimate is going to play an important role in the following.

Since the bound in~\eqref{eq:main_estim_Ruelle_thm} is invariant under translations, we can assume that the cube $Q$ of interest is centered at the origin of space. We will also assume for simplicity that the side length of $Q$ is an even integer, so that $Q$ can be written as the union of the smaller cubes $C_i$.

The proof will require to estimate the interaction of some points $x_j\in Q$ with an arbitrary number of points $y_j\in\R^d\setminus Q$. For this we will need to understand how the $y_j$ are distributed in space and, in particular, what is the local number of points in any domain. Big clusters will typically generate a large interaction with the $x_j$.
To this end, we introduce a growing sequence of cubes $\{Q_j\}_{j\geq0}$ with $Q_0=Q$ (the cube in the statement of the theorem). The side length will increase exponentially fast, but not too fast. More precisely, we choose an increasing sequence of integers $\ell_0<\ell_1<\cdots$ which satisfies $\ell_j\to\ii$ as well as $\ell_{j+1}/\ell_j\approx 1+2\alpha$, where $\alpha>0$. We will later need to assume that $\ell_0$ is large enough and $\alpha$ is small enough, depending on $w$. In fact we will not need to have an exact limit and only require that
\begin{equation}
\boxed{1+\alpha\leq \frac{\ell_{j+1}}{\ell_j}\leq 1+3\alpha,\qquad \forall j\geq0.}
 \label{eq:bounds_ell_j}
\end{equation}
To be more concrete, we can for instance take $\ell_j=\lfloor\ell_0(1+2\alpha)^j\rfloor$
and choose $\ell_0\in\N$ large enough to have~\eqref{eq:bounds_ell_j}.
We then call
$$Q_j:=(-\ell_j,\ell_j)^d,\qquad V_j:=(2\ell_j)^d$$
the cube which is the union of $(2\ell_j)^d$ smaller cubes $C_i$ and has volume $V_j$.
We will focus our attention on the annulus-type regions $A_j:=Q_j\setminus Q_{j-1}$ which has the volume
$$|A_j|=(2\ell_j)^d-(2\ell_{j-1})^d=(2\ell_j)^d \myp[\bigg]{1-\frac{\ell_{j-1}^d}{\ell_{j}^d}},$$
hence
\begin{equation}
 \frac{\alpha d V_j}{2}\leq V_j\left(1-\frac{1}{(1+\alpha)^d}\right)\leq |A_j|\leq V_j\left(1-\frac{1}{(1+3\alpha)^d}\right)\leq 3\alpha d V_j
 \label{eq:volume_A_j}
\end{equation}
for $\alpha$ small enough (depending only on the dimension $d$).

Our goal will be to look at how many particles there are in each $A_j$ or $Q_j$, and whether this number is of the order of the volume or not. The fact that the interaction is integrable permits a slight deviation from a volume term, which will be described by a function $\psi$ depending on the radial function $\phi$ (the lower bound to the stable part $w_2$ of the interaction), and is the function in the statement (after an appropriate rescaling). We take $\ell\in\N\mapsto \psi(\ell)\in[1,\ii)$ any increasing function so that
\begin{equation}
\boxed{\frac{\psi(\ell+1)}{\psi(\ell)}\leq \frac{\ell+1}{\ell},\quad \text{for all $\ell\geq1$, and}\qquad \sum_{\ell\geq1}\psi(\ell)\phi(\ell)\ell^{d-1}<\ii,}
\label{eq:assump_psi}
\end{equation}
that is, $\psi$ must diverge to infinity sufficiently slowly. Recall that since $ x \mapsto \phi (\lvert x\rvert) $ is in $ L^1 (\R^d) $ and is radial decreasing, we have $\sum_{\ell\geq1}\phi(\ell)\ell^{d-1}<\ii$.
A typical example is when $\phi(|x|)=(1+|x|)^{-s}$ with $s>d$, in which case we just take $\psi(\ell)=\ell^\eps$ for $\eps<\min(1,s-d)$. Denote finally
$\psi_j:=\psi(\ell_j).$

\medskip

\noindent\textbf{Step 2: Pointwise lower bound on $H$.}
Next we explain the main Ruelle estimate, which is a pointwise bound on the total energy $H$. We consider an arbitrary  configuration of points in $\R^d$, which we split into the points $x_j$ in $Q_0$ and those outside of $Q_0$. The idea of Ruelle in~\cite{Ruelle-70} is to distinguish between `good' configurations where the points in $\R^d\setminus Q_0$ are well distributed in space, with a reasonable number of particles in any bounded region, from the `bad' configurations where too many of these points concentrate in a small volume, which results in a stronger interaction with the $x_j$'s. The method is then to merge the bad points with the $x_j$, use the superstability for the union to prove that their energy is very positive, and use this large energy to control the interaction with the remaining good points. The latter interaction is essentially of the order of the volume where the bad points live. Thus we need the energy of the bad points to be much larger than the volume they occupy, in order to control this interaction.

Let now $Y\subset \R^d\setminus Q_0$ be a finite set. The condition which characterizes good and bad configurations of $Y$ is whether $\sum_{i\in \Z^d\cap Q_j}n_i(Y)^2$ is of the order of $V_j\psi_j$ or not. In other words, we use the function $\psi_j$ to quantify large but still controllable deviations of the number of particles (squared locally) compared to the volume of the large cube $Q_j$. To be more precise, we call $q=q_Y\geq0$ the \emph{smallest integer} such that
\begin{equation}
 \sum_{C_i\subset Q_j}n_i(Y)^2\leq  V_j\psi_j\qquad \forall j> q_Y.
 \label{eq:def_q_y}
\end{equation}
Since we have
$$\sum_{C_i\subset Q_j}n_i(Y)^2\leq \myp[\bigg]{ \sum_{C_i\subset Q_j}n_i(Y) }^2\leq (\#Y)^2$$
the inequality~\eqref{eq:def_q_y} always holds for $j$ large enough, hence $q=q_Y<\ii$. When $q=0$, the inequality~\eqref{eq:def_q_y} holds for all $j\geq0$ (it is valid for $j=0$ since the left side vanishes in this case). When $q>0$, we have by definition
\begin{equation}
 \sum_{C_i\subset Q_{q}}n_i(Y)^2>  V_q\psi_q\qquad \text{if $q>0$.}
 \label{eq:q_inequality_critical}
\end{equation}
Now we are going to split our configuration space in the two regions $Q_{q}$ and $\R^d\setminus Q_{q}$. We write $Y=y\cup z$ where $y$ are the particles in $Q_{q}$ (if $q>0$) and $z$ those in $\R^d\setminus Q_q$. We write our full energy in the form
$$H(x\cup y\cup z)=H(x\cup y)+H(z)+W(x\cup y,z)$$
where the third term is the interaction:
$$W(x,y):=\sum_{i=1}^{\#x}\sum_{k=1}^{\#y}w(x_i-y_k).$$
Our goal is to find a lower bound on $H(x\cup y\cup z)$ which only involves $x$, $H(z)$ and $q$. To state the bound, we recall that $h(x)$ is the repulsive part of the energy.

\begin{proposition}[Pointwise lower bound on $H$]\label{prop:pointwise_H}
One can choose $\alpha$ small enough and $\ell_0$ large enough (depending only on $d$ and $w$) so that for all $x\subset Q_0$ and all $Y\subset \R^d\setminus Q_0$, we have the pointwise estimate
\begin{equation}
 H(x\cup y\cup z)\geq \frac{h(x)}2 +\frac{A}{16}\psi_q V_q+\frac{A}{16V_q}(\#y)^2+
H(z)-\left(\kappa+\frac{A}2\right)\#x
 \label{eq:pointwise_bd_H}
\end{equation}
whenever $q=q_Y>0$, with the decomposition $Y=y\cup z$, where $y=Y\cap Q_q$ and $z=Y\cap (\R^d\setminus Q_q)$. If $q=q_Y=0$, we have
\begin{equation}
H(x\cup Y)\geq \frac{h(x)}2 -\frac{4^{d-2}A}{\alpha}\psi_0V_0+H(Y)-\left(\kappa+\frac{A}2\right)\#x.
 \label{eq:pointwise_bd_H_q0}
\end{equation}
\end{proposition}

Recall that $A$ is the minimum of $w_1(x)$ for $|x|\leq\delta=\sqrt{d}$ and that $\kappa$ is the stability constant for $w_2$.

\begin{proof}[Proof of Proposition~\ref{prop:pointwise_H}]
We call  $n=\#x$, $k=\#y$ and $m=\#z$ the number of particles in each group. We assume $q=q_Y>0$ in the whole proof and treat the simpler case $q=0$ at the end. First we use the stability of $w_2$ and the positivity of $w_1$ to get
\begin{align}
H(x\cup y\cup z)&=H(x\cup y)+H(z)+W(x\cup y,z)\nn\\
&\geq h(x)+h(y)-\kappa(n+k)+H(z)+W(x\cup y,z)\nn\\
&\geq \frac{h(x)}2 +\frac{A}{4}\sum_i n_i(x)^2+\frac{A}{2}\sum_i n_i(y)^2\nn\\
&\qquad +H(z)+W(x\cup y,z)-C(n+k),\label{eq:decompos_Hxyz}
\end{align}
where $C=\kappa+A/2$.
Our first step will be to control the easy term $-Ck$. We use the Cauchy-Schwarz (or Jensen) inequality which says that
\begin{equation}
\sum_i n_i(y)^2\geq V_q^{-1} \myp[\bigg]{ \sum_i n_i(y) }^2=V_q^{-1}k^2.
 \label{eq:estim_square_y}
\end{equation}
Since
$\frac{A}{16V_q}k^2-Ck\geq -\frac{4C^2}{A}V_q$
we obtain
$$\frac{A}{2}\sum_i n_i(y)^2-Ck\geq \frac{3A}{8}\sum_i n_i(y)^2+\left(\frac{A\psi_q}{16}-\frac{4C^2}{A}\right)V_q,$$
where we have used that $\sum_i n_i(y)^2\geq V_q\psi_q$ due to the definition of $q=q_Y$.
Using that $\psi$ tends to infinity, we can assume that $\ell_0$ is large enough so that
\begin{equation}
\boxed{\frac{A}{16}\psi(\ell_0)\geq \frac{4C^2}{A}=\frac{(2\kappa+A)^2}{A}.}
\label{eq:condition_ell_0_1}
\end{equation}
Then we arrive at
\begin{multline}
 H(x\cup y\cup z)\geq \frac{h(x)}2 +\frac{A}{4}\sum_i n_i(x)^2+\frac{3A}{8}\sum_i n_i(y)^2+H(z)\\
 -Cn+W(x\cup y,z).
\label{eq:decompos_Hxyz2}
\end{multline}

Our next and main task is to bound the interaction $W(x\cup y,z)$. To simplify our writing, we call
$$\delta_{ij}=\min_{\substack{x\in C_i\\ y\in C_j}}|x-y|, \qquad \phi_{ij}:=\max_{\substack{x\in C_i\\ y\in C_j}}\phi(|x-y|)=\phi(\delta_{ij})$$
the distance between the cubes $C_i$ and $C_j$ and the corresponding interaction. Note that when the two cubes have a common boundary, we just get $\delta_{ij}=0$ and $\phi_{ij}=\phi(0)$. Splitting our space into the small cubes $C_i$ and using the lower regularity of $w_2$, we obtain
\begin{align}
 W(x\cup y,z)&\geq -\sum_{i,j}n_i(x)n_j(z)\phi_{ij}-\sum_{i,j}n_i(y)n_j(z)\phi_{ij}\nn\\
 &\geq-\tau \sum_{i,j} n_i(x)^2\phi_{ij}-\tau \sum_{i,j}n_i(y)^2\phi_{ij}-\frac1{2\tau} \sum_{\substack{i,j\\ C_i\subset Q_q}}n_j(z)^2\phi_{ij}.
 \label{eq:interaction_x_y_z}
\end{align}
In the first two sums we have for simplicity dropped the condition that $C_j\subset \R^d\setminus Q_q$.
The integrability of $\phi$ implies that
$$\boxed{S:=\sum_{i\neq0}\phi_{0i}<\ii.}$$
In view of~\eqref{eq:decompos_Hxyz2} we take $\tau :=A/(8S)$
so that the first two error terms are controlled. We thus obtain
\begin{equation}
 H(x\cup y\cup z)\geq \frac{h(x)}2 +H(z) -Cn +\frac{A}{4}\sum_{i}n_i(y)^2-\frac{4S}{A} \sum_{\substack{i,j\\ C_i\subset Q_q}}n_j(z)^2\phi_{ij},\label{eq:decompos_Hxyz3}
\end{equation}
after dropping the positive term $(A/8)\sum_i n_i(x)^2$.
It remains to estimate the last sum by $\psi_q V_q$. We distinguish between the $C_j$ belonging to the first shell $A_{q+1}=Q_{q+1}\setminus Q_q$ and the ones further away.

\medskip

\noindent\textbf{First shell $A_{q+1}$.} After adding all the missing terms for $i$, the sum over all the $C_j\subset A_{q+1}$ can be estimated by
$$\sum_{\substack{C_i\subset Q_q\\ C_j\subset A_{q+1}}}n_j(z)^2\phi_{ij}\leq S\sum_{C_j\subset A_{q+1}}n_j(z)^2.$$
By definition of $q$ in~\eqref{eq:def_q_y} and~\eqref{eq:q_inequality_critical}, we have
$$\sum_{C_j\subset A_{q+1}}n_j(z)^2+\sum_{C_j\subset Q_{q}}n_j(y)^2\leq V_{q+1}\psi_{q+1},\qquad
\sum_{C_j\subset Q_{q}}n_j(y)^2\geq V_{q}\psi_{q}
$$
since $q>0$ and therefore we obtain
$$\sum_{C_j\subset A_{q+1}}n_j(z)^2\leq V_{q+1}\psi_{q+1}-V_q\psi_q.$$
Here comes the importance of the properties of $\ell_j$ and $\psi$. By~\eqref{eq:bounds_ell_j} we have
$$V_{q+1}=(2\ell_{q+1})^d\leq (1+3\alpha)^d(2\ell_q)^d=(1+3\alpha)^dV_q.$$
Similarly, the estimate $\psi(\ell+1)\leq (\ell+1)\psi(\ell)/\ell$ in~\eqref{eq:assump_psi} implies by induction
\begin{equation}
\psi_{j+1}=\psi(\ell_{j+1})\leq \frac{\ell_{j+1}}{\ell_{j}}\psi(\ell_{j})\leq (1+3\alpha)\psi(\ell_{j})=(1+3\alpha)\psi_j.
\label{eq:induction_psi}
\end{equation}
Thus we obtain
$$\sum_{C_j\subset A_{q+1}}n_j(z)^2\leq \myp[\big]{ (1+3\alpha)^{d+1}-1 } V_q\psi_q.$$
This term can be absorbed into the positive term in~\eqref{eq:decompos_Hxyz3} provided we choose $\alpha$ so that
\begin{equation}
\boxed{\frac{4S^2}{A} \myp[\big]{(1+3\alpha)^{d+1}-1 } \leq \frac{A}{16}.}
\label{eq:condition_alpha}
\end{equation}
Then we obtain
\begin{equation}
 H(x\cup y\cup z)\geq \frac{h(x)}2 +H(z) -Cn+\frac{3A}{16}\sum_{i}n_i(y)^2-\frac{4S}{A}\!\! \sum_{\substack{C_i\subset Q_q\\ C_j\subset \R^d\setminus Q_{q+1}}}n_j(z)^2\phi_{ij}.\label{eq:decompos_Hxyz4}
\end{equation}

\medskip

\noindent\textbf{Other shells $A_{j}$ for $j\geq q+2$.} Any shell $A_j$ with $j\geq q+2$ is at least at a distance $\ell_{j-1}-\ell_q$ from $Q_q$. We thus group our small cubes according to these shells and estimate the interaction by $\phi(\ell_{j-1}-\ell_q)$. This gives
$$\sum_{\substack{C_i\subset Q_q\\ C_j\subset \R^d\setminus Q_{q+1}}}n_j(z)^2\phi_{ij}\leq V_q\sum_{j\geq q+2}\sum_{C_i\subset A_j}n_i(z)^2\phi(\ell_{j-1}-\ell_q)$$
since there are $V_q$ cubes $C_i$ in $Q_q$. We would like to use that
$\sum_{C_i\subset Q_j}n_i(z)^2\leq V_j\psi_j$
by definition of $q$ (we remove here the part involving $y$) and hence we rewrite the sum as
\begin{align*}
\MoveEqLeft[4] \sum_{j\geq q+2}\phi(\ell_{j-1}-\ell_q)\sum_{C_i\subset  A_j}n_i(z)^2 \\
={}& \sum_{j\geq q+2}\phi(\ell_{j-1}-\ell_q) \myp[\bigg]{ \sum_{C_i\subset Q_j}n_i(z)^2-\sum_{C_i\subset Q_{j-1}}n_i(z)^2 } \\
={}& \sum_{j\geq q+2} \myp[\bigg]{ \sum_{C_i\subset Q_j}n_i(z)^2 } \big(\phi(\ell_{j-1}-\ell_q)-\phi(\ell_{j}-\ell_q)\big)\\
&-\phi(\ell_{q+1}-\ell_q)\sum_{C_i\subset Q_{q+1}}n_i(z)^2\\
\leq{}& \sum_{j\geq q+2}V_j\psi_j\big(\phi(\ell_{j-1}-\ell_q)-\phi(\ell_{j}-\ell_q)\big)\\
={}& 2^d \sum_{j\geq q+2}\ell_j^d\psi(\ell_j)\big(\phi(\ell_{j-1}-\ell_q)-\phi(\ell_{j}-\ell_q)\big).
\end{align*}
This is now independent of $z$, as required. It will be useful to replace $\ell_j$ by $\ell_{j-1}-\ell_q$ in the factor $ \ell_j^d\psi(\ell_j) $. For the volume we can simply use that
$$\frac{(\ell_{j-1}-\ell_q)^{d}}{\ell_j^{d}}=\left(\frac{\ell_{j-1}}{\ell_{j}}\right)^{d}\left(1-\frac{\ell_q}{\ell_{j-1}}\right)^{d}\geq \frac1{(1+3\alpha)^d}\left(1-\frac{1}{1+3\alpha}\right)^{d}\geq \alpha^{d}$$
for $(1+3\alpha)^2\leq 3$. Similarly, using~\eqref{eq:induction_psi} we have
$$\psi(\ell_j)\leq \frac{\ell_j}{\ell_{j-1}-\ell_q}\psi(\ell_{j-1}-\ell_q)\leq \frac{\psi(\ell_{j-1}-\ell_q)}{\alpha}.$$
This gives
\begin{multline*}
\sum_{j\geq q+2}\phi(\ell_{j-1}-\ell_q)\sum_{C_i\subset  A_j}n_i(z)^2\\
\leq \frac{2^d}{\alpha^{d+1}} \sum_{j\geq q+2}(\ell_{j-1}-\ell_q)^d\psi(\ell_{j-1}-\ell_q)\big(\phi(\ell_{j-1}-\ell_q)-\phi(\ell_{j}-\ell_q)\big).
\end{multline*}
We now show that the sum on the right side is bounded above by
\begin{equation}
I(\ell_0):=\sum_{\ell\geq\ell_0}\ell^d\psi(\ell)\big(\phi(\ell)-\phi(\ell+1)\big).
\label{eq:def_I_ell_0}
\end{equation}
We use the fact that for $f$ a non-decreasing function, $g$ a non-increasing function and $k_j$ an increasing sequence of integers, we have
\begin{align*}
\MoveEqLeft[2] f(k_{j-1})\big(g(k_{j-1})-g(k_{j})\big)\\
={}& f(k_{j-1})\big(g(k_{j-1})-g(k_{j-1}+1)\big)+\cdots+f(k_{j-1})\big(g(k_{j}-1)-g(k_{j})\big)\\
\leq{}& f(k_{j-1})\big(g(k_{j-1})-g(k_{j-1}+1)\big)+\cdots+f(k_{j}-1)\big(g(k_{j}-1)-g(k_{j})\big).
\end{align*}
In the estimate we used that $f$ is non-decreasing and that $g(k_{j-1}+k)-g(k_{j-1}+k+1)\geq0$ since $g$ is non-increasing. This way we obtain all the $f(\ell)(g(\ell)-g(\ell+1))$ for $\ell$ between $k_{j-1}$ and $k_j-1$. After summing over $j$ we deduce that
$$\sum_j f(k_{j-1})\big(g(k_{j-1})-g(k_{j})\big)\leq \sum_\ell f(\ell)\big(g(\ell)-g(\ell+1)\big).$$
Applying this to $k_j=\ell_j-\ell_q$ and using that
$\ell_{j-1}-\ell_q\geq \ell_{q+1}-\ell_q\geq (1+\alpha)\ell_0\geq\ell_0,$
this proves as we wanted that
\begin{equation}
\sum_{j\geq q+2}\phi(\ell_{j-1}-\ell_q)\sum_{C_i\subset  A_j}n_i(z)^2\leq \frac{2^dI(\ell_0)}{\alpha^{d+1}}.
\end{equation}
Note that the series in~\eqref{eq:def_I_ell_0} is convergent. In fact, after changing indices this is the same as
$$\sum_{\ell\geq \ell_0} \myp[\big]{ (\ell+1)^d\psi(\ell+1)-\ell^d\psi(\ell) } \phi(\ell+1)<\ii,$$
which follows from the fact that $\psi(\ell+1)\leq(1+\ell^{-1})\psi(\ell)$ and $(1+\ell)^d\leq\ell^d(1+c/\ell)$, so that
$(\ell+1)^d\psi(\ell+1)-\ell^d\psi(\ell)\leq c\ell^{d-1}\psi(\ell).$
Thus $I(\ell_0)$ is the remainder of a convergent series due to our assumption~\eqref{eq:assump_psi} on $\psi$. As a conclusion, $\alpha>0$ being fixed so that~\eqref{eq:condition_alpha} holds, we need to take $\ell_0$ large enough so that
\begin{equation}
 \boxed{\frac{2^{d+2}S\, I(\ell_0)}{\alpha^{d+1}A}\leq \frac{A}{16}\psi(\ell_0).}
 \label{eq:condition_ell_0}
\end{equation}
As a conclusion, we obtain the bound
\begin{equation}
 H(x\cup y\cup z)\geq \frac{h(x)}2 +H(z) -Cn+\frac{A}{8}\sum_{i}n_i(y)^2,
 \label{eq:decompos_H_final}
\end{equation}
which reduces to the stated inequality~\eqref{eq:pointwise_bd_H} after using again~\eqref{eq:estim_square_y} and the definition of $q$.

It remains to deal with the case $q=0$. Then the energy is reduced to $H(x\cup z)$ with $z=Y$, that is, $y$ is empty. Of course, if $n=0$, then $H(x\cup z)=H(z)$  and we can therefore assume $n\geq1$. The exact same bounds as in~\eqref{eq:decompos_Hxyz3} (with $\tau=A/(4S)$ in~\eqref{eq:interaction_x_y_z}) provide
$$H(x\cup z)\geq \frac{h(x)}{2}-Cn+H(z) -\frac{2S}{A}\sum_{\substack{i,j\\ C_i\subset Q_0}}n_j(z)^2\phi_{ij}.$$
For the first shell we just have
$$\sum_{\substack{C_i\subset Q_0\\ C_j\subset A_{1}}}n_j(z)^2\phi_{ij}\leq S\psi_1V_1\leq S(1+3\alpha)^{d+1}V_0\psi_0,$$
whereas the next shells are estimated as above, leading to
$$H(x\cup z)\geq \frac{h(x)}{2}-Cn+H(z) -\frac{A}{32}\psi_0V_0\left(1+\frac{(1+3\alpha)^{d+1}}{(1+3\alpha)^{d+1}-1}\right).$$
Using the simple estimate
$$1+\frac{(1+3\alpha)^{d+1}}{(1+3\alpha)^{d+1}-1}\leq1+\frac{4^{d+1}}{6\alpha}\leq\frac{4^d}{\alpha}$$
since $\alpha<1$, we obtain~\eqref{eq:pointwise_bd_H_q0}. This concludes the proof of Proposition~\ref{prop:pointwise_H}.
\end{proof}

As a side remark, we notice that our conditions~\eqref{eq:condition_ell_0_1},~\eqref{eq:condition_alpha} and~\eqref{eq:condition_ell_0} are monotone in $\ell_0$. When they are valid for one $\ell_0$, then they also hold for larger values. A simple estimate at zero temperature follows immediately from Proposition~\ref{prop:pointwise_H}.

\begin{corollary}[Ruelle bound at $T=0$]\label{cor:T0}
Let $\Omega\subset\R^d$ be any domain and let $\mu\in\R$. Let $X\subset\Omega$ be any minimizer for the free energy
$$X\subset\Omega\mapsto H_{\Omega}(X)-\mu\, n_\Omega(X).$$
Let $L_\mu\geq0$ be the smallest integer such that $\psi(L_\mu/2)\geq 64\mu_+^2/A^2$. Then
\begin{align}
n_Q(X)&\leq \frac{4(2\kappa+A+\mu)_+}{A}\left(1+2^{d-2}\alpha^{-\frac12}\sqrt{\psi(L/2)}+(L_\mu/L)^{\frac d2}\right)L^d,\nn\\
h_Q(X)&\leq \frac{12(2\kappa+A+\mu)^2_+}{A}\left(1+2^{2d-1}\alpha^{-1}\psi(L/2)+(L_\mu/L)^{d}\right)L^d,
\label{eq:estim_T0}
\end{align}
for any cube $Q$ of side length $L\geq\zeta$.
\end{corollary}

For instance, for $\psi(\ell)=\ell^\eps$, we have $L_\mu\approx c\mu^{2/\eps}$. The local bounds from Corollary~\ref{cor:T0} can be used to pass to the limit and prove the existence of infinite ground states, as defined in~\cite{Radin-84,Radin-04,BelRadShl-10,Suto-05,Suto-11,Lewin-22}. With a hard core this was done previously in~\cite{Radin-04,BelRadShl-10}. Should the potential $w_1$ diverge at the origin, the bound on $h_Q$ implies a lower bound on the smallest distance between the points as in~\cite[Lem.~8]{Lewin-22}. For the Lennard-Jones interaction such an estimate appeared before in~\cite{Blanc-04}.

\begin{proof}
From the stability condition, we have
$$H_\Omega(X)-\mu n_\Omega(X)\geq h(X)-(\kappa+\mu)n_\Omega(X)\geq \frac{A}2\sum_i n_i(X)^2-(C+\mu)n_\Omega(X)$$
with the notation of the proof of Proposition~\ref{prop:pointwise_H}. Recall that $C=\kappa+A/2$. In particular, we deduce that there is just no particle at all in $\Omega$ ($n_\Omega(X)=0$) when $\mu<-C$. We can therefore always assume that $\mu\geq-C$.

Next we write $X=x\cup Y$ with $x=X\cap Q$. If $q:=q_Y>0$, we write again $Y=y\cup z$. Using $H_\Omega(z)-\mu n_\Omega(z)\geq H_\Omega(X)-\mu n_\Omega(X)$, we obtain from~\eqref{eq:pointwise_bd_H}
\begin{align*}
0&\geq \frac{h(x)}2 +\frac{A}{16}\psi_q V_q+\frac{A}{16V_q}k^2-(C+\mu) n-\mu k\\
&\geq \frac{h(x)}2 +\left(\frac{A}{16}\psi_q-\frac{4\mu_+^2}{A}\right) V_q-(C+\mu) n\geq \frac{h(x)}2 -\frac{4\mu_+^2L_\mu^d}{A}-(C+\mu) n,
\end{align*}
from the definition of $L_\mu$. We obtain
\begin{equation}
h_{Q_0}(X)=h(x)\leq 2(C+\mu) n+\frac{8L_\mu^d\mu_+^2}{A}.
\label{eq:estim_h_n_local}
\end{equation}
Using now
\begin{equation}
h(x)\geq \frac{A}2\sum_{C_i\subset Q_0}n_i^2-\frac{An}{2}\geq \frac{A}{2V_0}n^2-\frac{An}{2} 
 \label{eq:estim_h_0}
\end{equation}
and $A\leq 2C$ we find that the number of points in any cube of volume $V_0$ satisfies
$$n\leq \frac{4(2C+\mu)_+}{A}V_0+\frac{4\sqrt{V_0}L_\mu^{\frac d2}\mu_+}{A}\leq \frac{4(2C+\mu)_+V_0}{A} \myp[\big]{ 1+(L_\mu/L)^{\frac d2} }.$$
From~\eqref{eq:estim_h_n_local}, we then obtain for the energy
\begin{align*}
h(x)&\leq 2(C+\mu)n+\frac{8L_\mu^d\mu_+^2}{A}\\
&\leq \frac{8(2C+\mu)^2_+V_0}{A} \myp[\big]{ 1+(L_\mu/L)^{\frac d2}}+\frac{8L_\mu^d(2C+\mu)_+^2}{A}\\
&\leq \frac{8(2C+\mu)^2_+V_0}{A} \myp[\big]{ 1+(L_\mu/L)^{\frac d2}+(L_\mu/L)^{d}}\\
&\leq \frac{12(2C+\mu)^2_+V_0}{A} \myp[\big]{ 1+(L_\mu/L)^{d}}.
\end{align*}
The argument is similar when $q_Y=0$. From~\eqref{eq:pointwise_bd_H_q0}, we have in this case
$$\frac{A}{2V_0}n^2-\frac{An}{2}\leq h(x) \leq \frac{2^{2d-3}A}{\alpha}\psi_0V_0+2(C+\mu) n$$
and therefore obtain
$$n\leq \frac{4(2C+\mu)_+}{A}V_0+\frac{2^{d-1}\sqrt{\psi_0}}{\sqrt\alpha}V_0\leq \frac{4(2C+\mu)_+}{A}V_0\left(1+\frac{2^{d-2}\sqrt{\psi_0}}{\sqrt\alpha}\right)$$
since $1\leq (2C+\mu)/C\leq (2/A)(2C+\mu)_+$ for $\mu\geq-C$. The estimate on $h(x)$ is similar.
\end{proof}

\medskip

\noindent\textbf{Step 3: Exponential estimate on the Gibbs state.}
In Proposition~\ref{prop:pointwise_H} we have derived a simple pointwise bound which depends on the location of the particles outside of $Q_0$, and in particular on the value of $q=q_Y$. We now turn to the derivation of the local bound on the grand-canonical Gibbs state.
We have
\begin{align*}
\MoveEqLeft[2] Z_{\beta,\mu}(\Omega)\expec[\big]{e^{\frac\beta2 h_{Q_0}}}_{\beta,\mu,\Omega}\\
={}&\sum_{n,K}\frac{e^{\beta\mu(n+K)}}{n!\,K!}\int_{(Q_0\cap \Omega)^n}\int_{(\Omega\setminus Q_0)^K}e^{-\beta H(x\cup Y)+\frac\beta2 h(x)}\,\rd x\,\rd Y\\
={}&\sum_{q\geq0}\sum_{n,K}\frac{e^{\beta\mu(n+K)}}{n!\,K!}\int_{(Q_0\cap \Omega)^n}\int_{(\Omega\setminus Q_0)^K}\1(q_Y=q)e^{-\beta H(x\cup Y)+\frac\beta2 h(x)}\,\rd x\,\rd Y\\
={}&\sum_{q\geq0}\sum_{n,k,m}\frac{e^{\beta\mu(n+k+m)}}{n!\,k!\,m!}\int_{(Q_0\cap \Omega)^n}\iint_{(\Omega\cap Q_q\setminus Q_0)^k\times (\Omega\setminus Q_q)^m} \1\left(q_{y\cup z}=q\right)\\
&\qquad\qquad\times e^{-\beta H(x\cup y\cup z)+\frac\beta2 h(x)}\rd x\,\rd y\,\rd z.
\end{align*}
Using the pointwise bound~\eqref{eq:pointwise_bd_H} and then simply removing the condition that $q_{y\cup z}=q$, we can bound
\begin{align}
&\expec[\big]{e^{\frac\beta2 h_{Q_0}}}_{\beta,\mu,\Omega}\nn\\
&\leq{} e^{\beta \frac{4^{d-2}A}{\alpha}\psi_0V_0}\sum_{n}\frac{e^{\beta(\mu+C)n}(V_0)^n}{n!}+\sum_{q>0}\sum_{n,k}\frac{e^{\beta(\mu+C)n+\beta \mu k}(V_0)^n(V_q)^ke^{-\frac{\beta A}{16}V_q\psi_q}}{n!\,k!}\nn\\
&={} \exp \myp[\big]{ e^{\beta(\mu+C)}V_0 } \myp[\bigg]{ e^{\beta \frac{4^{d-2}A}{\alpha}\psi_0V_0}+\sum_{q>0}e^{-\beta\frac{A}{16}\psi_qV_q+e^{\beta\mu}V_q} } \nn\\
&\leq{} \exp \myp[\big]{ e^{\beta(\mu+C)}V_0 } \myp[\bigg]{ e^{\beta \frac{4^{d-2}A}{\alpha}\psi_0V_0}+\sum_{\ell>\ell_0}e^{-2^d\ell^d(\beta\frac{A}{16}\psi(\ell)-e^{\beta\mu})}}.
\label{eq:estim_exp}
\end{align}
After scaling and replacing all the constants by their maximum $\zeta$, this can be written as in the statement. This concludes the proof of Theorem~\ref{thm:Ruelle}.\qed

\begin{remark}[Pointwise bound on the correlation functions]\label{rmk:correlation_bounds}
Arguing as before we can estimate the correlation functions uniformly over $Q_0$. For any $x\subset (Q_0)^j$ we split the other particles according to the three domains $Q_0$, $Q_q$ and $\Omega\setminus Q_q$. We obtain
\begin{align*}
\rho^{(j)}(x)=&Z_{\beta,\mu}(\Omega)^{-1}\sum_{q\geq0}\sum_{n,k,m}\frac{e^{\beta\mu(j+n+k+m)}}{n!\,k!\,m!}\int_{(Q_0\cap \Omega)^n}\int_{(\Omega\cap Q_q\setminus Q_0)^k}\int_{(\Omega\setminus Q_q)^m}\times\\
&\qquad \times \1\left(q_{y\cup z}=q\right)e^{-\beta H(x\cup x'\cup y\cup z)}\rd x'\,\rd y\,\rd z.
\end{align*}
We then merge $x$ and $x'$ and use the lower bound~\eqref{eq:pointwise_bd_H} from Proposition~\ref{prop:pointwise_H} in the form
$$H(x\cup x'\cup y\cup z)\geq \frac{h(x)}{2}-C(n+j)+\frac{A}{16}\psi_q V_q+H(z).$$
We obtain with $z=e^{\beta\mu}$
\begin{multline}
\rho^{(j)}(x)\leq (e^{\beta C}z)^j e^{-\beta\frac{h(x)}{2}}e^{V_0 ze^{\beta C}} \myp[\bigg]{ e^{\beta \frac{4^{d-2}A}{\alpha}\psi_0V_0}+\sum_{\ell>\ell_0}e^{-2^d\ell^d(\beta\frac{A}{16}\psi(\ell)-z)} } \\
\forall x\subset (Q_0)^j.
\label{eq:estim_correlation_fn}
\end{multline}
This pointwise bound requires that all the $x_i$ are in $Q_0$. To get a true uniform bound on $\rho^{(j)}$ we need a modification of Proposition~\ref{prop:pointwise_H} dealing with the situation that some of the particles are far away from the others. We do not discuss this here and refer instead to~\cite{Ruelle-70}.
\end{remark}

\subsection{Proof of Corollary~\ref{cor:Ruelle_polynomial}}\label{sec:proof_cor_Ruelle}

When $\phi=B(1+|x|)^{-s}$ we can take $\psi(r)=r^\eps$ for $0<\eps<\min(1,s-d)$ and thus obtain
$$\expec[\big]{e^{\frac\beta2 h_{Q_0}}}_{\beta,\mu,\Omega}\leq \exp \myp[\big]{e^{\beta(\mu+C)}2^d\ell_0^d } \myp[\bigg]{ e^{\beta \frac{2^{3d-4}A}{\alpha}\ell_0^{d+\eps}}+\sum_{\ell>\ell_0}e^{-2^d\ell^d(\beta\frac{A}{16}\ell^\eps-e^{\beta\mu})} }.$$
The first estimate of Corollary~\ref{cor:Ruelle_polynomial} thus follows from estimating the last sum. A non-optimal but simple bound is provided in the following

\begin{lemma}\label{lem:estim_sum}
For $a,b>0$ and $2\leq \ell_0\in\N$, we have
\begin{equation}
\sum_{\ell>\ell_0}e^{-b\ell^{d+\eps}+a\ell^d}
\leq \exp \myp[\big]{Ba^{\frac{d+\eps}\eps}b^{-\frac{d}{\eps}}+\log(b/B)_- }
\label{eq:estim_sum}
\end{equation}
for a constant $B$ depending only on $d$ and $\eps$.
\end{lemma}

\begin{proof}
If $\ell_0^\eps\geq 2a/b$, then we can simply bound
\begin{align*}
\sum_{\ell>\ell_0}e^{-b\ell^{d+\eps}+a\ell^d}\leq \sum_{\ell>\ell_0}e^{-\frac{b}2\ell^{d+\eps}}&\leq \int_{\ell_0}^\ii e^{-\frac{b}2t^{d+\eps}}\,\rd t\\
&= \myp[\Big]{\frac{b}2 }^{-\frac1{d+\eps}}\int_{(b/2)^{\frac1{d+\eps}}\ell_0}^\ii e^{-t^{d+\eps}}\,\rd t\leq C\frac{e^{-\frac{b\ell_0^{d+\eps}}4}}{b^{\frac1{d+\eps}}}.
\end{align*}
If $1<\ell_0^\eps< 2a/b=:\ell_1^\eps$ we split the sum in two parts. We use that the number of integers $\ell$ satisfying $\ell_0<\ell\leq\ell_1$ is less than or equal to $\ell_1+1-\ell_0\leq\ell_1$ since $\ell_0\geq2$. Hence, using the previous estimate and the fact that $xe^x\leq e^{2x}$ with $x=a\ell_1^d$, we find
\begin{align*}
\sum_{\ell>\ell_0}e^{-b\ell^{d+\eps}+a\ell^d}&\leq \ell_1e^{a\ell_1^d}+\sum_{\ell>\ell_1}e^{-b\ell^{d+\eps}+a\ell^d}\leq \frac{ e^{2a\ell_1^d}}{a^{\frac1d}}+C\frac{e^{-\frac{b\lfloor\ell_1\rfloor^{d+\eps}}4}}{b^{\frac1{d+\eps}}}\\
&\leq \frac{ e^{2a\ell_1^d}}{2(b/B)^{\frac1d}}+\frac{e^{-\frac{b\ell_0^{d+\eps}}4}}{2(b/B)^{\frac1{d+\eps}}}\leq \frac{e^{\frac1d\log(b/B)_-}}2\left(e^{2a\ell_1^d}+e^{-\frac{b\ell_0^{d+\eps}}4}\right),
\end{align*}
with $B=\max(2^{1+d},(2C)^{d+\eps})$. Since $\ell_1>\ell_0$, we have replaced $\lfloor\ell_1\rfloor$ by $\ell_0$ in the the last exponential. To simplify the bound further, we use that $e^{-b\ell_0^{d+\eps}/4}\leq1$ and $(e^x+1)/2\leq e^{x}$ for $x\geq0$, leading to~\eqref{eq:estim_sum} after further increasing $B$ and using $d\geq1$.
\end{proof}

Using~\eqref{eq:estim_sum} and $e^x+e^y\leq e^{x+y+1}$ for $x,y\geq0$, we find
$$\expec[\big]{e^{\frac\beta2 h_{Q_0}}}_{\beta,\mu,\Omega}\leq e^{\zeta \beta\cR},\qquad \cR=\frac{ze^{\beta C}V_0}{\beta\zeta}+\ell_0^{d+\eps}+ \myp[\Big]{ \frac{z}\beta }^{\frac{d+\eps}\eps} + \frac{(\log\beta)_-+1}{\beta}$$
for a large enough constant $\zeta$. By Jensen's inequality, we obtain
$$\pscal{h_{Q_0}}_{\beta,\mu,\Omega}\leq 2\zeta \cR.$$
Using~\eqref{eq:estim_h_0} and the fact that $ \cR\geq  \ell_0^{d+\eps}\geq  V_0/2^d$, we obtain after further increasing $\zeta$
$$\pscal{n^2_{Q_0}}_{\beta,\mu,\Omega}\leq \zeta \cR V_0.$$
On the other hand, our bound~\eqref{eq:estim_correlation_fn} on the correlation functions provides
\begin{align*}
\pscal{n_{Q_0}}_{\beta,\mu,\Omega}&=\int_{Q_0}\rho^{(1)}\leq (e^{\beta C}z V_0)e^{\zeta \beta\cR},\\
\pscal{n_{Q_0}(n_{Q_0}-1)}_{\beta,\mu,\Omega}&=\iint_{(Q_0)^2}\rho^{(2)}\leq (e^{\beta C}z V_0)^2e^{\zeta \beta\cR}.
\end{align*}
Thus we arrive at~\eqref{eq:main_estim_Ruelle_polynomial} after using $1+e^{\beta C}z V_0\leq e^{e^{\beta C}z V_0}\leq e^{\zeta\beta\cR}$. This concludes the proof of Corollary~\ref{cor:Ruelle_polynomial}.\qed

 \subsection{Ruelle bounds in an external potential}
We can extend the previous result to the situation where we have a bounded-below external potential $v(x)$, that is, the energy is of the form
$$H_v(x)=\sum_{1\leq j<k\leq n}w(x_j-x_k)+\sum_{i=1}^n v(x_i)=:H(x)+V(x)$$
with $\int_{\R^d}e^{-\beta v(x)}\rd x<\ii$. The previous situation corresponds to $v\equiv -\mu$ on $\Omega$ and $v\equiv+\ii$ outside of $\Omega$.

\begin{corollary}[Local estimates with potential]\label{thm:Ruelle_V}
Assume that $w=w_1+w_2$ satisfies the same assumptions as in Theorem~\ref{thm:Ruelle}. Let $v:\R^d\to\R\cup \{+\ii\}$ be any measurable function such that
$$v\geq -\mu_0\text{ a.e.},\qquad \int_{\R^d}e^{-\beta v(x)}\,dx<\ii.$$
The corresponding grand-canonical Gibbs state satisfies
\begin{equation}
\expec[\big]{e^{\frac\beta2 h_{Q}}}_{\beta,v} 
\leq \exp \myp[\Big]{ e^{\zeta \beta}\int_{Q}e^{-\beta v(x)}\,\rd x } \myp[\bigg]{ e^{\zeta\beta\psi(L)L^d}+\sum_{\substack{\ell\in\N\\ \ell> L}}e^{-\beta\big(\frac{\psi(\ell)}{\zeta}-\zeta e^{\beta\mu_0}\big)\ell^d} }
\label{eq:main_estim_Ruelle_thm_V}
\end{equation}
for any $\beta>0$ and any cube $Q$ of side length $L\geq \zeta$, where $\psi$ and $\zeta$ are the same as in Theorem~\ref{thm:Ruelle}. If $\phi(|x|)=B(1+|x|)^{-s}$ for some $s>d$ and $0<\eps<\min(1,s-d)$, then we have
\begin{equation}
\expec[\big]{e^{\frac\beta2 h_{Q}}}_{\beta,v}
\leq \exp \myp[\Big]{ e^{\zeta \beta}\int_{Q}e^{-\beta v}+\zeta\beta\cR_0 },
\label{eq:main_estim_Ruelle_polynomial_V}
\end{equation}
\begin{multline}
 \pscal{n_{Q}^2}_{\beta,v}\leq L^d\,\min\bigg\{e^{\zeta\beta(\cR_0+1)} \myp[\big]{ 1+L^de^{\beta(\zeta+\mu_0)} } \int_{Q}e^{-\beta v};\\
 \zeta\beta^{-1}e^{\zeta \beta}\int_{Q}e^{-\beta v}+\zeta\cR_0\bigg\},
 \label{eq:main_estim_Ruelle_polynomial_V_n2}
\end{multline}
where $\cR_0:=L^{d+\eps}+ (e^{\beta\mu_0}/\beta)^{1+\frac{d}\eps}+\frac{1+(\log\beta)_-}\beta.$
\end{corollary}

\begin{proof}
The potential $v$ only occurs in Step 3 of the proof of Theorem~\ref{thm:Ruelle}. We have
\begin{multline*}
Z_{\beta,v}(\Omega)\expec[\big]{e^{\frac\beta2 h_{Q_0}}}_{\beta,v}
=\sum_{q\geq0}\sum_{n,k,m}\frac{1}{n!\,k!\,m!}\int_{(Q_0)^n}\iint_{( Q_q\setminus Q_0)^k\times (\R^d\setminus Q_q)^m} \times\\
\times \1\left(q_{y\cup z}=q\right)e^{-\beta \big(H(x\cup y\cup z)+V(x)+V(y)+V(z)\big)+\frac\beta2 h(x)}\rd x\,\rd y\,\rd z.
\end{multline*}
We use $V(y)\geq -\mu_0 k$ and the previous bound on $H(x\cup y\cup z)$, leading to
\begin{multline}
\expec[\big]{e^{\frac\beta2 h_{Q_0}}}_{\beta,\mu,\Omega}
\leq e^{\beta \frac{4^{d-2}A}{\alpha}\psi_0V_0}\sum_{n}\frac{e^{\beta Cn} \myp[\big]{\int_{Q_0}e^{-\beta v} }^n}{n!}\nn\\
+\sum_{q>0}\sum_{n,k}\frac{e^{\beta(Cn+\mu_0 k)} \myp[\big]{\int_{Q_0}e^{-\beta v} }^n(V_q)^ke^{-\frac{\beta A}{16}V_q\psi_q}}{n!\,k!}
\end{multline}
and~\eqref{eq:main_estim_Ruelle_thm_V} follows as for~\eqref{eq:estim_exp}. When $\phi$ decreases polynomially, the bound~\eqref{eq:main_estim_Ruelle_polynomial_V} is proved the same as Corollary~\ref{cor:Ruelle_polynomial}. Finally, we have the pointwise bound on the correlation function
$$\rho^{(j)}(x)\leq e^{-\beta \myp[\big]{ \frac{h(x)}{2}-Cj+\sum_{i=1}^jv(x_i)}} \myp[\bigg]{ e^{\beta \frac{4^{d-2}A}{\alpha}\psi_0V_0}+\sum_{\ell>\ell_0}e^{-2^d\ell^d(\beta\frac{A}{16}\psi(\ell)-z)} }$$
for $x\subset (Q_0)^j$ and~\eqref{eq:main_estim_Ruelle_polynomial_V_n2} follows as in the proof of Corollary~\ref{cor:Ruelle_polynomial}.
\end{proof}

We have a similar result at zero temperature.

\begin{corollary}[Local estimates with potential, $T=0$]\label{thm:Ruelle_V_T0}
Assume that $w=w_1+w_2$ satisfies the same assumptions as in Theorem~\ref{thm:Ruelle}. Let $v:\R^d\to\R\cup \{+\ii\}$ be any lower semi-continuous function such that $v\geq -\mu_0$. Let $X\subset\R^d$ be any minimizer for the grand-canonical free energy
$$X\mapsto H(X)+V(X),\qquad V(X)=\sum_{x\in X}v(x).$$
Then the number of points in any cube $Q$ of side length $L\geq\zeta$ satisfies
\begin{multline}
n_Q(X)\leq \frac{4L^d}{A} \myp[\Big]{ 2\kappa+A-\min_Q v }_+\\
+ \myp[\bigg]{ \frac{16V_0(\mu_0)_+^2L_{\mu_0}^d}{A^2}+\frac{2^{2d-2}}{\alpha}\psi(L/2)L^{2d}-\frac{4L^d}{A} \myp[\Big]{ 2\kappa+A-\min_Q v }_- }^{\frac12}_+,
\label{eq:main_estim_Ruelle_thm_V_T0}
\end{multline}
where $L_{\mu_0}$ is, like in Corollary~\ref{cor:T0}, the smallest integer such that $\psi(L_{\mu_0}/2)\geq 64(\mu_0)_+^2/A^2$.
\end{corollary}

The bound~\eqref{eq:main_estim_Ruelle_thm_V_T0} states that the number of particles in a cube can be bounded in terms of the ``local'' chemical potential $\mu=-\min_Q v$. The bound is linear for $\mu\geq -(2\kappa+A)$ and vanishes when $\mu$ is really negative. In between it interpolates continuously. As usual we should not think of $L$ large here. Taking $L=\zeta$ and covering any large enough cube by finitely many smaller cubes, we get a bound in the form
\begin{equation}
 \boxed{n_Q\leq C|Q| \myp[\Big]{ C-\min_Q(v) }_+}
 \label{eq:simplified_bound_v_T0}
\end{equation}
where $C$ depends on $\mu_0$.

\begin{proof}
We introduce $\mu:=-\min_Q v$ and write $X=x\cup Y$ with $x=X\cap Q$. If $q:=q_Y>0$, we write again $Y=y\cup z$. For the particles in $x\subset Q$ we use $v\geq -\mu$ but for the particles in $y$ we use $v\geq -\mu_0$. Arguing exactly as in the proof of Corollary~\ref{cor:T0}, we obtain
\begin{equation}
\frac{A}{V_0}n^2-4(2C+\mu)n\leq \frac{16(\mu_0)_+^2L_{\mu_0}^d}{A}.
\label{eq:estim_v_not_cnst}
\end{equation}
If $q_Y=0$ we get the same but with the right side replaced by $\frac{2^{2d-2}A}{\alpha}\psi_0V_0$. To conclude, we use that $n^2-\alpha n\leq \beta$ for an integer $n$ and $\beta>0$ implies $n\leq\alpha_++(\beta-\alpha_-)^{1/2}_+$.
\end{proof}


\end{document}